\documentclass[11pt,journal,onecolumn,a4paper]{IEEEtran}

\usepackage{graphicx}
\usepackage{epstopdf}
\DeclareGraphicsExtensions{.eps,.pdf} 
\graphicspath{ {figures/} }
\usepackage{cite}
\usepackage{subfigure}
\usepackage{amssymb,amsmath}
\usepackage{algorithm}
\usepackage{algorithmic}
\usepackage{color}
\usepackage{multirow}
\usepackage{multicol}
\usepackage{cases}
\usepackage{setspace}
\usepackage{etoolbox}

\let\labelindent\relax
\usepackage{enumitem}

\usepackage{booktabs}
\usepackage{longtable}
\usepackage{array}
\newcolumntype{P}[1]{>{\endgraf\vspace*{-\baselineskip}}p{#1}}

\makeatletter
\newcommand\restartchapters{\par
  \setcounter{chapter}{0}%
  \setcounter{section}{0}%
  \gdef\@chapapp{\chaptername}%
  \gdef\thechapter{\@arabic\c@chapter}}
\makeatother

\newtheorem{theorem}{\bf {Theorem}}

\newtheorem{remark}{{\bf{Remark}}}
\newtheorem{definition}{\bf {Definition}}
\newtheorem{lemma}{\bf {Lemma}}
\newtheorem{procedure}{\bf {Procedure}}

\newcommand{\tr}{{\mathtt{Tr}}}
\newcommand{\diag}{{\mathtt{diag}}}

\newcommand{\st}{{\mathrm{s.t.}}}
\newcommand{\thmend}{\hspace*{\fill}~\QEDopen\par\endtrivlist\unskip}

\newcommand{\AP}{\mathtt{AP}_{m}}

\newcommand{\RSI}{\mathtt{RSI}}
\newcommand{\APx}[1]{\mathtt{AP}_{#1}}
\newcommand{\ULU}{\mathtt{U}_{\ell}^{\mathtt{u}}}
\newcommand{\ul}{\mathtt{u}}
\newcommand{\SI}{\mathtt{SI}}
\newcommand{\AtoA}{\mathtt{AA}}
\newcommand{\dl}{\mathtt{d}}
\newcommand{\DLU}{\mathtt{U}_{k}^{\mathtt{d}}}
\newcommand{\DLUi}[1]{\mathtt{U}_{#1}^{\mathtt{d}}}

\newcommand{\cM}{\mathcal{M}}
\newcommand{\cK}{\mathcal{K}}
\newcommand{\cL}{\mathcal{L}}
\newcommand{\tZF}{\mathtt{ZF}}
\newcommand{\tIZF}{\mathtt{IZF}}
\newcommand\argmax{\operatornamewithlimits{argmax}}

\let\mybibitem\bibitem
\renewcommand{\bibitem}[1]{%
  \ifstrequal{#1}{ZhangMAT5G2019}
    {\color{black}\mybibitem{#1}}
    {\color{black}\mybibitem{#1}}%
}

\newcommand*{\hili}{\color{black}}

\newcommand*{\majrev}{\color{black}}

\usepackage{capt-of}% or \usepackage{caption}
\usepackage{dblfloatfix}
\usepackage{varwidth}
\makeatletter
\g@addto@macro\normalsize{%
 \setlength\abovedisplayskip{1.1pt}
 \setlength\belowdisplayskip{1.1pt}
 \setlength\abovedisplayshortskip{1.1pt}
 \setlength\belowdisplayshortskip{1.1pt}
}

\newcommand{\subparagraph}{}

\usepackage{titlesec}
\titlespacing{\section}{0pt}{2pt}{0pt}

\begin{document}
\bstctlcite{IEEEexample:BSTcontrol}

\title{{\hili On the Spectral and Energy Efficiencies of Full-Duplex Cell-Free Massive MIMO}}
\author{
	\IEEEauthorblockN{Hieu V. Nguyen, Van-Dinh Nguyen, Octavia~A.~Dobre, Shree Krishna Sharma, \\ Symeon Chatzinotas, Bj$\ddot{\text{o}}$rn Ottersten,  and Oh-Soon Shin \vspace{-20pt}}
	\\
	\thanks{H. V. Nguyen and O.-S. Shin are with the School of Electronic Engineering \& Department of ICMC Convergence Technology, Soongsil University, Seoul 06978, South Korea (e-mail: \{hieuvnguyen, osshin\}@ssu.ac.kr).}
	\thanks{V.-D. Nguyen is with the Interdisciplinary Centre for Security, Reliability and Trust (SnT) – University of Luxembourg, L-1855 Luxembourg. He was with the Department of ICMC Convergence Technology, Soongsil University, Seoul 06978, South Korea  (email: dinh.nguyen@uni.lu). }	
	\thanks{O.~A.~Dobre is with the Faculty of Engineering and Applied Science, Memorial University, St. John's, NL A1X3C5, Canada  (e-mail: odobre@mun.ca).}
	\thanks{S. K. Sharma, S. Chatzinotas, and B. Ottersten are with the Interdisciplinary Centre for Security, Reliability and Trust (SnT) – University of Luxembourg, L-1855 Luxembourg (e-mail: \{shree.sharma, symeon.chatzinotas, bjorn.ottersten\} @uni.lu)}
	}
%	\thanks{This work was supported }

\maketitle
%\thispagestyle{empty}
%\pagestyle{empty}
%\vspace*{-40pt}
\begin{abstract}
In-band full-duplex (FD) operation {\hili is practically more suited for short-range communications such as WiFi and small-cell networks, due to its current practical limitations on the self-interference cancellation}. In addition,   cell-free massive multiple-input multiple-output (CF-mMIMO) is a new and scalable version of MIMO networks, which is designed to bring  service antennas closer to  end  user equipments (UEs). To achieve higher  spectral and energy efficiencies (SE-EE) of a wireless network, it is of practical interest  to incorporate FD capability into CF-mMIMO systems to utilize their combined benefits. 
We formulate a novel and comprehensive optimization problem for  the maximization of SE and EE in which  power control, access point-UE (AP-UE) association and AP selection  are jointly optimized under a realistic power consumption model, resulting in a difficult class of mixed-integer nonconvex programming. To tackle the binary nature of the formulated problem, we propose  an efficient approach by exploiting a strong coupling between binary and continuous variables, leading to a  more tractable  problem. In this regard, two low-complexity transmission designs based on zero-forcing (ZF) are proposed. Combining tools from  inner approximation framework and  Dinkelbach method, we develop  simple iterative  algorithms with polynomial computational complexity in each iteration and  strong theoretical performance guaranteed.  
Furthermore, towards a robust design for FD CF-mMIMO, a novel heap-based pilot assignment algorithm is proposed to mitigate effects of pilot contamination.
Numerical results show that our proposed designs with realistic parameters significantly outperform the well-known approaches (i.e.,  small-cell and  collocated mMIMO) in terms of the SE and EE. Notably, the proposed ZF designs require much less  execution time than  the simple maximum ratio transmission/combining.
\end{abstract}
\begin{IEEEkeywords}
Cell-free massive multiple-input multiple-output,  energy efficiency, full-duplex radio,  inner approximation, spectral efficiency, successive interference cancellation. 
\end{IEEEkeywords}

\newpage

\section{Introduction} \label{Introduction} 
Peak data rates in the order of tens of Gbits/s, massive connectivity and seamless area
coverage requirement along with different use cases  are expected in beyond 5G networks \cite{Cisco2017_CVNI,Osseiran:Cambridge:2016,Yadav:IEEEWirelessComm:Aug2018}.   Multiple-antenna technologies, which offer extra degrees-of-freedom (DoF) have been a key element to provide huge spectral efficiency gains of modern wireless communication systems \cite{ZhangMAT5G2019,Chatzinotas:IEEETWC:July2009}. However, multiple-antenna systems based on  half-duplex (HD) radio will apparently reach their capacity limits in near future due to under-utilization of radio resources. In addition, the use of multiple-antenna also causes a serious concern over  the global climate and tremendous electrical costs due to the number of associated radio frequency (RF) elements \cite{BuzziJSAC16}. Consequently, spectral efficiency (SE) and energy efficiency (EE) will certainly be considered as major figure-of-merit in the design of beyond 5G networks.

In-band full-duplex (FD) has been envisaged as a key enabling technology to improve the SE of traditional wireless communication systems \cite{Yadav:IEEEVTM:June2018,Sharma:IEEECST:2018}. By enabling downlink (DL) and uplink (UL) transmissions on the same time-frequency resource, FD radios  are expected to increase the SE of a wireless link over its HD counterparts by a factor close to two \cite{Sabharwal:JSAC:Feb2014,GoyalCMag15}. The main barrier in implementing FD  is the self-interference (SI) that leaks from the transmitter to its own receiver on the same device. Fortunately, recent advances in  active and passive SI suppression techniques have been successful to  bring the SI power at the background noise level \cite{Bharadia13,Bharadia14},  thereby making FD a realistic technology for modern wireless systems. However,  there still exists a small, but not negligible, amount of  SI due to  imperfect SI suppression, referred to as  \textit{residual} SI (RSI). As a result,  FD-enabled base station (BS) systems  have been widely studied in small-cell (SC) cellular networks \cite{Yadav:Access,Dan:TWC:14,Dinh:Access,Hieu:IEEETWC:June2019,Hieu:IEEETCOM:June2019, Dinh:JSAC:18,Dinh:TCOMM:2017,Tam:TCOM:16,Aquilina:TCOMM:2017}, {\hili where the residual SI  can be further handled by  power control algorithms.}

{\majrev Recently, a new concept of multiple-input multiple-output (MIMO) networks and distributed antenna systems (DAS), called cell-free massive MIMO (CF-mMIMO), has been proposed to overcome the inter-cell interference,
as well as to provide handover-free and balanced quality-of-experience (QoE) services for cell-edge users (UEs) \cite{ZhangMAT5G2019,Ngo:TWC:Mar2017,Nayebi:IEEETWC:Jul2017,InterdonatoGLOBECOM16,Bashar:IEEETWC:Apr2019,Nguyen:IEEELCOM:Aug2017,Ngo:IEEETGCN:Mar2018}. In CF-mMIMO, a very large number of access point (AP) antennas are distributed over a wide area to coherently serve numerous UEs in the same resources; this inherits key characteristics of collocated massive MIMO  (Co-mMIMO) networks, such as  favorable propagation and channel hardening \cite{Marzetta:IEEETWC:Nov2010, Marzetta:Cambridge:2016}. CF-mMIMO  has significantly better performances in terms of SE and EE, compared to small-cell   \cite{Ngo:TWC:Mar2017} and Co-mMIMO networks \cite{Ngo:IEEETGCN:Mar2018}, respectively. It can be easily foreseen that performance gains of CF-mMIMO come from  the joint processing  of a large number of distributed APs at a  central processing unit (CPU). The CPU is essentially the same as the edge-cloud computing in  cloud radio access networks (C-RANs). Thus, C-RAN can be viewed as an enabler of CF-mMIMO \cite{ZhangMAT5G2019}.
 }

\subsection{Motivation}
From the aforementioned reasons, it is not too far-fetched to envisage a wireless system
employing the FD technology in  CF-mMIMO, called FD CF-mMIMO. {\hili It is expected to reap all key advantages of FD and CF-mMIMO, towards enhancing the SE and EE performances of future wireless networks}. More importantly, FD CF-mMIMO can be considered as a practical and promising technology for beyond 5G networks since low-power and low-cost FD-enabled APs are well suited for short-range transmissions between APs and UEs. {\majrev Despite the clear benefits of these two technologies, FD CF-mMIMO poses the following obvious challenges on radio resource allocation problems: $(i)$ Residual SI still remains  a challenging task in the design of FD CF-mMIMO, having a negative impact on its potential performance gains; $(ii)$  A large number of APs and legacy UEs result in stronger inter-AP interference (IAI) and co-channel interference (CCI, caused by the UL transmission to DL UEs), compared to traditional FD cellular networks \cite{Dan:TWC:14,Tam:TCOM:16,Dinh:TCOMM:2017,Aquilina:TCOMM:2017,Dinh:Access,Yadav:Access, Dinh:JSAC:18,Hieu:IEEETWC:June2019,Hieu:IEEETCOM:June2019}; $(iii)$ FD CF-mMIMO  increases the network power consumption due to  additional number of APs.  It has been noted that low power APs consume about 30\% of the total power consumption of a mobile network operator  \cite{AuerWC11}. 
These motivate us to investigate a joint design of precoder/receiver, AP-UE association and AP selection along with an efficient transmission
strategy to attain the optimal SE and EE performances of FD CF-mMIMO systems.}

\subsection{Review of Related Literature} 
FD small-cell systems have been investigated in many prior works. For example, the authors in \cite{Dan:TWC:14} studied a single-cell network with the aim of maximizing the SE under the assumption of perfect channel state information (CSI). This work was generalized in \cite{Dinh:Access} where  user grouping and time allocation were jointly designed. To accelerate the use of FD operation, \cite{Hieu:IEEETWC:June2019} proposed a  half-array antenna mode selection to mitigate the effect of residual SI and CCI by serving UEs in two separate phases. This design is capable of enabling hybrid modes of  HD and FD  to utilize a full-array antenna at the BS. The application of FD to emerging subjects has also been investigated,  including  FD non-orthogonal multiple access \cite{Hieu:IEEETCOM:June2019},   FD physical layer security \cite{Dinh:JSAC:18} and FD wireless-powered MIMO  \cite{ChaliseTCOM17,Dinh:TCOMM:2017}. 
The SE maximization for FD  multi-cell networks  was  considered in \cite{Tam:TCOM:16} and with the worst-case robust design in \cite{Aquilina:TCOMM:2017}, where coordinated multi-point transmission was adopted. It is widely believed that this interference-limited technique can no longer provide a high edge throughput and requires a large amount
of backhaul signaling to be shared among BSs. In addition, the common transmission design used in these works is linear beamforming for DL and  minimum mean square error and successive interference cancellation (MMSE-SIC) receiver for UL. Although such a design can provide a very good performance, it is only suitable for networks of small-to-medium sizes.

CF-mMIMO  has recently received considerable attention.  In particular, the work in \cite{Ngo:TWC:Mar2017} first derived closed-form expressions of DL and UL achievable rates which confirm the SE gain  of the CF-mMIMO over a small-cell system. Assuming mutually orthogonal pilot sequences assigned to  UEs,
\cite{Nayebi:IEEETWC:Jul2017} analyzed impacts of the power allocation for DL transmission using maximum ratio transmission (MRT) and zero-forcing (ZF). The results showed that the achievable per-user rates of CF-mMIMO can be substantially improved, compared to those of small-cell systems. To  further improve the network performance of CF-mMIMO, a beamformed DL training was proposed in \cite{InterdonatoGLOBECOM16}. The authors in \cite{Bashar:IEEETWC:Apr2019} examined the problem of maximizing the minimum signal-to-interference-plus-noise ratio (SINR) of  UL UEs subject to  power constraints. More recently,  the EE problem for DL CF-mMIMO was investigated  using ZF in \cite{Nguyen:IEEELCOM:Aug2017} and  MRT in \cite{Ngo:IEEETGCN:Mar2018}, by taking into account the effects of  power control, non-orthogonality of pilot sequences, channel estimation and  hardware power consumption. Two simple AP selection schemes were also proposed in \cite{Ngo:IEEETGCN:Mar2018} to reduce the backhaul power consumption. The conclusion in these works was that ZF and MRT beamforming methods can offer excellent performance when the number of APs is large.

Despite its potential, there is only a few attempts on characterizing the performance of FD CF-mMIMO in the literature. In this regard, the authors in \cite{Vu:ICC:May2019} analyzed the performance of FD CF-mMIMO with the channel estimation taken into account, where all APs operate in FD mode.  By a simple conjugate beamforming/matched filtering transmission design, it was shown that such a design requires a deep SI suppression to unveil the performance gains of the FD CF-mMIMO system over HD CF-mMIMO one. Tackling the imperfect CSI and spatial correlation was studied in \cite{Wang:arxiv:2019}, showing that  FD CF-mMIMO with a genetic algorithm-based user scheduling strategy is able to help alleviating the CCI and obtaining a significant SE improvement.  However, the effects of SI, IAI and CCI are not fully addressed  in the above-cited works, leading to the need of optimal solutions, which is the focus of this paper.

\subsection{Research Gap and Main Contributions}
Though in-depth results of multiple-antenna techniques were presented for HD \cite{ZhangMAT5G2019} and FD operations \cite{Yadav:Access,Dan:TWC:14,Dinh:Access,Hieu:IEEETWC:June2019,Hieu:IEEETCOM:June2019, Dinh:JSAC:18,Dinh:TCOMM:2017,Tam:TCOM:16,Aquilina:TCOMM:2017}, they are  not very practical for FD CF-mMIMO due to the large size of optimization variables.  In addition,  a direct application of CF-mMIMO in \cite{Ngo:TWC:Mar2017,Nayebi:IEEETWC:Jul2017,InterdonatoGLOBECOM16,Bashar:IEEETWC:Apr2019,Nguyen:IEEELCOM:Aug2017,Ngo:IEEETGCN:Mar2018} to FD CF-mMIMO systems would result in a poor performance since the additional interference (i.e., residual SI, IAI and CCI) is strongly involved in both DL and UL transmissions.  In FD CF-mMIMO systems, the number of APs is large, but still  finite, and thus the effects of SI, IAI and CCI are acute and unavoidable. These issues have not been fully addressed in  \cite{Vu:ICC:May2019,Wang:arxiv:2019}.
Meanwhile, the EE performance in terms of bits/Joule is considered as a key performance indicator for green communications \cite{AuerWC11}, which is even more important in FD CF-mMIMO systems due to a large number of deployed APs. To the authors' best knowledge, the  EE maximization, that takes into
account the effects of   imperfect CSI, power allocation, AP-DL UE association, AP selection,   and  load-dependent power consumption, has not been previously studied for FD CF-mMIMO.

In the above context, this paper considers an FD CF-mMIMO system under time-division duplex (TDD) operation, where  FD-enabled multiple-antenna APs  simultaneously serve many UL and DL UEs on the same time-frequency resources. The network SE and EE of the FD CF-mMIMO system are investigated under a relatively realistic power consumption model. Furthermore, we propose an efficient transmission design for FD CF-mMIMO to  resolve  practical restrictions given above. More precisely, power control and AP-DL UE association are jointly optimized to reduce network interference. Inspired by the ZF method \cite{SpencerTSP04}, two simple, but efficient transmission schemes are proposed. It is well-known that in massive MIMO networks,  ZF-based schemes achieve close performance to the MMSE ones but with much less complexity \cite{Marzetta:Cambridge:2016}. In contrast with \cite{Vu:ICC:May2019} and \cite{Wang:arxiv:2019}, AP selection is also taken into account to preserve the hardware power consumption.  We note that the EE objective function strikes the balance between the SE and total power consumption. {\majrev The main contributions of this paper are summarized as follows:
\begin{itemize}
	\item Aiming at the optimization of SE and EE, we introduce new binary variables to establish AP-DL UE association and AP selection. This design not only helps mitigate network interference (residual SI, IAI and CCI) but also saves power consumption (i.e., some APs can be switched off if necessary). In our system design, APs can be  automatically switched between FD and HD operations, which allows to exploit the full potential of FD CF-mMIMO. 
	\item We  formulate a generalized maximization problem for SE and EE by incorporating various aspects such as joint power control, AP-DL UE association and AP selection, which belongs to the difficult class of mixed-integer nonconvex optimization problem. To develop an unified approach to all  considered transmission schemes, we first provide fundamental insights into the structure of the optimal solutions of binary variables, and then transform the original problem into a more tractable form. 
\item	Two low-complexity transmission schemes for DL/UL are proposed, \textit{namely} ZF and improved ZF (IZF) by employing principal component analysis (PCA) for DL and SIC for UL. The latter based on the orthonormal basis (ONB) is capable of  canceling multiuser interference (MUI)  and further alleviating residual SI and IAI.
We then employ the combination of  inner convex approximation (ICA) framework \cite{Marks:78,Beck:JGO:10} and Dinkelbach method \cite{Dinkelbach67}   to develop iterative algorithms of low-computational complexity, which converge rapidly to the optimal solutions as well as require much lower execution time than  the simple maximum ratio transmission/combining (MRT/MRC).
\item Towards practical applications, we further consider a robust design under channel uncertainty, where the UL training is taken into account. To this end, we develop a novel heap-based pilot assignment algorithm to reduce both the pilot contamination and training complexity.
\item Extensive numerical results confirm  that the proposed algorithms greatly improve the SE and EE performance over the current state-of-the-art approaches, i.e., SC-MIMO and Co-mMIMO under both HD and FD operation modes. It also reveals the effectiveness of joint AP selections in terms of the achieved EE performance.
\end{itemize} 
}

\subsection{Paper Organization and Notation}
The remainder of this paper is organized as follows. The FD CF-mMIMO system model is introduced in Section \ref{SystemModelandProblemFormulation}. The formulation of the optimization problem and the  derivation of its tractable form are provided  in Section \ref{OptimizationProblemDesign}. Two ZF-based transmission designs along the proposed algorithm are presented in Sections \ref{sec:ZF} and \ref{sec:IZF}. The pilot assignment algorithm is presented in Section \ref{sec:Pilo Assignment}.  Numerical results are given in Section \ref{NumericalResults}, while Section \ref{Conclusion} concludes the paper.

\emph{Notation}: Bold lower and upper case letters denote vectors and matrices, respectively. $\mathbf{X}^{T}$ and $\mathbf{X}^{H}$ represent normal transpose and Hermitian transpose of $\mathbf{X}$, respectively. $\tr(\cdot)$, $\|\cdot\|$ and $|\cdot|$ are the trace,  Euclidean norm and  absolute value, respectively.   $ \mathbf{a} \preceq \mathbf{b}  $ stands for  the element-wise comparison of vectors.  $\mathtt{diag}(\mathbf{a})$ returns  the diagonal matrix
with the main diagonal constructed from elements of $\mathbf{a}$.   $\mathbb{C}$ and $\mathbb{R}$ denote the space of complex and real   matrices, respectively. 
Finally, $x\sim\mathcal{CN}(0,\sigma^2)$ and $x\sim\mathcal{N}(0,\sigma^2)$ represent circularly symmetric complex and  real-valued Gaussian random variable with zero mean and variance $\sigma^2$, respectively.

\section{System Model} \label{SystemModelandProblemFormulation}
\subsection{Transmission Model}
\begin{figure}[!h]
	\centering
	\vspace{-5pt}
	\includegraphics[width=0.4\columnwidth,trim={0cm 0.0cm 0cm 0.0cm}]{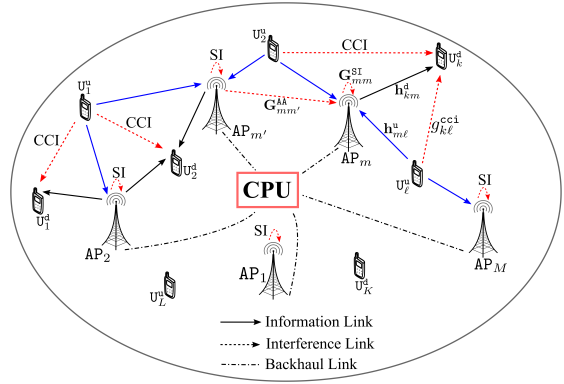}
		\vspace{-5pt}
	\caption{An illustration of the full-duplex cell-free massive MIMO system.}
	\label{fig: system model}
\end{figure}
An FD CF-mMIMO system operated in TDD mode is considered, where the set $\mathcal{M}\triangleq\{1,2,\cdots,M\}$ of $M=|\mathcal{M}|$ FD-enabled APs simultaneously serves the sets  $\mathcal{K}\triangleq\{1,2,\cdots,K\}$ of $K=|\mathcal{K}|$  DL UEs and $\mathcal{L}\triangleq\{1,2,\cdots,L\}$ of $L=|\mathcal{L}|$  UL UEs in the same time-frequency resource, as illustrated in Fig. \ref{fig: system model}. The total number of  APs' antennas is $N=\sum_{m\in\mathcal{M}}N_m $ where $N_m$ is the number of antennas at AP $m$, while each UE has a single-antenna.  We  assume that  APs, DL and UL UEs are randomly placed in a wide  area. All APs are equipped with  FD capability by circulator-based FD radio prototypes \cite{Bharadia13,Bharadia14}, which are connected to the CPU through   perfect backhaul links with sufficiently large capacities (i.e., high-speed optical ones) \cite{Ngo:TWC:Mar2017}. In this paper, we focus on slowly time-varying channels, and thus, conveying the CSI via the backhaul links occurs less frequently than data transmission. 

We assume that data transmission is performed within a coherence interval, which is similar to TDD operation in the context of massive MIMO \cite{Marzetta:Cambridge:2016}. Based on the joint processing at the CPU, the message sent by an UL UE is decoded by aggregating  the received signals from all active APs due to the UL broadcast transmission. In addition, APs are  geographically distributed in a large area, and thus, each DL UE should be served by a subset of active APs with good channel conditions \cite{Ngo:TWC:Mar2017,Ngo:IEEETGCN:Mar2018}. This is done by introducing new binary variables to establish AP-DL UE associations. Such a design offers the following two obvious advantages: $(i)$ improving the  SE for a given  system bandwidth and power budget of APs, while still ensuring the quality of service (QoS) for all UEs; $(ii)$ managing the network interference more effectively. 

For notational convenience, let us denote the $m$-th AP,  $ k $-th DL UE and $\ell $-th UL UE  by $\AP$, $\DLU$ and $\ULU$, respectively. The channel vectors and matrices from  $ \AP  \rightarrow\DLU$, $\ULU\rightarrow \AP $, $\ULU\rightarrow\DLU$ and $ \APx{m'}\rightarrow\APx{m},\forall m'\in\cM $ are denoted by $ \mathbf{h}_{km}^\dl\in \mathbb{C}^{1\times N_m} $, $ \mathbf{h}_{m\ell}^\ul \in \mathbb{C}^{N_m\times1}$, $ g_{k\ell}^{\mathtt{cci}}\in \mathbb{C} $ and $ \mathbf{G}_{mm'}^{\AtoA} \in \mathbb{C}^{N_m\times N_{m'}} $, respectively. Note that $ \mathbf{G}_{mm}^{\AtoA}$ is the SI channel at $\AP$, while $ \mathbf{G}_{mm'}^{\AtoA}, \forall m\neq m'$ is referred to as the inter-AP interference (IAI) channel since   UL signals received at $ \APx{m} $ are corrupted by  DL signals sent from $ \APx{m'} $. The reason for this is that the SI signal can only be suppressed at the local APs \cite{Bharadia13,Bharadia14}.  To differentiate the residual SI and IAI channels, we model $ \mathbf{G}_{mm'}^{\AtoA}$ as follows:
\begin{numcases}{ \mathbf{G}_{mm'}^{\AtoA}=}
\sqrt{\rho_{mm}^{\mathtt{RSI}}}\mathbf{G}_{mm}^{\mathtt{SI}}, & if $m=m'$, \nonumber\\
\mathbf{G}_{mm'}^{\AtoA}, & otherwise, \nonumber
\end{numcases}
where $\mathbf{G}_{mm}^{\mathtt{SI}}$ denotes the fading loop channel at $\AP$ which interferes the UL reception due to the concurrent DL transmission, and $\rho_{mm}^{\mathtt{RSI}}\in[0,1)$ is the residual SI suppression (SiS) level after all real-time cancellations in analog-digital domains \cite{Sabharwal:JSAC:Feb2014,Dinh:TCOMM:2017,Dinh:JSAC:18,Hieu:IEEETWC:June2019}. The fading loop channel $\mathbf{G}_{mm}^{\mathtt{SI}}$ can be  characterized as the Rician probability distribution
with a small Rician factor \cite{Duarte:TWC:12}, while other channels are generally modeled as $ \mathbf{h}=\sqrt{\beta}\mathbf{\ddot{h}} $ with $ \mathbf{h}\in\{\mathbf{G}_{mm'}^{\AtoA}, \mathbf{h}_{km}^{\dl}, \mathbf{h}_{m\ell}^{\ul}, g_{k\ell}^{\mathtt{cci}} \} $,  accounting for the effects
of large-scale fading $\beta$ (i.e., path loss and shadowing) and small-scale fading $\mathbf{\ddot{h}} $ whose elements follow independent and identically distributed (i.i.d.) $\mathcal{CN}(0,1)$ random variables (RVs).

%\vspace{2pt}
\subsubsection{Downlink Data Transmission}
%\vspace{2pt}
Let us denote by $ x_{k}^\dl $ and $ x_{\ell}^\ul$ the data symbols with unit power (i.e., $ \mathbb{E}\bigl[|x_{k}^{\dl}|^2\bigr]=1$ and $\mathbb{E}\bigl[|x_{\ell}^{\ul}|^2\bigr]=1$) intended for $\DLU$ and sent from $\ULU$, respectively. The beamforming vector $ \mathbf{w}_{km}\in \mathbb{C}^{N_m \times 1} $ is employed to precode the data symbol $ x_{k}^\dl $ of $\DLU$ in the DL, while $ p_{\ell}$ denotes the transmit power  of $\ULU$ in the UL. After performing a joint radio resource management algorithm at the CPU, the data of $\DLU $ is routed to $\AP$ via the $m$-th backhaul link only if $\|\mathbf{w}_{km}\|> 0$. To do so, let us introduce new binary variables $ \alpha_{km}\in\{0,1\}, \forall k\in\mathcal{K}, m\in\mathcal{M}$ to  represent the association relationship between $\AP$ and $\DLU$, i.e.,
 $ \alpha_{km}=1 $ implying that $\DLU$ is served by $ \AP $ and $ \alpha_{km}=0 $, otherwise.  Using these notations, the  signal received at  $\DLU$ can be expressed as
\begingroup\allowdisplaybreaks\begin{align} \label{eq: received signal at DLU}
	y_{k}^\dl = \sum\nolimits_{m\in\mathcal{M}}\alpha_{km}\mathbf{h}_{km}^\dl\mathbf{w}_{km}x_{k}^\dl + \underbrace{\sum\nolimits_{\ell\in\mathcal{L}} \sqrt{p_{\ell}} g_{k\ell}^{\mathtt{cci}} x_{\ell}^\ul}_{\text{CCI}}%\nonumber\\
	+ \underbrace{\sum\nolimits_{m\in\mathcal{M}}\sum\nolimits_{k'\in\mathcal{K}\setminus\{k\}}\alpha_{k'm}\mathbf{h}_{km}^\dl\mathbf{w}_{k'm}x_{k'}^\dl}_{\text{MUI}}  + n_{k}, 
\end{align}\endgroup
where $ n_k\sim \mathcal{CN}(0,\sigma_{k}^2) $ is the additive white Gaussian noise (AWGN), and $\sigma_{k}^2$ is the noise variance. By treating
MUI and CCI as noise, the received SINR at  $ \DLUi{k} $ is given as
\begin{align} \label{eq: DL SINR - general}
\gamma_{k}^{\dl}(\mathbf{w}, \mathbf{p},\boldsymbol{\alpha}) = \frac{\sum_{m\in\mathcal{M}}\alpha_{km}|\mathbf{h}_{km}^\dl\mathbf{w}_{km}|^2 }{\chi_k(\mathbf{w}, \mathbf{p},\boldsymbol{\alpha})},
\end{align}
where $\chi_k(\mathbf{w}, \mathbf{p},\boldsymbol{\alpha})\triangleq\sum_{m\in\mathcal{M}}\sum_{k'\in\mathcal{K}\setminus\{k\}}\alpha_{k'm}|\mathbf{h}_{km}^\dl\mathbf{w}_{k'm}|^2+\sum_{\ell\in\mathcal{L}}p_{\ell}|g_{k\ell}^{\mathtt{cci}}|^2+\sigma_{k}^2$, $\mathbf{w}\triangleq[\mathbf{w}_{1}^H,\cdots,\mathbf{w}_K^H]^H\in\mathbb{C}^{NK\times 1}$ with $\mathbf{w}_k\triangleq[\mathbf{w}_{k1}^H,\cdots,\mathbf{w}_{kM}^H]^H\in\mathbb{C}^{N\times 1}$, $\mathbf{p}\triangleq[p_{1},\cdots,p_L]^T\in\mathbb{R}^{L\times 1}$, and $\boldsymbol{\alpha}\triangleq\{\alpha_{km}\}_{\forall k\in\mathcal{K}, m\in\mathcal{M}}$. We note that in \eqref{eq: DL SINR - general}, $ \alpha_{km} $ is equal to $ \alpha_{km}^2 $ for any $\alpha_{km}\in\{0,1\}$.

\vspace{5pt}
\subsubsection{Uplink Data Transmission}
\vspace{5pt}
The received signal at $ \AP $ can be expressed as
\begingroup
\allowdisplaybreaks{\small\begin{IEEEeqnarray}{rCl} \label{eq: received signal at AP}
\mathbf{y}_m^{\ul} & = \sum\nolimits_{\ell\in\cL} \sqrt{p_{\ell}}  \mathbf{h}_{m\ell}^{\ul} x_{\ell}^{\ul}  + \underbrace{\sum\nolimits_{ \substack{m'\in\cM}}\sum\nolimits_{k\in\cK} \alpha_{km'} \mathbf{G}_{mm'}^{\AtoA} \mathbf{w}_{km'} x_{k}^{\dl}}_{\text{RSI + IAI}}  + \mathbf{n}_m,\quad
\end{IEEEeqnarray}}\endgroup
 where $ \mathbf{n}_m\sim \mathcal{CN}(0,\sigma_{\mathtt{AP}}^2\mathbf{I}) $ is the AWGN. The CPU aggregates  the received signals from all APs, and  the $ \ULU $'s signal is then extracted by using a specific receiver. In general, let us denote the receiver vector to decode the $ \ULU $'s message received  at $ \APx{m} $ by $\mathbf{a}_{m\ell}\in\mathbb{C}^{1\times N_m}$, and thus, the received signal of $ \ULU $ at $ \APx{m} $ can be expressed as $r_{m\ell}^{\ul}=\mathbf{a}_{m\ell}\mathbf{y}_m^{\ul}$.
Consequently, the post-detection signal for decoding the $ \ULU $'s signal is $r_{\ell}^{\ul}=\sum_{m\in\mathcal{M}}r_{m\ell}^{\ul}$.  By defining $\mathbf{h}_{\ell}^{\ul}\triangleq\bigl[(\mathbf{h}_{1\ell}^{\ul})^H,\cdots,(\mathbf{h}_{M\ell}^{\ul})^H\bigr]^H\in\mathbb{C}^{N\times 1}$, $ \mathbf{\bar{G}}_{m'}^{\AtoA}\triangleq\bigl[(\mathbf{G}_{1m'}^{\AtoA})^H,\cdots,(\mathbf{G}_{Mm'}^{\AtoA})^H\bigr]^H\in\mathbb{C}^{N\times N_{m'}}$,  $ \mathbf{a}_{\ell}=[\mathbf{a}_{1\ell},\cdots,\mathbf{a}_{M\ell}] \in\mathbb{C}^{1\times N}$ and  $\mathbf{n}\triangleq[\mathbf{n}_1^H,\cdots,\mathbf{n}_M^H]^H\in\mathbb{C}^{N\times 1}$, the SINR in decoding   $ \ULU $'s message is given as
\begin{IEEEeqnarray}{rCl} \label{eq: UL SINR - general}
	\gamma_{\ell}^{\ul}(\mathbf{w}, \mathbf{p},\boldsymbol{\alpha}) = \frac{ p_{\ell}|\mathbf{a}_{\ell}  \mathbf{h}_{\ell}^{\ul}|^2 }{\sum_{\ell'\in\mathcal{L}\setminus\{\ell\}} p_{\ell'} |\mathbf{a}_{\ell}  \mathbf{h}_{\ell'}^{\ul}|^2 + \mathcal{I}^{\AtoA}_{\ell} + \sigma_{\mathtt{AP}}^2\|\mathbf{a}_{\ell}\|^2 },\quad
\end{IEEEeqnarray}
where $ \mathcal{I}^{\AtoA}_{\ell}\triangleq \sum_{m'\in\mathcal{M}}\sum_{k\in\mathcal{K}} \alpha_{km'} |\mathbf{a}_{\ell} \mathbf{\bar{G}}_{m'}^{\AtoA}\mathbf{w}_{km'}|^2 $ is the aggregation of  RSI and IAI.

\subsection{Power Consumption Model}
We now present a power consumption model that accounts for  power consumption for data transmission and baseband processing, as well as circuit operation  \cite{Tombaz:IEEEWC:Oct2011, Bjornson:IEETWC:June2015}. As previously discussed, we introduce new binary variables $ \mu_{m}\in\{0,1\}, \forall  m\in\mathcal{M}$ to represent operation statuses of $\AP$. In particular, $ \APx{m} $ is selected to be in the active mode if $ \mu_{m}=1 $, and switched to sleep mode otherwise. With $ \boldsymbol{\mu}\triangleq\{\mu_{m}\}_{\forall m\in\mathcal{M}} $,
the total power consumption is generally written as
\begin{align}\label{eq:total power consumption}
P_{\mathtt{T}}(\mathbf{w}, \mathbf{p},\boldsymbol{\alpha},\boldsymbol{\mu}) = P_{\mathtt{D}}(\mathbf{w}, \mathbf{p},\boldsymbol{\alpha},\boldsymbol{\mu}) + P_{\mathtt{C}}(\boldsymbol{\mu}),
\end{align}
where $ P_{\mathtt{D}}(\mathbf{w}, \mathbf{p},\boldsymbol{\alpha},\boldsymbol{\mu}) $ is the power consumption for data transmission and baseband processing, and $ P_{\mathtt{C}}(\boldsymbol{\mu}) $ is the power consumption for circuit operation; these are detailed as follows:
\begin{itemize}
	\item The  power consumption  $P_{\mathtt{D}}(\mathbf{w}, \mathbf{p},\boldsymbol{\alpha},\boldsymbol{\mu})$  can be sub-categorized into three main types as:
	\begingroup\allowdisplaybreaks\begin{IEEEeqnarray}{rCl} \label{eq: power consump. for data}
	P_{\mathtt{D}}(\mathbf{w}, \mathbf{p},\boldsymbol{\alpha},\boldsymbol{\mu}) = \underbrace{\sum\nolimits_{m\in\mathcal{M}} \frac{\mu_{m}}{\nu_{\AP}} \sum\nolimits_{k\in\mathcal{K}} \|\mathbf{w}_{km}\|^2 + \sum\nolimits_{\ell\in\mathcal{L}}\frac{p_{\ell}}{\nu_{\ell}^{\ul}}}_{\text{radiated power}} %\qquad\nonumber\\
	+ \underbrace{B\cdot F_{\mathtt{SE}}(\mathbf{w}, \mathbf{p},\boldsymbol{\alpha})\cdot P^{\mathtt{bh}}}_{\text{load-dependent power}} \qquad\nonumber \\
	+ \underbrace{\sum\nolimits_{m\in\mathcal{M}}\sum\nolimits_{k\in\mathcal{K}}\mu_{m}\alpha_{km}P_{km}^{\dl} + L\sum\nolimits_{m\in\mathcal{M}}\mu_{m}P_{m}^{\ul}}_{\text{baseband power}}.\qquad
	\end{IEEEeqnarray}\endgroup
The \textit{radiated power}	is the power consumption for the transmitted data between APs and UEs, where  $\nu_{\AP}\in[0,1] \text{ and } \nu_{\ell}^{\ul}\in[0,1]$ are the power amplifier (PA) efficiencies at $ \APx{m} $ and $ \ULU $ depending on the design techniques and operating conditions of the
PA \cite{AuerWC11}. {\majrev The \textit{load-dependent power} is the power consumption spent  to transfer the data between  APs
and  CPU in the  backhaul which is proportional to the achievable sum rates \cite{Bjornson:IEETWC:June2015}, where $ B $, $F_{\mathtt{SE}}(\mathbf{w}, \mathbf{p},\boldsymbol{\alpha}) $ and $ P^{\mathtt{bh}} $ are the system bandwidth, total SE and  average backhaul traffic power of all  links, respectively.} The \textit{baseband power} is the required power for data processing, waveform design, sync and precoder/receiver computing for $ \DLUi{k} $ (denoted by $ P_{km}^{\dl} $) and per-UL-user reception (denoted by $ P_{m}^{\ul} $). It is obvious that when $\mu_m=0$, $P_{km}^{\dl} =  P_{m}^{\ul} = 0$.
	\item The power consumption $P_{\mathtt{C}}(\boldsymbol{\mu})$ can be modeled as 
\begin{IEEEeqnarray}{rCl}
	P_{\mathtt{C}}(\boldsymbol{\mu}) =&& \sum_{m\in\mathcal{M}} \mu_{m}P_{\AP}^{\mathtt{a}} + \sum_{m\in\mathcal{M}} (1-\mu_{m})P_{\AP}^{\mathtt{s}} %\nonumber\\ &&
	+ \sum_{m\in\mathcal{M}}P_{\mathtt{AP}_m}^{\mathtt{cir}} + \sum_{k\in\mathcal{K}}P_{k}^{\dl,\mathtt{cir}} + \sum_{\ell\in\mathcal{L}}P_{\ell}^{\ul,\mathtt{cir}},\label{curcirpower}
	\end{IEEEeqnarray}
	where $ P_{\AP}^{\mathtt{a}} $ and $ P_{\AP}^{\mathtt{s}} $ are the fixed powers to keep $ \APx{m} $ in the active  and sleep modes, respectively; $ P_{\mathtt{AP}_m}^{\mathtt{cir}} $, $ P_{k}^{\dl,\mathtt{cir}} $ and $ P_{\ell}^{\ul,\mathtt{cir}} $ are the powers required for circuit operation at $ \AP $, DL and UL UEs, respectively. 
\end{itemize}
	
\section{Optimization Problem Design}\label{OptimizationProblemDesign}
\subsection{Original Problem Formulation}
From \eqref{eq: DL SINR - general} and \eqref{eq: UL SINR - general}, the SE in nats/s/Hz is given as
	\begin{align}
	F_{\mathtt{SE}}\bigl(\mathbf{w}, \mathbf{p},\boldsymbol{\alpha}\bigr)&\triangleq \sum\nolimits_{k\in\mathcal{K}}R\bigl(\gamma_{k}^{\dl}(\mathbf{w}, \mathbf{p},\boldsymbol{\alpha})\bigr)%\nonumber\\	&\;
	+\sum\nolimits_{\ell\in\mathcal{L}}R\bigl(\gamma_{\ell}^{\ul}(\mathbf{w}, \mathbf{p},\boldsymbol{\alpha})\bigr), \label{eq: F-SE} 
	\end{align}
	where $R(x)\triangleq \ln(1+x)$. {\majrev The EE in nats/Joule is defined as the ratio between the sum throughput (nats/s) and  the total power consumption (Watt):}
\begin{align}
	F_{\mathtt{EE}}\bigl(\mathbf{w}, \mathbf{p},\boldsymbol{\alpha},\boldsymbol{\mu}\bigr)&\triangleq \frac{B\cdot F_{\mathtt{SE}}(\mathbf{w}, \mathbf{p},\boldsymbol{\alpha})}{P_{\mathtt{T}}(\mathbf{w}, \mathbf{p},\boldsymbol{\alpha},\boldsymbol{\mu})}. \label{eq: F-EE}
\end{align} 
To lighten the notations, the system bandwidth $B$ will be omitted in the derivation of the algorithms in this paper, without affecting the optimal solutions. 
%As previously mentioned, the utility function should simultaneously capture both the SE and EE maximization in a single framework. 
By introducing the constant $ \eta \in \{0,1\} $  for the objective-function selection between the SE and EE,  the utility function can be written as
\begin{align}
F(\mathbf{w}, \mathbf{p},\boldsymbol{\alpha},\boldsymbol{\mu})=\eta F_{\mathtt{SE}}(\mathbf{w}, \mathbf{p},\boldsymbol{\alpha}) + (1-\eta) F_{\mathtt{EE}}(\mathbf{w}, \mathbf{p},\boldsymbol{\alpha},\boldsymbol{\mu}). \nonumber
\end{align}
It is worth mentioning that if $\eta$ = 1 ($\eta$ = 0, respectively), we arrive at the SE maximization problem (the EE maximization, respectively).

 Assuming perfect CSI between APs and  UEs, we study a joint design of  power control, AP-DL UE association
and AP selection, which is formulated as
\begingroup
\allowdisplaybreaks\begin{subequations} \label{eq: prob. general form bi-obj. trade-off}
	\begin{IEEEeqnarray}{cl}
		\underset{\mathbf{w}, \mathbf{p},\boldsymbol{\alpha},\boldsymbol{\mu}}{\max} &\quad  F(\mathbf{w}, \mathbf{p},\boldsymbol{\alpha},\boldsymbol{\mu}) \label{eq: prob. general form bi-obj. trade-off :: a} \\
		\st & \quad \mu_{m} \in \{0,1\}, \; \forall m \in \mathcal{M}, \label{eq: prob. general form bi-obj. trade-off :: f} \\
		& \quad \alpha_{km} \in \{0,1\}, \; \forall k \in \mathcal{K},m \in \mathcal{M}, \label{eq: prob. general form bi-obj. trade-off :: g} \\
		& \quad \|\mathbf{w}_{km}\|^2 \leq \alpha_{km} P_{\mathtt{AP}_m}^{\max} , \; \forall k \in \mathcal{K},m \in \mathcal{M}, \label{eq: prob. general form bi-obj. trade-off :: h} \\
		&\quad \sum\nolimits_{k\in\cK}\alpha_{km}\|\mathbf{w}_{km}\|^2 \leq \mu_{m} P_{\mathtt{AP}_m}^{\max},\;\forall m\in\mathcal{M}, \label{eq: prob. general form bi-obj. trade-off :: b} \qquad\\
		& \quad 0 \leq p_{\ell} \leq P_{\ell}^{\max},\; \forall \ell \in \mathcal{L}, \label{eq: prob. general form bi-obj. trade-off :: c} \\
		& \quad 
		R\bigl(\gamma_{k}^{\dl}(\mathbf{w}, \mathbf{p},\boldsymbol{\alpha})\bigr) \geq \bar{R}_{k}^{\dl}, \; \forall  k \in \mathcal{K}, \label{eq: prob. general form bi-obj. trade-off :: d} \\
		&  \quad
		R\bigl(\gamma_{\ell}^{\ul}(\mathbf{w}, \mathbf{p},\boldsymbol{\alpha})\bigr) \geq \bar{R}_{\ell}^{\ul}, \; \forall \ell \in \mathcal{L}. \label{eq: prob. general form bi-obj. trade-off :: e}
	\end{IEEEeqnarray}							
\end{subequations}\endgroup
 {\hili Constraint \eqref{eq: prob. general form bi-obj. trade-off :: h} is used to express the AP-UE association, while constraints} \eqref{eq: prob. general form bi-obj. trade-off :: b} and \eqref{eq: prob. general form bi-obj. trade-off :: c} imply that the  transmit powers at  $\AP$ and $\DLU$ are limited by their maximum power budgets $ P_{\mathtt{AP}_m}^{\max} $ and $ P_{\ell}^{\max} $, respectively.  Moreover, constraints \eqref{eq: prob. general form bi-obj. trade-off :: d} and \eqref{eq: prob. general form bi-obj. trade-off :: e} are used to ensure the predetermined rate requirements $\bar{R}_{k}^{\dl}$ and $\bar{R}_{k}^{\ul}$ for $\DLU$ and $\ULU$, respectively. We can see that the objective \eqref{eq: prob. general form bi-obj. trade-off :: a} is  nonconcave and the feasible set is also nonconvex. Hence, problem \eqref{eq: prob. general form bi-obj. trade-off} is a mixed-integer nonconvex optimization problem due to  binary variables involved, which is generally NP-hard.

\vspace{-10pt}
\subsection{Tractable Problem Formulation for \eqref{eq: prob. general form bi-obj. trade-off}}
The major difficulty in solving  problem \eqref{eq: prob. general form bi-obj. trade-off} is due to  binary variables involved. It is not practical to try all 
possible AP-DL UE associations and AP selections, especially for networks of large size. In addition, the strong coupling between  continuous variables $(\mathbf{w}, \mathbf{p})$ and  binary variables ($\boldsymbol{\alpha},\boldsymbol{\mu}$) makes problem \eqref{eq: prob. general form bi-obj. trade-off} even more difficult.  Consequently, the problem is intractable and it is impossible to solve it directly. Even for a fixed ($\boldsymbol{\alpha},\boldsymbol{\mu}$), a direct application of the well-known Dinkelbach  algorithm \cite{Dinkelbach67} for \eqref{eq: prob. general form bi-obj. trade-off} still involves a nonconvex  problem, and thus, its convergence may not be always guaranteed \cite{BuzziJSAC16}. In what follows, we present a tractable form of \eqref{eq: prob. general form bi-obj. trade-off} by exploiting the special relationship between  continuous   and  binary variables, based on which the combination of  ICA method and Dinkelbach transformation can be applied to solve it efficiently for various transmission strategies.
\subsubsection{Binary Reduction of $\boldsymbol{\alpha}$}
The binary variables $\boldsymbol{\alpha}$ and the continuous variables $\mathbf{w}$ are strongly coupled, as revealed by the following lemma.
\begin{lemma} \label{thm: relationship alpha and beamforming vector}
	If problem \eqref{eq: prob. general form bi-obj. trade-off} contains the optimal solution $ \alpha_{km}^{*}=0 $ for some ($k,m$),  it also admits $\mathbf{w}_{km}=\mathbf{0}$ as an optimal solution for the corresponding beamforming vector.
\end{lemma}
\begin{proof}
Please see Appendix \ref{app: relationship alpha and beamforming vector}.
\end{proof}
\noindent From \textbf{Lemma} \ref{thm: relationship alpha and beamforming vector}, it is straightforward to see that constraint \eqref{eq: prob. general form bi-obj. trade-off :: h} is naturally satisfied at the optimal point. In particular, when $\alpha_{km}=1$ (or $\alpha_{km}=0$), constraint \eqref{eq: prob. general form bi-obj. trade-off :: h} is addressed by  tighter constraint \eqref{eq: prob. general form bi-obj. trade-off :: b} (or by \textbf{Lemma} \ref{thm: relationship alpha and beamforming vector}). In connection to \textbf{Lemma} \ref{thm: relationship alpha and beamforming vector}, 
 we establish the following main result.
\begin{theorem} \label{thr: theorem 1}
	For any $k\in\mathcal{K}$ and $m\in\mathcal{M}$, let the null space (including zero vector) of $ \mathbf{h}_{km}^\dl $ be $\ker(\mathbf{h}_{km}^\dl)$ and $\mathbf{u}\in\mathcal{U}\triangleq\bigl\{(\mathbf{w}, \mathbf{p},\boldsymbol{\alpha},\boldsymbol{\mu})|\mathbf{w}_{km}\in\ker(\mathbf{h}_{km}^\dl)\bigr\} \subseteq \mathcal{F}$, where $\mathcal{F}$ is the feasible set of \eqref{eq: prob. general form bi-obj. trade-off}. The following state is obtained:
	\begin{IEEEeqnarray}{rCl} \label{eq: theorem 1}
	\mathbf{\tilde{u}} = \underset{\mathbf{u}\in\mathcal{U}}{\argmax} \;F(\mathbf{u}) = (\mathbf{w}, \mathbf{p},\boldsymbol{\alpha},\boldsymbol{\mu}|\mathbf{w}_{km}=\mathbf{0}\;\&\;\alpha_{km}=0).\quad\;\;
	\end{IEEEeqnarray} 
\end{theorem}
\begin{proof}
	Please see Appendix \ref{app: theorem 1}.
\end{proof}
{\hili The merits of \textbf{Theorem} \ref{thr: theorem 1} are detailed as follows.}  First, even some possible values of $ \mathbf{w}_{km}\in\ker\bigl(\mathbf{h}_{km}^\dl\bigr) $ can make $ |\mathbf{h}_{km}^\dl\mathbf{w}_{km}|^2 $ to be null, the zero vector $ \mathbf{w}_{km}=\mathbf{0} $ is the best value  among them. Second, there  exists only one of two pairs for the beamforming vector and AP-DL UE association in the optimal solution, which is $(\mathbf{w}_{km}^*,\alpha_{km}^*)\in \bigl\{(\mathbf{0},0),(\mathbf{\bar{w}}, 1) \bigr\}$ with $\mathbf{\bar{w}}\notin\ker\bigl(\mathbf{h}_{km}^\dl\bigr)$.  Without loss of optimality, we can replace $\boldsymbol{\alpha}$ by $ \mathbf{1} $ in the component containing the compound of $\boldsymbol{\alpha}$ and $\mathbf{w}$, and use a substituting function of $\mathbf{w}$ for $\boldsymbol{\alpha}$ in others. Particularly, we define the 2-tuple of continuous variables as $\mathcal{C}\triangleq\{\mathbf{w}, \mathbf{p}\}$, and $\boldsymbol{\Gamma}_{\dl}\triangleq\bigl\{\gamma_{k}^{\dl}(\mathcal{C},\mathbf{1})|\forall k\in\mathcal{K}\bigr\}$ and $\boldsymbol{\Gamma}_{\ul}\triangleq\bigl\{\gamma_{\ell}^{\ul}(\mathcal{C},\mathbf{1})|\forall \ell\in\mathcal{L}\bigr\}$ with all entries of $ \boldsymbol{\alpha} $ being replaced by ones. 

In short, problem \eqref{eq: prob. general form bi-obj. trade-off} can be rewritten as
\begingroup
\allowdisplaybreaks\begin{subequations} \label{eq: prob. bi-obj. - no alpha}
	\begin{IEEEeqnarray}{cl}
		\underset{\mathcal{C}\triangleq\{\mathbf{w}, \mathbf{p}\},\boldsymbol{\mu}}{\max} &\quad  \eta \bar{F}_{\mathtt{SE}}\bigl(\boldsymbol{\Gamma}_{\dl},\boldsymbol{\Gamma}_{\ul}\bigr) + (1-\eta) \bar{F}_{\mathtt{EE}}\bigl(\boldsymbol{\Gamma}_{\dl},\boldsymbol{\Gamma}_{\ul},\mathcal{C},\boldsymbol{\mu}\bigr) \label{eq: prob. bi-obj. - no alpha :: a} \qquad\\
		\st & \quad  \sum_{k\in\cK}\|\mathbf{w}_{km}\|^2 \leq \mu_{m} P_{\mathtt{AP}_m}^{\max},\;\forall m\in\mathcal{M}, \label{eq: prob. bi-obj. - no alpha :: c} \\
				& \quad 
		R\bigl(\gamma_{k}^{\dl}(\mathcal{C},\mathbf{1})\bigr) \geq \bar{R}_{k}^{\dl}, \; \forall  k \in \mathcal{K}, \label{eq: prob. bi-obj. - no alpha :: e} \\
		&  \quad
		R\bigl(\gamma_{\ell}^{\ul}(\mathcal{C},\mathbf{1})\bigr) \geq \bar{R}_{\ell}^{\ul}, \; \forall \ell \in \mathcal{L}, \label{eq: prob. bi-obj. - no alpha :: f}\\
			&  \quad\eqref{eq: prob. general form bi-obj. trade-off :: f}, \eqref{eq: prob. general form bi-obj. trade-off :: c},\label{eq: prob. bi-obj. - no alpha :: g}
	\end{IEEEeqnarray}							
\end{subequations}\endgroup
where $\bar{F}_{\mathtt{SE}}\bigl(\boldsymbol{\Gamma}_{\dl},\boldsymbol{\Gamma}_{\ul}\bigr)  \triangleq  R_{\Sigma}(\boldsymbol{\Gamma}_{\dl}) + R_{\Sigma}(\boldsymbol{\Gamma}_{\ul}),$ $ \bar{F}_{\mathtt{EE}}\bigl(\boldsymbol{\Gamma}_{\dl},\boldsymbol{\Gamma}_{\ul},\mathcal{C},\boldsymbol{\mu}\bigr) \triangleq \frac{\bar{F}_{\mathtt{SE}}\bigl(\boldsymbol{\Gamma}_{\dl},\boldsymbol{\Gamma}_{\ul}\bigr)}{\bar{P}_{\mathtt{T}}(\mathcal{C},\boldsymbol{\mu})}$, and $\bar{P}_{\mathtt{T}}(\mathcal{C},\boldsymbol{\mu})\triangleq P_{\mathtt{T}}(\mathcal{C},f_{\mathtt{spr}}(\mathbf{w}),$ $\boldsymbol{\mu})$
with $ R_{\Sigma}(\mathcal{X}) = \sum_{x\in\mathcal{X}} R(x) $. The signal-power ratio function is defined as
\begin{IEEEeqnarray}{cl}
f_{\mathtt{spr}}:\mathbf{w}\rightarrow \mathbf{r}^{\mathtt{sp}}\triangleq\bigl[r_{\mathtt{sp}}\bigl(\mathbf{w}_{km},\mathbf{h}_{km}^{\dl}|\mathbf{w}_{k}^{(\kappa)},\mathbf{h}_{k}^{\dl}\bigr)\bigr]_{\forall k\in\mathcal{K},m\in\mathcal{M}},\qquad
\end{IEEEeqnarray}
with $\mathbf{h}_k^\dl\triangleq[\mathbf{h}_{k1}^\dl,\cdots,\mathbf{h}_{kM}^\dl]\in\mathbb{C}^{1\times N}$, and
\begin{align} \label{eq: consumed power rate function}
r_{\mathtt{sp}}\bigl(\mathbf{x}_{1},\mathbf{c}_{1}|\mathbf{x}_{2},\mathbf{c}_{2}\bigr)\triangleq\frac{|\mathbf{c}_{1}\mathbf{x}_{1}|^2}{|\mathbf{c}_{2}\mathbf{x}_{2}|^2+\epsilon},
\end{align}
where $\epsilon$ is a very small real number added to avoid a numerical problem when $ \AP $ turns to sleep mode, and $ \mathbf{w}_{k}^{(\kappa)}$ is a feasible point of $ \mathbf{w}_{k}$ at the $\kappa$-th iteration of an iterative algorithm presented shortly. 
 $ \alpha_{km} $ is correspondingly replaced by $ r_{\mathtt{sp}}\bigl(\mathbf{w}_{km},\mathbf{h}_{km}^{\dl}|\mathbf{w}_{k}^{(\kappa)},\mathbf{h}_{k}^{\dl}\bigr) $ in  $ P_{\mathtt{T}}(\mathcal{C},f_{\mathtt{spr}}(\mathbf{w}),\boldsymbol{\mu}) $. {\majrev We note that  \eqref{eq: consumed power rate function} is considered as a soft converter from the binary variables into continuous ones, which also indicates the quality of connection between an AP and a UE.} As a result, $ \mathbf{0} \preceq \mathbf{r}^{\mathtt{sp}} \preceq \mathbf{1} $ is considered as an estimate of $ \boldsymbol{\alpha} $ after solving problem \eqref{eq: prob. bi-obj. - no alpha}, i.e.,
\begin{IEEEeqnarray}{cl} \label{eq: compute alpha}
\alpha_{km}^*=\mathcal{B}\bigl(r_{\mathtt{sp}}\bigl(\mathbf{w}_{km}^{*},\mathbf{h}_{km}^{\dl}|\mathbf{w}_{k}^{*},\mathbf{h}_{k}^{\dl}\bigr),\varpi\bigr), \forall k\in\mathcal{K},m\in\mathcal{M},\qquad
\end{IEEEeqnarray}
where
\begin{align}\label{eq:Bfunction}
\mathcal{B}(x, \varpi) \triangleq \begin{cases}
1, & \text{if } x>\varpi, \\
0, & \text{if } x\leq\varpi,
\end{cases}
\end{align}
and the per-AP power signal ratio $\varpi$ is a small number indicating $x\approx 0 $ if $ x\leq\varpi $, and $ \mathbf{w}_{km}^{*} $ is the optimal solution of $ \mathbf{w}_{km}$. 

\begin{remark}By \textbf{Lemma} \ref{thm: relationship alpha and beamforming vector}, it is true that $ r_{\mathtt{sp}}\bigl(\mathbf{w}_{km}^{*},\mathbf{h}_{km}^{\dl}|\mathbf{w}_{k}^{*},\mathbf{h}_{k}^{\dl}\bigr)\leq\varpi $ yields  $\mathbf{w}_{km}^{*} \rightarrow \mathbf{0} $. Without loss of optimality, we can omit $ \boldsymbol{\alpha} $ in the following derivations.\thmend
\end{remark}

\subsubsection{Binary Reduction of $\boldsymbol{\mu}$}
The binary variable $\boldsymbol{\mu}$ is mainly involved in \eqref{eq:total power consumption}. Using \textbf{Lemma} \ref{thm: relationship alpha and beamforming vector}, we have the following result.
\begin{theorem}\label{thr: theorem 2}
By treating  $ \mu_{m}, \forall m\in\mathcal{M}$ as a constant in each iteration, its solution in the next iteration is iteratively updated as:
\begin{IEEEeqnarray}{rCl} \label{eq: compute mu}
&&\mu_{m}^{(\kappa+1)} = \max \left\{ \underset{k\in\mathcal{K}}{\max}\; \mathcal{B}\Bigl(r_{\mathtt{sp}}\bigl(\mathbf{w}_{km}^{(\kappa)},\mathbf{h}_{km}^{\dl}|\mathbf{w}_{k}^{(\kappa)},\mathbf{h}_{k}^{\dl}\bigr),\varpi\Bigr), %\right.\nonumber\quad\\ &&\qquad\left.
 \underset{\ell\in\mathcal{L}}{\max}\; \mathcal{B}\Bigl(r_{\mathtt{sp}}\bigl(\sqrt{p_{\ell}^{(\kappa)}}\mathbf{h}_{m\ell}^{\ul},\mathbf{a}_{m\ell}|\sqrt{p_{\ell}^{(\kappa)}}\mathbf{h}_{\ell}^{\ul},\mathbf{a}_{\ell}\bigr),\varpi\Bigr) \right\},\quad
\end{IEEEeqnarray}
where the function $\mathcal{B}(:,:)$ was defined in \eqref{eq:Bfunction}.
\end{theorem}
\begin{proof}
	Please see Appendix \ref{app: theorem 2}.
\end{proof}
From \textbf{Theorem \ref{thr: theorem 2}}, the total power consumption in \eqref{eq:total power consumption} can be rewritten as
\begin{align}
\bar{P}_{\mathtt{T}}(\mathcal{C},\boldsymbol{\mu}^{(\kappa)})=P_{\mathtt{D}}(\mathcal{C},f_{\mathtt{spr}}(\mathbf{w}),\boldsymbol{\mu}^{(\kappa)}) + P_{\mathtt{C}}(\boldsymbol{\mu}^{(\kappa)}),
\end{align}
which involves the  continuous variables in $\mathcal{C}$ only.

In summary,  the original problem \eqref{eq: prob. general form bi-obj. trade-off} can be cast as the following simplified problem:
\begin{subequations} \label{eq: prob. bi-obj. - continuous vars.}
	\begin{IEEEeqnarray}{cl}
		\underset{\mathcal{C}\triangleq\{\mathbf{w}, \mathbf{p}\}}{\max} &\ \; \eta \bar{F}_{\mathtt{SE}}\bigl(\boldsymbol{\Gamma}_{\dl},\boldsymbol{\Gamma}_{\ul}\bigr) + (1-\eta) \bar{F}_{\mathtt{EE}}\bigl(\boldsymbol{\Gamma}_{\dl},\boldsymbol{\Gamma}_{\ul},\mathcal{C},\boldsymbol{\mu}^{(\kappa)}\bigr) \label{eq: prob. bi-obj. - continuous vars. :: a} \qquad\\
		\st & \ \sum\nolimits_{k\in\cK}\|\mathbf{w}_{km}\|^2 \leq \mu_{m}^{(\kappa)} P_{\mathtt{AP}_m}^{\max},\;\forall m\in\mathcal{M}, \label{eq: prob. bi-obj. - continuous vars. :: b} \\
		& \quad \eqref{eq: prob. general form bi-obj. trade-off :: c}, \eqref{eq: prob. bi-obj. - no alpha :: e}, \eqref{eq: prob. bi-obj. - no alpha :: f}.
	\end{IEEEeqnarray}							
\end{subequations}

{\majrev
\begin{remark}
It is clear from the discussion above that solving \eqref{eq: prob. bi-obj. - continuous vars.} boils down to finding a saddle point for $\mathcal{C}\triangleq\{\mathbf{w}, \mathbf{p}\}$, while the binary variables $ \boldsymbol{\alpha}$ and $ \boldsymbol{\mu} $ are post-updated by \eqref{eq: compute alpha} and \eqref{eq: compute mu}, respectively. We should emphasize that the binary variables are relaxed to soft-update functions in \eqref{eq: prob. bi-obj. - continuous vars.} to reduce the complexity, while maintaining their roles as in the original problem \eqref{eq: prob. general form bi-obj. trade-off}. These results hold true for arbitrary linear precoder/receiver schemes, which are discussed in detail next.\thmend
\end{remark}
}

\section{Proposed Solution Based on Zero-Forcing}\label{sec:ZF}
In this section, we first present an efficient transmission design; this retains the simplicity of the well-known ZF method while enjoys the similar performance of the optimal MMSE method as in the context of massive MIMO \cite{Marzetta:Cambridge:2016}. Then, an iterative  algorithm based on the ICA method and Dinkelbach
transformation is developed to solve the problem design, followed by the initialization discussion.
\subsection{ZF-Based Transmission Design}
 To make ZF feasible,  the total  number of APs' antennas is required to be larger than the number of UEs, i.e., $ N>\max\{K,L\} $, which can be easily satisfied in  massive MIMO systems.
For ease of presentation, we first rearrange  the  sets of beamforming vectors, channel vectors and power allocation between transceivers as follows: $\mathbf{W}\triangleq[\mathbf{w}_{1},\cdots,\mathbf{w}_{K}]\in\mathbb{C}^{N\times K}$, $\mathbf{H}^\dl\triangleq\bigl[(\mathbf{h}_{1}^\dl)^H, \cdots,(\mathbf{h}_{M}^\dl)^H\bigr]^H \in \mathbb{C}^{K\times N}$,  $ \mathbf{H}^{\ul}\triangleq\bigl[\mathbf{h}_{1}^{\ul},\cdots, \mathbf{h}_{L}^{\ul} \bigr]\in\mathbb{C}^{N\times L} $, $\mathbf{G}^{\mathtt{cci}}\triangleq\bigl[(\mathbf{g}_{1}^{\mathtt{cci}})^H,\cdots, (\mathbf{g}_{K}^{\mathtt{cci}})^H\bigr]^H\in \mathbb{C}^{K\times L}$ with $\mathbf{g}_{k}^{\mathtt{cci}}\triangleq\bigl[g_{k1}^{\mathtt{cci}},\cdots,g_{kL}^{\mathtt{cci}}\bigr]\in \mathbb{C}^{1\times L}$,  $ \mathbf{\tilde{G}}^{\AtoA}\triangleq\bigl[\mathbf{\bar{G}}_1^{\AtoA},\cdots,$ $\mathbf{\bar{G}}_M^{\AtoA}\bigr]\in\mathbb{C}^{N\times N} $,  and $\mathbf{D}^{\ul}\triangleq\diag\bigl(\bigl[\sqrt{p_{1}}\cdots \sqrt{p_{L}}\bigr]\bigr)$.

\subsubsection{ZF-Based  DL Transmission}
For $ \mathbf{H}^\tZF= (\mathbf{H}^\dl)^H\bigl(\mathbf{H}^\dl(\mathbf{H}^\dl)^H\bigr)^{-1}$, the ZF precoder matrix is simply computed as
$\mathbf{W}^{\mathtt{ZF}} = \mathbf{H}^\tZF (\mathbf{D}^{\dl})^{\frac{1}{2}},
$
where  $ \mathbf{D}^{\dl}=\diag\bigl(\bigl[\omega_1\cdots\omega_K\bigr]\bigr) $ and $ \omega_{k} $ represents the  weight for $\DLU$. As a result, constraint  \eqref{eq: prob. bi-obj. - continuous vars. :: b} becomes
\begin{align} \label{eq: power constraint - ZF}
\tr\bigl((\mathbf{H}^{\tZF})^H\mathbf{B}_{m}\mathbf{H}^{\tZF}\mathbf{D}^{\dl}\bigr) \leq \mu_{m}^{(\kappa)} P_{\mathtt{AP}_m}^{\max}, \;\forall m\in\mathcal{M},
\end{align}
which is  a linear constraint, where $ \mathbf{B}_{m} = \diag(\mathbf{b}_{m})\in \{0,1\}^{N\times N} $ with 
\begin{align}
\mathbf{b}_{m} = \bigl( \underbrace{0\cdots 0}_{\sum\limits_{m'=1}^{m-1}N_{m'}}\;\underbrace{1\cdots 1}_{N_m}\; 0 \cdots 0\bigr).
\end{align}
The simplicity of ZF is attributed  to the fact that the size of $NK$ scalar variables of $\mathbf{w}$ is now reduced to $K$ scalar variables of $\boldsymbol{\omega}\triangleq[\omega_1,\cdots,\omega_K]^T\in\mathbb{R}^{K\times 1}$. The SINR of $\DLU$ with ZF precoder is
\begin{align}\label{eq:ZFSINIDL}
\gamma_{k}^{\dl,\tZF}(\boldsymbol{\omega}, \mathbf{p}) = \frac{\omega_k|\mathbf{h}_{k}^\dl\mathbf{h}_{k}^\tZF|^2}{\|\mathbf{g}_{k}^{\mathtt{cci}}\mathbf{D}^{\ul}\|^2+\sigma_{k}^2},
\end{align}
where $ \mathbf{h}_{k}^\tZF $ is the $ k $-th column of the ZF precoder $ \mathbf{H}^{\tZF} $ and the MUI term $ |\mathbf{h}_{k}^\dl\mathbf{w}_{k'}|^2 \approx 0, \forall k'\in\cK\setminus\{k\}$. 
\begin{remark}
The following result characterizes the relationship between  $ \boldsymbol{\omega} $ and $ \mathbf{W}^{\mathtt{ZF}} $:
	\begin{align}\label{eq:relationshipWOmega}
	\mathbf{W}^{\mathtt{ZF}}=f_{\mathbf{W}}(\boldsymbol{\omega},\mathbf{H}^{\tZF}),
	\end{align}
where $ f_{\mathbf{W}}(\boldsymbol{\omega},\mathbf{X}) \triangleq \mathbf{X}\bigl(\diag(\boldsymbol{\omega})\bigr)^{\frac{1}{2}} $. Hence,  $ \mathbf{w}_{km} $ is recovered by extracting from the $ ((m-1)N_m+1) $-th to $ (mN_m) $-th elements of $ \mathbf{w}_k $, where $ \mathbf{w}_k $ is the $ k $-th column of $ \mathbf{W}^{\mathtt{ZF}} $.\thmend
\end{remark}
\subsubsection{ZF-Based  UL Transmission}
Let $ \mathbf{A}^{\mathtt{ZF}}=\bigl((\mathbf{H}^\ul)^H\mathbf{H}^\ul\bigr)^{-1}(\mathbf{H}^\ul)^H \in\mathbb{C}^{L\times N} $ be the ZF receiver matrix at the CPU. The SINR of $ \ULU $ with ZF receiver is
\begin{align}\label{eq:ZFSINIUL}
\gamma_{\ell}^{\ul,\tZF}(\boldsymbol{\omega},\mathbf{p}) = \frac{p_{\ell}|\mathbf{a}_{\ell}^{\mathtt{ZF}}\mathbf{h}_{\ell}^{\ul}|^2 }{\underbrace{\|\mathbf{a}_{\ell}^{\mathtt{ZF}}\mathbf{\tilde{G}}^{\AtoA} \mathbf{W}^{\mathtt{ZF}}\|^2}_{\text{IAI + RSI}}+\sigma_{\mathtt{AP}}^2\|\mathbf{a}_{\ell}^{\mathtt{ZF}}\|^2},
\end{align}
where $ \mathbf{a}_{\ell}^{\mathtt{ZF}} $ is the $\ell$-th row of $ \mathbf{A}^{\mathtt{ZF}}$.

\subsection{Proposed Algorithm}
Before proceeding, we provide some useful approximate functions following  ICA properties \cite{Marks:78,Beck:JGO:10}, which will be frequently employed to devise the proposed solutions.
\begin{itemize}
	\item Consider the convex function $ h_{\mathtt{fr}}(x,y)\triangleq x^2/y $ with $(x,y)\in\mathbb{R}_{++}^2$. The concave lower bound of $h_{\mathtt{fr}}(x,y)$ at the feasible point $ (x^{(\kappa)},y^{(\kappa)}) $ is given as \cite{Dinh:JSAC:18}
\begin{align} \label{eq: frac. approx.}
h_{\mathtt{fr}}(x,y) \geq \frac{2x^{(\kappa)}}{y^{(\kappa)}}x - \frac{(x^{(\kappa)})^2}{(y^{(\kappa)})^2}y := h_{\mathtt{fr}}^{(\kappa)}(x,y).
\end{align}
	
	\item {\majrev For the quadratic function $ h_{\mathtt{qu}}(z)\triangleq z^2 $ with $z\in\mathbb{R}_{++}$, its concave lower bound at  $z^{(\kappa)} $ is
\begin{align} \label{eq: quad. approx.}
h_{\mathtt{qu}}(z) \geq 2z^{(\kappa)} z - (z^{(\kappa)})^2 := h_{\mathtt{qu}}^{(\kappa)}(z).
\end{align}}
\end{itemize}
\vspace{-20pt}
Next, problem \eqref{eq: prob. bi-obj. - continuous vars.} with ZF design now reduces to the following problem
	\begin{IEEEeqnarray}{rCl}\label{eq: prob. bi-obj. - ZF}
		\underset{\boldsymbol{\omega}, \mathbf{p}}{\max} &&\;\;  \eta \bar{F}_{\mathtt{SE}}\bigl(\boldsymbol{\Gamma}_{\dl}^{\mathtt{ZF}},\boldsymbol{\Gamma}_{\ul}^{\mathtt{ZF}}\bigr) + (1-\eta) \bar{F}_{\mathtt{EE}}\bigl(\boldsymbol{\Gamma}_{\dl}^{\mathtt{ZF}},\boldsymbol{\Gamma}_{\ul}^{\mathtt{ZF}},\mathcal{C}^{\mathtt{ZF}},\boldsymbol{\mu}^{(\kappa)}\bigr) \IEEEyessubnumber\label{eq: prob. bi-obj. - ZF :: a} \qquad\\
		\st	&&\quad   
		R\bigl(\gamma_{k}^{\dl,\tZF}(\boldsymbol{\omega}, \mathbf{p})\bigr) \geq \bar{R}_{k}^{\dl}, \; \forall  k \in \mathcal{K}, \IEEEyessubnumber\label{eq: prob. bi-obj. - ZF :: c} \\
		&&\quad   		R\bigl(\gamma_{\ell}^{\ul,\tZF}(\boldsymbol{\omega}, \mathbf{p})\bigr) \geq \bar{R}_{\ell}^{\ul}, \; \forall \ell \in \mathcal{L}, \IEEEyessubnumber\label{eq: prob. bi-obj. - ZF :: d} \\
		&&\quad  \eqref{eq: prob. general form bi-obj. trade-off :: c}, \eqref{eq: power constraint - ZF}, \IEEEyessubnumber\label{eq: prob. bi-obj. - ZF :: b} 
	\end{IEEEeqnarray}							
 where $ \boldsymbol{\Gamma_{\dl}^{\mathtt{ZF}}}\triangleq\{\gamma_{k}^{\dl,\tZF}(\boldsymbol{\omega}, \mathbf{p})|\forall k\in\mathcal{K}\} $ and $ \boldsymbol{\Gamma_{\ul}^{\mathtt{ZF}}}\triangleq\{\gamma_{\ell}^{\ul,\tZF}(\boldsymbol{\omega}, \mathbf{p})|\forall \ell\in\mathcal{L}\} $. Problem \eqref{eq: prob. bi-obj. - ZF} involves the nonconcave  objective  \eqref{eq: prob. bi-obj. - ZF :: a},  and nonconvex constraints \eqref{eq: prob. bi-obj. - ZF :: c} and \eqref{eq: prob. bi-obj. - ZF :: d}. To apply ICA method,  a new transformation with an equivalent feasible set is required.
Let us start by rewriting the objective \eqref{eq: prob. bi-obj. - ZF :: a} as $\bar{F}\bigl(\boldsymbol{\Gamma}_{\dl}^{\mathtt{ZF}},\boldsymbol{\Gamma}_{\ul}^{\mathtt{ZF}},\mathcal{C}^{\mathtt{ZF}}\bigr) \triangleq \bar{F}_{\mathtt{SE}}\bigl(\boldsymbol{\Gamma}_{\dl}^{\mathtt{ZF}},\boldsymbol{\Gamma}_{\ul}^{\mathtt{ZF}}\bigr) \bar{P}(\mathcal{C}^{\mathtt{ZF}}),$ where $ \bar{P}(\mathcal{C}^{\mathtt{ZF}},\boldsymbol{\mu}^{(\kappa)})\triangleq \bigl(\eta +\frac{(1-\eta)}{\bar{P}_{\mathtt{T}}(\mathcal{C}^{\mathtt{ZF}},\boldsymbol{\mu}^{(\kappa)})}\bigr) $ and $\mathcal{C}^{\mathtt{ZF}}\triangleq\{\boldsymbol{\omega}, \mathbf{p}\}$. 

\begin{theorem}\label{thr: theorem 3}
Problem \eqref{eq: prob. bi-obj. - ZF} is equivalent to the following problem
\begingroup
\allowdisplaybreaks\begin{subequations} \label{eq: prob. bi-obj. - equiZF}
	\begin{IEEEeqnarray}{cl}
		\underset{\boldsymbol{\omega}, \mathbf{p},\boldsymbol{\lambda},\phi}{\max} &\quad  \frac{\tilde{F}_{\mathtt{SE}}\bigl(\boldsymbol{\Lambda}_{\dl},\boldsymbol{\Lambda}_{\ul}\bigr)}{\phi}		\label{eq: prob. bi-obj. - equiZF :: a} \\
		\st	& \quad  \bar{P}(\mathcal{C}^{\mathtt{ZF}},\boldsymbol{\mu}^{(\kappa)}) \geq 1/\phi, \label{eq: prob. bi-obj. - equiZF :: b} \\
		& \quad \gamma_{k}^{\dl,\tZF}(\boldsymbol{\omega}, \mathbf{p})\geq\lambda_{k}^{\dl}, \; \forall  k \in \mathcal{K},\label{eq: prob. bi-obj. - equiZF :: c} \\
	& \quad \gamma_{\ell}^{\ul,\tZF}(\boldsymbol{\omega}, \mathbf{p})\geq\lambda_{\ell}^{\ul}, \; \forall  \ell \in \mathcal{L}, \label{eq: prob. bi-obj. - equiZF :: d} \\
				& \quad \lambda_{k}^{\dl} + 1 \geq \exp(\bar{R}_{k}^{\dl}), \; \forall  k \in \mathcal{K},   \label{eq: prob. bi-obj. - equiZF :: e} \\
		& \quad \lambda_{\ell}^{\ul} + 1 \geq \exp(\bar{R}_{\ell}^{\ul}), \; \forall  \ell \in \mathcal{L}, \label{eq: prob. bi-obj. - equiZF :: f} \\
		&\quad \eqref{eq: prob. general form bi-obj. trade-off :: c}, \eqref{eq: power constraint - ZF}, \label{eq: prob. bi-obj. - equiZF :: g} 
	\end{IEEEeqnarray}							
\end{subequations}\endgroup
where  $\tilde{F}_{\mathtt{SE}}\bigl(\boldsymbol{\Lambda}_{\dl},\boldsymbol{\Lambda}_{\ul}\bigr)\triangleq \ln|\mathbf{I} + \boldsymbol{\Lambda}_{\dl} | + \ln|\mathbf{I} + \boldsymbol{\Lambda}_{\ul} |$ is a concave function, with $ \boldsymbol{\Lambda}_{\dl}\triangleq\diag([\lambda_{1}^{\dl}\cdots\lambda_{K}^{\dl}]) $ and $ \boldsymbol{\Lambda}_{\ul}\triangleq\diag([\lambda_{1}^{\ul}\cdots\lambda_{L}^{\ul}]) $;  $\phi$ and $\boldsymbol{\lambda}\triangleq\{\boldsymbol{\lambda}_\dl,\boldsymbol{\lambda}_\ul\}$ with $ \boldsymbol{\lambda}_{\dl}\triangleq\{\lambda_{k}^{\dl}\}_{\forall k\in\mathcal{K}} $ and $ \boldsymbol{\lambda}_{\ul}\triangleq\{\lambda_{\ell}^{\dl}\}_{\forall \ell\in\mathcal{L}} $ are newly introduced variables. Here $\lambda_{k}^{\dl}$ and $\lambda_{\ell}^{\dl}$ can be viewed as  soft SINRs for $\DLU$ and $\ULU$, respectively.
\end{theorem}
\begin{proof}	Please see Appendix \ref{app: theorem 3}.\end{proof}

 In problem \eqref{eq: prob. bi-obj. - equiZF}, the nonconvex parts include \eqref{eq: prob. bi-obj. - equiZF :: b}-\eqref{eq: prob. bi-obj. - equiZF :: d}.
Following the spirit of the ICA method, we introduce a new variable $\xi\in\mathbb{R}_{++}$ to equivalently split constraint \eqref{eq: prob. bi-obj. - equiZF :: b}  into two constraints as
\begin{subnumcases}{\label{eq: convexity - Ptotal - equivAD.}
\eqref{eq: prob. bi-obj. - equiZF :: b} \Leftrightarrow} 
\bar{P}_{\mathtt{T}}(\mathcal{C}^\tZF,\boldsymbol{\mu}^{(\kappa)}) \leq  h_{\mathtt{qu}}(\xi) \triangleq \xi^2, \IEEEyessubnumber\label{eq: convexity - Ptotal - equivAD.a}\\
\xi^2\leq \phi(\eta\bar{P}_{\mathtt{T}}(\mathcal{C}^\tZF,\boldsymbol{\mu}^{(\kappa)})+1-\eta)\IEEEyessubnumber\label{eq: convexity - Ptotal - equivAD.b}.
\end{subnumcases}
We note that $\bar{P}_{\mathtt{T}}(\mathcal{C}^\tZF,\boldsymbol{\mu}^{(\kappa)})$ is a linear function in $(\boldsymbol{\omega}, \mathbf{p})$ due to \eqref{eq:relationshipWOmega}, and thus, \eqref{eq: convexity - Ptotal - equivAD.b} is a second order cone (SOC) representative \cite[Sec. 3.3]{Ben:2001}. Using \eqref{eq: quad. approx.}, the nonconvex constraint \eqref{eq: convexity - Ptotal - equivAD.a}  can be iteratively replaced by the following convex one
\begin{IEEEeqnarray}{rCl}\label{eq: convexity - Ptotal - equiv.}
\bar{P}_{\mathtt{T}}(\mathcal{C}^\tZF,\boldsymbol{\mu}^{(\kappa)}) &\leq&  h_{\mathtt{qu}}^{(\kappa)}(\xi).
\end{IEEEeqnarray}
Next, we can rewrite \eqref{eq: prob. bi-obj. - equiZF :: c} equivalently as
\begin{subnumcases}{\label{eq: prob. bi-obj. - equiZF :: cEq}
\eqref{eq: prob. bi-obj. - equiZF :: c} \Leftrightarrow} 
h_{\mathtt{fr}}(\sqrt{\omega_k},\psi_{k}^{\dl})\triangleq \omega_k/\psi_{k}^{\dl} \geq  \lambda_{k}^{\dl},\;\forall k\in\mathcal{K}, \IEEEyessubnumber\label{eq: prob. bi-obj. - equiZF :: cEq.a}\qquad\\
\psi_{k}^{\dl} \geq \frac{\|\mathbf{g}_{k}^{\mathtt{cci}}\mathbf{D}^{\ul}\|^2+\sigma_{k}^2}{|\mathbf{h}_{k}^\dl\mathbf{h}_{k}^{\tZF}|^2}, \forall k\in\mathcal{K}, \IEEEyessubnumber\label{eq: prob. bi-obj. - equiZF :: cEq.b}\qquad
\end{subnumcases}
where $\boldsymbol{\psi^{\dl}}\triangleq\{\psi_{k}^{\dl}\}_{\forall k\in\cK}$ are new variables. Since \eqref{eq: prob. bi-obj. - equiZF :: cEq.b} is a linear constraint, we focus on convexifying \eqref{eq: prob. bi-obj. - equiZF :: cEq.a} using \eqref{eq: frac. approx.} as
	\begin{align}\label{eq: SINR cons DL convex.}
	h_{\mathtt{fr}}^{(\kappa)}(\sqrt{\omega_k},\psi_{k}^{\dl}) & \geq \lambda_{k}^{\dl},\;\forall k\in\mathcal{K}.
	\end{align}
Similarly, with new variables $\boldsymbol{\psi^{\ul}}\triangleq\{\psi_{\ell}^{\ul}\}_{\forall \ell\in\cL}$, constraint \eqref{eq: prob. bi-obj. - equiZF :: d} is iteratively replaced by the following two convex ones
\begin{subequations} \label{eq: prob. bi-obj. - equiZF :: dConcex}
	\begin{align} 
	&h_{\mathtt{fr}}^{(\kappa)}(\sqrt{p_{\ell}},\psi_{\ell}^{\ul})  \geq \lambda_{\ell}^{\ul},\;\forall \ell\in\mathcal{L}, \label{eq: prob. bi-obj. - equiZF :: dConcex:a}\\
	&\psi_{\ell}^{\ul} \geq \frac{\|\mathbf{a}_{\ell}^{\mathtt{ZF}}\mathbf{\tilde{G}}^{\AtoA} \mathbf{W}^{\mathtt{ZF}}\|^2+\sigma_{\mathtt{AP}}^2\|\mathbf{a}_{\ell}^{\mathtt{ZF}}\|^2}{|\mathbf{a}_{\ell}^{\mathtt{ZF}}\mathbf{h}_{\ell}^{\ul}|^2}, \forall \ell\in\mathcal{L}.\label{eq: prob. bi-obj. - equiZF :: dConcex:b}\quad
	\end{align}
\end{subequations}

With the above discussions based on the ICA method,  we obtain the following approximate problem of \eqref{eq: prob. bi-obj. - equiZF} (and hence \eqref{eq: prob. bi-obj. - ZF}) with the convex set solved  at iteration $ (\kappa+1) $:
\begin{subequations} \label{eq: prob. bi-obj. - frac. prog.}
	\begin{IEEEeqnarray}{cl}
		\underset{\boldsymbol{\omega}, \mathbf{p},\boldsymbol{\lambda},\boldsymbol{\psi},\xi,\phi}{\max} &\quad  \frac{\tilde{F}_{\mathtt{SE}}\bigl(\boldsymbol{\Lambda}_{\dl},\boldsymbol{\Lambda}_{\ul}\bigr)}{\phi} \label{eq: prob. bi-obj. - frac. prog. :: a} \\
		\st&  \quad \eqref{eq: prob. general form bi-obj. trade-off :: c}, \eqref{eq: power constraint - ZF}, \eqref{eq: prob. bi-obj. - equiZF :: e}, \eqref{eq: prob. bi-obj. - equiZF :: f}, \nonumber\\
	&\quad \eqref{eq: convexity - Ptotal - equivAD.b},	  \eqref{eq: convexity - Ptotal - equiv.}, \eqref{eq: prob. bi-obj. - equiZF :: cEq.b}, \eqref{eq: SINR cons DL convex.}, \eqref{eq: prob. bi-obj. - equiZF :: dConcex},
		 \label{eq: prob. bi-obj. - frac. prog. :: b} \qquad
	\end{IEEEeqnarray}							
\end{subequations}
where $\boldsymbol{\psi}\triangleq\{\boldsymbol{\psi^{\dl}},\boldsymbol{\psi^{\ul}}\}.$ We can see that  the set of variables in \eqref{eq: prob. bi-obj. - frac. prog.} is independent of the numbers of APs and antennas, which differs from the original problem \eqref{eq: prob. general form bi-obj. trade-off}.
The objective \eqref{eq: prob. bi-obj. - frac. prog. :: a} is a concave-over-linear function, which can be efficiently addressed by the Dinkelbach transformation. In particular, we have
\begin{align} \label{eq: prob. bi-obj. - Dinkelbach}
\mathcal{V}^{\tZF,(\kappa+1)} = \underset{\mathcal{V}^\tZF\in\mathcal{F}^{(\kappa)}}{\argmax}\; \ddot{F}^{(\kappa)}\triangleq\tilde{F}_{\mathtt{SE}}\bigl(\boldsymbol{\Lambda}_{\dl},\boldsymbol{\Lambda}_{\ul}\bigr) - t^{(\kappa)}\phi,
\end{align}
where $ \mathcal{F}^{(\kappa)}\triangleq\bigl\{ \mathcal{V}^\tZF\triangleq\bigl\{\boldsymbol{\omega}, \mathbf{p},\boldsymbol{\lambda},\boldsymbol{\psi},\xi,\phi\bigr\}|\eqref{eq: prob. bi-obj. - frac. prog. :: b} \mbox{ holds}\bigr\} $ and $ t^{(\kappa)}\triangleq\frac{\tilde{F}_{\mathtt{SE}}\bigl(\boldsymbol{\Lambda}_{\dl}^{(\kappa)},\boldsymbol{\Lambda}_{\ul}^{(\kappa)}\bigr)}{\phi^{(\kappa)}} $.  To start the computational procedure, an initial feasible point $(\boldsymbol{\omega}^{(0)}, \mathbf{p}^{(0)},\boldsymbol{\psi}^{(0)},\xi^{(0)},\phi^{(0)})$ for \eqref{eq: prob. bi-obj. - Dinkelbach} is required. This is done by guaranteeing QoS constraints  \eqref{eq: prob. bi-obj. - equiZF :: e} and \eqref{eq: prob. bi-obj. - equiZF :: f} to be satisfied. Thus, we successively solve the following simplified problem of  \eqref{eq: prob. bi-obj. - frac. prog.}
\begingroup
\allowdisplaybreaks\begin{subequations} \label{eq: prob. bi-obj. - frac. prog. - init}
	\begin{IEEEeqnarray}{cl}
		\underset{\mathcal{V}^\tZF,\boldsymbol{\theta}}{\max} &\quad \Theta \triangleq  \sum\nolimits_{k\in\mathcal{K}}\theta_{k}^{\dl} + \sum\nolimits_{\ell\in\mathcal{L}}\theta_{\ell}^{\ul} \label{eq: prob. bi-obj. - frac. prog. - init :: a} \\
		\st &  \quad \lambda_{k}^{\dl} + 1 - \exp(\bar{R}_{k}^{\dl})  \geq \theta_{k}^{\dl}, \; \forall  k \in \mathcal{K}, \\
		& \quad \lambda_{\ell}^{\ul} + 1 - \exp(\bar{R}_{\ell}^{\ul})  \geq \theta_{\ell}^{\ul}, \; \forall  \ell \in \mathcal{L}, \\
		& \quad \theta_{k}^{\dl} \leq 0, \; \forall  k \in \mathcal{K},\;\; \theta_{\ell}^{\ul} \leq 0,  \forall  \ell \in \mathcal{L},\\
		&\quad \eqref{eq: prob. general form bi-obj. trade-off :: c}, \eqref{eq: power constraint - ZF},  \eqref{eq: convexity - Ptotal - equivAD.b},	  \eqref{eq: convexity - Ptotal - equiv.}, \eqref{eq: prob. bi-obj. - equiZF :: cEq.b}, \eqref{eq: SINR cons DL convex.}, \eqref{eq: prob. bi-obj. - equiZF :: dConcex}, \label{eq: prob. bi-obj. - frac. prog. - init :: b} \quad
	\end{IEEEeqnarray}							
\end{subequations}\endgroup
where $ \boldsymbol{\theta}\triangleq\{\theta_{k}^{\dl},\theta_{\ell}^{\ul}\}_{\forall k\in\mathcal{K}, \ell\in\mathcal{L}} $ are slack variables. The initial feasible point $(\boldsymbol{\omega}^{(0)}, \mathbf{p}^{(0)},\boldsymbol{\psi}^{(0)},\xi^{(0)})$ for \eqref{eq: prob. bi-obj. - Dinkelbach} is found  when the objective  \eqref{eq: prob. bi-obj. - frac. prog. - init :: a} is close to zero, and $\phi^{(0)}=1/\bar{P}(\mathcal{C}^{\mathtt{ZF},{(0)}},\boldsymbol{\mu}^{(0)})$. The proposed algorithm for solving the ZF-based SE-EE problem \eqref{eq: prob. general form bi-obj. trade-off} is summarized in Algorithm \ref{alg: ZFD problem}. 

\begin{algorithm}[t]
	\begin{algorithmic}[1]
		\fontsize{10}{11}\selectfont
		\protect\caption{Proposed Algorithm to Solve ZF-based SE-EE Problem \eqref{eq: prob. general form bi-obj. trade-off}}
		
		\label{alg: ZFD problem}
		
		\STATE \textbf{Initialization:} Compute ZF  precoder and receiver: $\mathbf{H}^{\tZF}$ and $\mathbf{A}^{\tZF}$.
		\STATE Set $ \ddot{F}^{(\kappa)} := -\infty $, $\kappa:=0$, and solve \eqref{eq: prob. bi-obj. - frac. prog. - init} to generate an initial feasible point $(\boldsymbol{\omega}^{(0)}, \mathbf{p}^{(0)},\boldsymbol{\psi}^{(0)},\xi^{(0)},\phi^{(0)})$.
		
		\REPEAT[Solving \eqref{eq: prob. bi-obj. - ZF}]
		\STATE Solve \eqref{eq: prob. bi-obj. - Dinkelbach} to obtain the optimal solution $(\boldsymbol{\omega}^{\star}, \mathbf{p}^{\star},\boldsymbol{\lambda}^{\star},\boldsymbol{\psi}^{\star},\xi^{\star},\phi^{\star})$ and $\ddot{F}^{(\kappa+1)}$.
		
		\STATE Update $(\boldsymbol{\omega}^{(\kappa+1)}, \mathbf{p}^{(\kappa+1)},\boldsymbol{\psi}^{(\kappa+1)},\xi^{(\kappa+1)},\phi^{(\kappa+1)}) :=(\boldsymbol{\omega}^{\star}, \mathbf{p}^{\star},\boldsymbol{\psi}^{\star},\xi^{\star},\phi^{\star})$.
		\STATE Set $ \kappa := \kappa + 1 $.
		\UNTIL Convergence
		
		\STATE Update $(\boldsymbol{\omega}^{*}, \mathbf{p}^{*}) := (\boldsymbol{\omega}^{(\kappa)}, \mathbf{p}^{(\kappa)} )$.
		
		\STATE Use \eqref{eq:relationshipWOmega} to recover $\mathbf{w}_{km},\;\forall k\in\cK, m\in\cM$.
		
		\STATE Compute $ \boldsymbol{\alpha}^{*} $ and $ \boldsymbol{\mu}^{*} $ as in \eqref{eq: compute alpha} and \eqref{eq: compute mu}, respectively.
		
		\STATE Repeat Steps 1-9 with fixed values of $ \boldsymbol{\alpha}^{*} $ and $ \boldsymbol{\mu}^{*} $ to find the exact solution of $(\mathbf{w}^{*}, \mathbf{p}^{*})$.
		
		\STATE Use $ (\mathbf{w}^{*}, \mathbf{p}^{*},\boldsymbol{\alpha}^{*},\boldsymbol{\mu}^{*}) $ to compute $ F(\mathbf{w}^{*}, \mathbf{p}^{*},\boldsymbol{\alpha}^{*},\boldsymbol{\mu}^{*}) $ in \eqref{eq: prob. general form bi-obj. trade-off :: a}.

		\STATE {\textbf{Output:}  The optimal solution $ (\mathbf{w}^{*}, \mathbf{p}^{*},\boldsymbol{\alpha}^{*},\boldsymbol{\mu}^{*}) $ and objective value $ F(\mathbf{w}^{*}, \mathbf{p}^{*},\boldsymbol{\alpha}^{*},\boldsymbol{\mu}^{*}) $.}
	\end{algorithmic} 
\end{algorithm}

\subsection{Convergence and Complexity Analysis}

\subsubsection{Convergence Analysis}
Algorithm~\ref{alg: ZFD problem} is mainly based on inner approximation and Dinkelbach transformation, where their convergences were proved in \cite{Marks:78} and \cite{Dinkelbach67}, respectively. Specifically, as provided in \cite{Dinkelbach67}, the optimal solution of problem \eqref{eq: prob. bi-obj. - Dinkelbach} is derived as a minorant obtained at each iteration of the ICA-based approximate problem \eqref{eq: prob. bi-obj. - frac. prog.}. From the properties of the ICA method \cite{Beck:JGO:10},  it follows that $ \mathcal{F}^{(\kappa)} \subset \mathcal{F}^{(\kappa+1)} $, resulting in a sequence $(\boldsymbol{\omega}^{(\kappa)}, \mathbf{p}^{(\kappa)},\boldsymbol{\psi}^{(\kappa)},\xi^{(\kappa)},\phi^{(\kappa)})$ of improved points of \eqref{eq: prob. bi-obj. - frac. prog.} and a sequence of non-decreasing objective values. Moreover, $ \mathcal{F}^{(\kappa)} $ is a convex connected set, as shown in \cite{Hieu:IEEETWC:June2019}. Therefore, Algorithm~\ref{alg: ZFD problem} is guaranteed to arrive at least at  a locally optimal solution for \eqref{eq: prob. bi-obj. - equiZF} (and hence \eqref{eq: prob. bi-obj. - ZF}) when $ \kappa \rightarrow \infty $,  satisfying the Karush-Kuhn-Tucker conditions according to \cite[Theorem 1]{Marks:78}.

\subsubsection{Computational Complexity}
Before deriving the complexity, we consider the following stages of Algorithm~\ref{alg: ZFD problem}:
\begin{itemize}
	\item The pre-processing stage computes  constant matrices, i.e., $ \mathbf{H}^{\mathtt{ZF}} $ and $ \mathbf{A}^{\mathtt{ZF}} $. This stage contributes a minor part to the total complexity since it only executes the matrix computation, which can be done easily. For the ZF design, it implies a computational complexity of $ \mathcal{O}(N^3) $ floating operations (flops).
	\item {\majrev The major complexity comes from the optimization of the involved variables. This is associated to the main loop in solving \eqref{eq: prob. bi-obj. - ZF} (i.e., Steps 3-7 in Algorithm \ref{alg: ZFD problem}), of which the per-iteration complexity is  $ \mathcal{O}\bigl(c^{2.5}(v^2+c)\bigr) $, with $ v=(3K+3L+2) $ scalar decision variables and $ c= (M+3K+4L+2)  $ linear/SOC constraints \cite{Hieu:IEEETWC:June2019}.}
\end{itemize}
It can be observed that the per-iteration complexity for the main loop is less dependent on $ M $, since problem \eqref{eq: prob. bi-obj. - Dinkelbach} only contains $ M $ linear constraints in \eqref{eq: power constraint - ZF}. Moreover, the size of $ \mathbf{H}^{\mathtt{ZF}} $ and $ \mathbf{A}^{\mathtt{ZF}} $ remains unchanged. Therefore, the complexity based on the proposed design is almost the same for different transmission strategies. {\majrev In Table~\ref{tab: Complexities}, we provide the major complexities of the proposed ZF and MRT/MRC, which are quite comparative. However, the execution time partially depends on the complexity of solving the successive approximate program in an iterative algorithm,  as well as the feasible region under the structure of constant matrices in the pre-processing stage. This will be further elaborated through  numerical examples.}

\begin{table*}[t]
	\centering
	\captionof{table}{Complexity Comparison.}
	\label{tab: Complexities}
%	\vspace{-10pt}
	\scalebox{0.9}{\majrev
		\begin{tabular}{l|c|c}
			\hline
			Transmission strategies & Pre-processing (flops) & Per-iteration complexity for optimization\\
			\hline\hline
			Proposed ZF-based design & $\mathcal{O}(N^3)$ & $ \mathcal{O}\bigl(c^{2.5}(v^2+c)\bigr) $ \\ \hline
			MRT/MRC &  $\mathcal{O}(N^2)$ & $ \mathcal{O}\bigl(c^{2.5}(v^2+c)\bigr) $ \\
			\hline		   				
	\end{tabular} }
%	\vspace{5pt}
	%\hrule
\end{table*}

\section{Proposed Solution Based on Improved Zero-Forcing}\label{sec:IZF}
From \eqref{eq:ZFSINIUL}, it can be seen that the IAI and RSI are still the main limitations  of FD CF-mMIMO. Thus, our next endeavor  is to propose an IZF-based design  to manage the network interference more effectively. In particular,  ONB-ZF with PCA, referred to as ONB-ZF-PCA, is developed for the DL transmission to mitigate the effects of IAI and RSI. In addition, we also adopt a ZF-SIC receiver  for UL reception to further accelerate the performance of the ZF-based design.
\subsection{IZF-Based Transmission Design}
\subsubsection{ONB-ZF-PCA-Based DL Transmission}
{\majrev The key idea of the ONB-ZF-PCA method is to utilize ONB-ZF for MUI cancellation and exploit PCA to depress the effects of IAI and RSI on UL transmission.  For  $\mathbf{W}^{\mathtt{IZF}}\triangleq[\mathbf{w}_{1}\;\dots\;\mathbf{w}_{K}]\in\mathbb{C}^{N\times K}$, we introduce the ONB-ZF-PCA procedure and its operation as follows.}
\begin{procedure}\label{procedure1}
	The ONB-ZF-PCA precoder is computed as
	\begin{align} \label{eq: PCA precoding matrix}
	\mathbf{W}^{\mathtt{IZF}} = \mathbf{P}\mathbf{Q}^H\mathbf{\tilde{T}}(\mathbf{D}^{\dl})^{\frac{1}{2}},
	\end{align}
	where $\mathbf{D}^{\dl}$ was already defined in \eqref{eq: power constraint - ZF}, and other matrix components $\mathbf{P}$, $\mathbf{Q}$, and $\mathbf{\tilde{T}}$  are determined  by the following steps:
	\begin{enumerate}
		\item Using the PCA method, we can express the covariance matrix of $ \mathbf{\tilde{G}}^{\AtoA} $ as
		\begin{align} \label{eq: SVD}
		(\mathbf{\tilde{G}}^{\AtoA})^H\mathbf{\tilde{G}}^{\AtoA} = \mathbf{U}\mathbf{E}\mathbf{U}^H,
		\end{align}
		where $ \mathbf{U} $ and $ \mathbf{E} $ are  unitary and diagonal matrices, respectively, which are derived by using singular value decomposition (SVD).
		\item Let $ \mathbf{e}\triangleq[e_1,\cdots,e_N]$, where $ e_1 \geq e_2 \geq \cdots \geq e_N \geq 0 $ are  eigenvalues on the diagonal of $ \mathbf{E} $. We define $ \bar{N} $ as
		\begin{align} \label{eq: N_bar top eigenvalues}
		\bar{N} \triangleq \min \Bigl\{ \bigl\{n\in \mathcal{N}\triangleq\{1,\cdots,N\}\bigl| f_{n}^{\mathtt{ER}}(\mathbf{e})\geq\delta \; % \nonumber\\
		\wedge \;n<N-1 \bigr\} \cup \{N-1\} \Bigr\},
		\end{align}
		where $ f_{n}^{\mathtt{ER}}(\mathbf{e})=\frac{\sum_{i=1}^{n}e_i}{\sum_{i=1}^{N}e_i} $ denotes the ratio  of the first $ n $ eigenvalues to the sum of all eigenvalues, and $ \delta $ is a percentage threshold for the sum of first $ \bar{N} $ eigenvalues over all eigenvalues with $ \delta =99\%$.
		\item \textbf{Compute} $ \mathbf{P}$: We compute $ \mathbf{P}=\mathbf{I}-\mathbf{\bar{U}}\mathbf{\bar{U}}^H $, where  $ \mathbf{\bar{U}} $ is generated from the first $ \bar{N} $ columns of $ \mathbf{U} $.
		\item \textbf{Compute}  $ \mathbf{Q} $: The economy-size LQ decomposition is applied to the compound matrix $\mathbf{H}^{\dl} \mathbf{P}$ such as  $\mathbf{H}^{\dl} \mathbf{P}= \mathbf{T}\mathbf{Q} $, where $ \mathbf{T}\in\mathbb{C}^{K\times K} $ is a lower triangular matrix and $ \mathbf{Q}\in\mathbb{C}^{K\times N} $ is an ONB matrix. Since $ K \ll N $ in CF-mMIMO, the economy-size decomposition can be exploited to reduce the computational complexity, leading to $\mathbf{Q}\mathbf{Q}^H=\mathbf{I}$ but $\mathbf{Q}^H\mathbf{Q}\neq\mathbf{I}$.
		\item \textbf{Compute} $ \mathbf{\tilde{T}}$: The entry at the $ i $-th row and $ j $-th column of $ \mathbf{\tilde{T}}\in\mathbb{C}^{K\times K} $, denoted by $\tilde{T}_{ij}$, is generally computed by using the following recursive expression:
		\begin{align}
		\tilde{T}_{ij} = \begin{cases}
		-\frac{1}{T_{ii}}\sum_{j'=j}^{i-1}T_{ij'}\tilde{T}_{j'j}, & \text{if } i>j,\\
		1, & \text{if } i=j,\\
		0, & \text{otherwise},
		\end{cases}
		\end{align}
		where $ T_{ij} $ denotes the entry at the $ i $-th row and $ j $-th column of $\mathbf{T}$ obtained in Step 4.
	\end{enumerate} 
\end{procedure}

\begin{proof}
Please see Appendix \ref{app: procedure1}.
\end{proof}
\begin{remark}\label{remark4}
	We note that the matrix $ \mathbf{P} $ computed via the PCA method aims at mitigating the effects of IAI and RSI. On the other hand, we can use the ZF precoder matrix based on ONB only (i.e., by skipping Steps 1-3). The LQ decomposition in Step 4 is applied to $\mathbf{H}^{\dl} $ instead of $\mathbf{H}^{\dl} \mathbf{P} $, i.e., $\mathbf{H}^{\dl}= \mathbf{T}\mathbf{Q} $. Then, a simpler precoder matrix can be constructed as
	\begin{align}\label{eq:ONB-ZF}
	\mathbf{W}^{\mathtt{ONB\text{-}ZF}} = \mathbf{Q}^H\mathbf{\tilde{T}}(\mathbf{D}^{\dl})^{\frac{1}{2}}.
	\end{align}
\end{remark}

Based on \textbf{Procedure \ref{procedure1}} and in the same manner as \eqref{eq: power constraint - ZF}, the power constraint at $\AP$ becomes
\begin{align} \label{eq: power constraint - DPC}
\tr\bigl(\mathbf{\tilde{T}}^H\mathbf{Q}\mathbf{P}\mathbf{B}_{m}\mathbf{P}^H\mathbf{Q}^H\mathbf{\tilde{T}}\mathbf{D}^{\dl}\bigr) \leq \mu_{m}^{(\kappa)} P_{\mathtt{AP}_m}^{\max}, \;\forall m\in\mathcal{M}.
\end{align}
The SINR at $ \DLUi{k} $ with the ONB-ZF-PCA precoder is
\begin{align}
\gamma_{k}^{\dl,\tIZF}(\boldsymbol{\omega}, \mathbf{p}) = \frac{\omega_k|\mathbf{h}_{k}^\dl\mathbf{h}_{k}^\tIZF|^2}{\|\mathbf{g}_{k}^{\mathtt{cci}}\mathbf{D}^{\ul}\|^2+\sigma_{k}^2},
\end{align}
where $ \mathbf{h}_{k}^\tIZF $ is the $ k $-th column of  $ \mathbf{H}^{\tIZF} \triangleq \mathbf{P}\mathbf{Q}^H\mathbf{\tilde{T}}$. 

\subsubsection{ZF-SIC-Based UL Transmission}
The decoded signals are successively removed before decoding the next signals, following the SIC principle \cite{Tse:book:05}.  Assuming that the decoding
order follows the UL UEs' index, i.e., $\ell = 1,2,\cdots, L,$  the $\mathtt{U}^{\mathtt{u}}_\ell$'s signal is decoded  by treating  signals of $\mathtt{U}^{\mathtt{u}}_{\ell'}$ for  $\ell' \geq \ell +1$ as noise, while other signals are  removed by SIC. The remaining MUI at $\ULU$ is further canceled by the ZF receiver. Thus, the ZF-SIC receiver for decoding $\ULU$'s signal can be expressed as 
 $ \mathbf{a}_{\ell}^{\tIZF} $, which is the first row of  $ \mathbf{A}_{\ell}^{\tIZF}\triangleq\bigl((\mathbf{\bar{H}}_{\ell}^\ul)^H\mathbf{\bar{H}}_{\ell}^\ul\bigr)^{-1}(\mathbf{\bar{H}}_{\ell}^\ul)^H \in\mathbb{C}^{(L-\ell+1)\times N} $, where $\mathbf{\bar{H}}_{\ell}^{\ul}\triangleq\bigl[\mathbf{h}_{\ell}^{\ul},\cdots, \mathbf{h}_{L}^{\ul}\bigr]\in\mathbb{C}^{N\times (L-\ell+1)} $ and $\mathbf{\bar{D}}^{\ul}\triangleq\diag\bigl(\bigl[\sqrt{p_{\ell}},\cdots, \sqrt{p_{L}}\bigr]\bigr)$. Accordingly, the SINR of $ \ULU $ with  ZF-SIC receiver becomes
\begin{align}
\gamma_{\ell}^{\ul,\tIZF}(\boldsymbol{\omega}, \mathbf{p}) = \frac{p_{\ell}|\mathbf{a}_{\ell}^{\tIZF}\mathbf{h}_{\ell}^{\ul}|^2 }{\|\mathbf{a}_{\ell}^{\tIZF}\mathbf{\tilde{G}}^{\AtoA} \mathbf{W}^{\mathtt{IZF}}\|^2+\sigma_{\mathtt{AP}}^2\|\mathbf{a}_{\ell}^{\tIZF}\|^2},
\end{align}
where $ \| \mathbf{a}_{\ell}^{\tIZF}\mathbf{\bar{H}}_{\ell+1}^{\ul} \mathbf{\bar{D}}_{\ell+1}^{\ul}\|^2\rightarrow0 $ due to the ZF-SIC matrix $\mathbf{A}_{\ell}^{\tIZF}$.

\subsection{IZF-Based Optimization Problem}
Similarly to problem \eqref{eq: prob. bi-obj. - ZF}, the IZF-based optimization problem can be expressed as 
\begingroup\allowdisplaybreaks\begin{subequations} \label{eq: prob. bi-obj. - IZF}
	\begin{IEEEeqnarray}{cl}
		\underset{\boldsymbol{\omega}, \mathbf{p}}{\max} &\quad  \eta \bar{F}_{\mathtt{SE}}\bigl(\boldsymbol{\Gamma}_{\dl}^{\mathtt{IZF}},\boldsymbol{\Gamma}_{\ul}^{\mathtt{IZF}}\bigr) %\nonumber\\ &\quad	
	+ (1-\eta) \bar{F}_{\mathtt{EE}}\bigl(\boldsymbol{\Gamma}_{\dl}^{\mathtt{IZF}},\boldsymbol{\Gamma}_{\ul}^{\mathtt{IZF}},\mathcal{C}^{\mathtt{IZF}},\boldsymbol{\mu}^{(\kappa)}\bigr) \label{eq: prob. bi-obj. - IZF :: a} \qquad\\
		\st	& \quad 
		R\bigl(\gamma_{k}^{\dl,\tIZF}(\boldsymbol{\omega}, \mathbf{p})\bigr) \geq \bar{R}_{k}^{\dl}, \; \forall  k \in \mathcal{K}, \label{eq: prob. bi-obj. - IZF :: c} \\
		&  \quad
		R\bigl(\gamma_{\ell}^{\ul,\tIZF}(\boldsymbol{\omega}, \mathbf{p})\bigr) \geq \bar{R}_{\ell}^{\ul}, \; \forall \ell \in \mathcal{L}, \label{eq: prob. bi-obj. - IZF :: d} \\
		&\quad \eqref{eq: prob. general form bi-obj. trade-off :: c}, \eqref{eq: power constraint - DPC}, \label{eq: prob. bi-obj. - IZF :: b} 
	\end{IEEEeqnarray}							
\end{subequations} \endgroup
where $ \boldsymbol{\Gamma_{\dl}^{\mathtt{IZF}}}\triangleq\{\gamma_{k}^{\dl,\tIZF}(\boldsymbol{\omega}, \mathbf{p})|\forall k\in\mathcal{K}\} $, $ \boldsymbol{\Gamma_{\ul}^{\mathtt{IZF}}}\triangleq\{\gamma_{\ell}^{\ul,\tIZF}(\boldsymbol{\omega}, \mathbf{p})|\forall \ell\in\mathcal{L}\} $ and $\mathcal{C}^{\mathtt{IZF}}\triangleq\{\boldsymbol{\omega}, \mathbf{p}\}$.  It can be seen that the nonconvex parts in problems \eqref{eq: prob. bi-obj. - ZF} and \eqref{eq: prob. bi-obj. - IZF} have a similar structure, and thus, we can slightly modify Algorithm \ref{alg: ZFD problem}  to solve \eqref{eq: prob. bi-obj. - IZF}. 

\section{Proposed Heap-Based Pilot Assignment Strategy}\label{sec:Pilo Assignment}
The developments presented in the previous sections are based on the assumption of perfect CSI to realize the potential performance of the proposed FD CF-mMIMO.  However,  it is of practical interest to take imperfect CSI into account. Each coherence interval in FD CF-mMIMO can be divided into two phases: UL training and data transmission in DL-UL. The coherence interval is short, and thus, each UE should practically be assigned a non-orthogonal
pilot sequence, resulting in the well-known  pilot contamination problem \cite{Marzetta:IEEETWC:Nov2010}.  Therefore, the main goal of this section is to develop  a pilot assignment algorithm based on the heap structure to reduce the effect of pilot contamination and training complexity. 
 We note that  the pilot assignment based on greedy method given in \cite{Ngo:TWC:Mar2017} not only requires high complexity due to the strategy of trial and error, but also is inapplicable to FD CF-mMIMO due to the additional channel estimation of CCI links. 
\begin{remark}
The channels of fading loop and IAI (i.e., $\mathbf{G}_{mm}^{\mathtt{SI}}$ and $\mathbf{G}_{mm'}^{\AtoA},\forall m\neq m'$) are assumed
to be the same as before. The reason for the fading loop channel is that the transmit and receive antennas are co-located at the APs. On the other hand, APs are generally fixed in a given area without mobility, and thus, the IAI channels can be perfectly acquired at the CPU at the initial deployment of the FD CF-mMIMO networks. 
\end{remark}

\subsection{Channel Estimation  and MSE Minimization Problem}
We assume that all UEs share the same orthogonal set of pilots, and the DL and UL UEs send  training sequences in different intervals to allow the channel estimation of CCI links. Let $ \tau<\min\{K,L\} $ be the length of pilot sequences. Then, the pilot set is defined as $ \boldsymbol{\Xi}\triangleq[ \boldsymbol{\varphi}_{1},\;\cdots\;,\boldsymbol{\varphi}_{\tau} ]\in\mathbb{C}^{\tau \times\tau} $, where $ \boldsymbol{\varphi}_i\in\mathbb{C}^{\tau\times 1}$ satisfies the orthogonality, i.e., $\boldsymbol{\varphi}_{i}^H\boldsymbol{\varphi}_{i'}=1$ if $i= i' \in \mathcal{T}_{\mathtt{p}}\triangleq\{1,\cdots,\tau \} $, and $\boldsymbol{\varphi}_{i}^H\boldsymbol{\varphi}_{i'}=0$, otherwise.
  We introduce the assignment variable $ \upsilon_{ij}\in\{0,1\} $ to determine whether the $ i $-th pilot sequence is assigned to the $ j $-th UE, with $ j\in\mathcal{T}_{\ul}\triangleq\{1,\cdots,U\} $ and  $ U\in\{K,L\} $. As a result, the pilot assigned to  UE $j$ can be expressed as	$ \boldsymbol{\bar{\varphi}}_j = \boldsymbol{\varphi}_i $ if $ \upsilon_{ij}=1 $. Let $ \boldsymbol{\bar{\Xi}}\triangleq[ \boldsymbol{\bar{\varphi}}_{1},\cdots,\boldsymbol{\bar{\varphi}}_{U} ]\in\mathbb{C}^{\tau \times U} $ be the pilot assignment matrix, such as $\boldsymbol{\bar{\Xi}} = \boldsymbol{\Xi}\boldsymbol{\Upsilon},$
where $ \boldsymbol{\Upsilon}\triangleq[\upsilon_{ij}]_{i\in\mathcal{T}_{\mathtt{p}},j\in\mathcal{T}_{\ul} }\in\mathbb{C}^{\tau\times U} $ following by the condition:
$\sum_{i\in\mathcal{T}_{\mathtt{p}} } \upsilon_{ij} \leq 1,\;\forall j\in\mathcal{T}_{\ul}.$

The training procedure for FD CF-mMIMO in TDD operation is executed in two phases. In the first phase, UL UEs send their pilot signals to APs to perform the channel estimation, and at the same time DL UEs also receive UL pilots to estimate CCI channels. In the second phase, DL UEs send their pilot signals along with the estimates of CCI links to APs. The training signals received at $ \AP $ can be written as $\mathbf{Y}_{m}^{\mathtt{tr}} = \sum_{j'\in\mathcal{T}_{\ul}} \sqrt{\tau p_{j'}^{\mathtt{tr}}}\boldsymbol{\bar{\varphi}}_{j'} \mathbf{h}_{mj'} + \mathbf{Z}_m,$
where $ \mathbf{h}_{mj}\in\{\mathbf{h}_{km}^{\dl}, (\mathbf{h}_{m\ell}^{\ul})^H \}\in\mathbb{C}^{1\times N_m} $, and $ p_{j}^{\mathtt{tr}} $ and $ \mathbf{Z}_{m}\sim \mathcal{CN}(0,\sigma^2_{\mathtt{AP}}\mathbf{I}) $ denote the UL training power of UE $j$ and the AWGN, respectively.  Using the linear MMSE (LMMSE) estimation \cite{KayBook93}, the channel estimate of $ \mathbf{h}_{mj} $ is given as
\begin{align} \label{eq: channel estimate}
\mathbf{\hat{h}}_{mj} & = \frac{\sqrt{\tau p_{j}^{\mathtt{tr}} }\beta_{mj}}{\sum_{j'\in\mathcal{T}_{\ul}} \tau p_{j'}^{\mathtt{tr}} \beta_{mj'} |\boldsymbol{\bar{\varphi}}_{j}^H\boldsymbol{\bar{\varphi}}_{j'}|^2 + \sigma^2_{\mathtt{AP}}} \boldsymbol{\bar{\varphi}}_{j}^H\mathbf{Y}_{m}^{\mathtt{tr}},
\end{align}
where $\beta_{mj}\in\{\beta_{km}^{\dl},\beta_{m\ell}^{\ul}\} $ is  the large-scale fading of the link between $ \AP $ and UE $j$.
We denote $\mathbf{\tilde{h}}_{mj}=\mathbf{h}_{mj}-\mathbf{\hat{h}}_{mj}$ as the channel estimation error, which is independent of $\mathbf{h}_{mj}$. The
elements of $\mathbf{\tilde{h}}_{mj}$ can be modeled as i.i.d. $\mathcal{CN}(0,\varepsilon_{mj})$ RVs, where 
\begin{IEEEeqnarray}{cl} \label{eq: MSE}
\varepsilon_{mj}  =   \beta_{mj}  \Bigl(1 - \frac{\tau p_{j}^{\mathtt{tr}} \beta_{mj}}{\sum_{j'\in\mathcal{T}_{\ul}} \tau p_{j'}^{\mathtt{tr}} \beta_{mj'} |\boldsymbol{\bar{\varphi}}_{j}^H\boldsymbol{\bar{\varphi}}_{j'}|^2 + \sigma^2_{\mathtt{AP}}} \Bigr).\quad\
\end{IEEEeqnarray}
 In an analogous fashion,  the channel estimate and  channel estimation error of CCI link $g_{k\ell}^{\mathtt{cci}}$ executed at $ \DLU$ are given as
	\begin{align} \label{eq: channel estimate CCI}
	\hat{g}_{k\ell}^{\mathtt{cci}}  = \frac{\sqrt{\tau p_{\ell}^{\mathtt{tr}} }\beta_{k\ell}^{\mathtt{cci}}}{\sum_{\ell'\in\mathcal{L}} \tau p_{\ell'}^{\mathtt{tr}} \beta_{k\ell'}^{\mathtt{cci}} |\boldsymbol{\bar{\varphi}}_{\ell}^H\boldsymbol{\bar{\varphi}}_{\ell'}|^2 + \sigma_{k}^2} \boldsymbol{\bar{\varphi}}_{\ell}^H\mathbf{y}_{k}^{\mathtt{tr},\mathtt{cci}}, 
	\end{align}
and $\tilde{g}_{k\ell}^{\mathtt{cci}}\sim\mathcal{CN}(0,\varepsilon_{k\ell}^{\mathtt{cci}})$, respectively, where	
	\begin{align}
	\varepsilon_{k\ell}^{\mathtt{cci}}  = \beta_{k\ell}^{\mathtt{cci}}\Bigl(1-\frac{\tau p_{\ell}^{\mathtt{tr}} \beta_{k\ell}^{\mathtt{cci}}}{\sum_{\ell'\in\mathcal{L}} \tau p_{\ell'}^{\mathtt{tr}} \beta_{k\ell'}^{\mathtt{cci}} |\boldsymbol{\bar{\varphi}}_{\ell}^H\boldsymbol{\bar{\varphi}}_{\ell'}|^2 + \sigma_{k}^2}\Bigr).
	\end{align}
Here, $\beta_{k\ell}^{\mathtt{cci}}$ denotes the large-scale fading of  CCI link  $\ULU \rightarrow\DLU$,  and   $ \mathbf{y}_{k}^{\mathtt{tr},\mathtt{cci}}=\sum_{\ell\in\mathcal{T}_{\ul}} \sqrt{\tau p_{j'}^{\mathtt{tr}}}\boldsymbol{\bar{\varphi}}_{\ell} g_{k\ell}^{\mathtt{cci}} + \mathbf{z}_{k} $, with $ \mathbf{z}_{k}\sim\mathcal{CN}(0,\sigma_{k}^2\mathbf{I}) $, is the UL UEs' training signals  received at $ \DLU $.

The CSI of CCI links is directly fed back to APs using a dedicated control channel to ensure a low-complexity  channel estimation at DL UEs. To mitigate the effects of pilot contamination, a  pilot assignment for the main DL and UL channels  is far more important that of CCI channels. Thus, we consider the following  MSE minimization problem:
\begingroup\allowdisplaybreaks\begin{subequations} \label{eq: prob. MSE}
	\begin{IEEEeqnarray}{cl}
		\underset{\boldsymbol{\Upsilon}}{\min} &\quad  \underset{j\in\mathcal{T}_{\ul}}{\max} \sum\nolimits_{m\in\mathcal{M}} \frac{N_{m} \varepsilon_{mj}}{\beta_{mj}} \label{eq: prob. MSE :: a} \\
		\st & \quad \upsilon_{ij}\in\{0,1\},\;\sum\nolimits_{i\in\mathcal{T}_{\mathtt{p}} } \upsilon_{ij} \leq 1,\;\forall i\in\mathcal{T}_{\mathtt{p}},\forall j\in\mathcal{T}_{\ul}.\label{eq: prob. MSE :: b} \qquad
	\end{IEEEeqnarray}							
\end{subequations}\endgroup

\subsection{Heap Structure-Based Pilot Assignment Strategy}
 Problem \eqref{eq: prob. MSE} is a min-max problem for  sum of fractional functions, for which it is hard to find an optimal solution.  For an efficient solution, we first introduce the following theorem.
\begin{theorem} \label{thm: MSE prob.}
	Problem  \eqref{eq: prob. MSE} can be solved via the following tractable problem:
		\begin{IEEEeqnarray}{cl}\label{eq: prob. MSE quad.}
			\underset{\boldsymbol{\Upsilon}}{\min} \quad  \underset{j\in\mathcal{T}_{\ul}}{\max}  \sum\nolimits_{j'\in\mathcal{T}_{\ul}}  \tilde{\beta}_{j'}\boldsymbol{\upsilon}_{j}^H\boldsymbol{\upsilon}_{j'},\ \st\quad \eqref{eq: prob. MSE :: b},\quad
		\end{IEEEeqnarray}							
	where $ \tilde{\beta}_{j'}\triangleq\sum_{m\in\mathcal{M}}N_{m}\tau p_{j'}^{\mathtt{tr}} \beta_{mj'} $. 
\end{theorem}
\begin{proof}
	Please see Appendix \ref{app: MSE prob.}.
\end{proof}
\begin{remark}
	For $ \bar{\upsilon}_{jj'}=\boldsymbol{\upsilon}_{j}^H\boldsymbol{\upsilon}_{j'} $, the objective function in \eqref{eq: prob. MSE quad.} becomes a kind of bottleneck assignment problem with the partial of  max function replaced by  $ \min\; \underset{j\in\mathcal{T}_{\ul}}{\max}  \sum_{j'\in\mathcal{T}_{\ul}}  \tilde{\beta}_{j'} \bar{\upsilon}_{jj'} $. This  indicates that all pilots play the same roles and the optimal solution for pilot assignment depends on how to cluster  UEs which share the same training sequence. \thmend
\end{remark}

Thus far, we have provided the tractable MSE minimization problem. We now propose the heap structure-based pilot assignment strategy. To do this,  the following definition is invoked.
\begin{definition}
	Min heap $ (\mathcal{H}^{\mathtt{min}}) $ is a tree-based structure, where  $ \mathcal{H}_{\mathtt{p}}^{\mathtt{min}} $ is a parent node of an arbitrary node $ \mathcal{H}_{\mathtt{c}}^{\mathtt{min}} $. Then, the key of $ \mathcal{H}_{\mathtt{p}}^{\mathtt{min}} $ is less than or equal to that of $ \mathcal{H}_{\mathtt{c}}^{\mathtt{min}} $. In a max heap $ (\mathcal{H}^{\mathtt{max}}) $, the key of $ \mathcal{H}_{\mathtt{p}}^{\mathtt{max}} $ is greater than or equal to that of $ \mathcal{H}_{\mathtt{c}}^{\mathtt{max}} $ \cite{Cormen:2001:IA:500824}. We note that a node of  heap contains not only the keys (values) to build a heap structure, but also other specific parameters depending on storage purposes. Therefore, if a node is moved on the heap, all its parameters are also moved in company.
\end{definition}

Let $ \mathcal{H}\in\{ \mathcal{H}^{\mathtt{min}}, \mathcal{H}^{\mathtt{max}} \} $, the following main operations are involved:
\begin{itemize}
	\item \textbf{Generate a heap}  $\bigl( \mathcal{G}(\mathbf{x},\{\mathbf{y}\}) \rightarrowtail \mathcal{H} \bigl)$: The size and keys of $ \mathcal{H} $ follow the size and values of vector $ \mathbf{x} $, where $ \{\mathbf{y}\} $ is a set of parameters.
	\item \textbf{Find min/max value} $\bigl( \mathcal{H} \rightarrow (x,\{\mathbf{y}\}) \bigl)$: To return the root key $ x $ and the parameter set $ \{\mathbf{y}\} $ of $ \mathcal{H} $.
	\item \textbf{Extract the root node} $\bigl( \mathcal{H} \vdash (x,\{\mathbf{y}\}) \bigl)$: To pop the root node out of $ \mathcal{H} $ (i.e., extract the max/min value), and then assign to $ (x,\{\mathbf{y}\}) $. Next, $ \mathcal{H} $ is updated to restore the heap condition.
	\item \textbf{Replace and Sift-down} $\bigl( \mathcal{H} \dashv (x,\{\mathbf{y}\}\bigr) $: To replace the root node with the key $ x $ and its parameter set $ \{\mathbf{y}\} $, and then,  move the root node down in the tree to restore the heap condition.
\end{itemize}

{\majrev Let $\mathbf{u}_{i}^{\mathtt{ID}}$ be the $ i $-th column of the identity matrix of size $\tau$ ($\mathbf{I}_{\tau}$). Then, the feasible set of the $ j $-th column variable of matrix $ \boldsymbol{\Upsilon} $, corresponding to UE $ j $, is determined as $ \boldsymbol{\upsilon}_j \in \{\mathbf{u}_{i}^{\mathtt{ID}}|i\in \mathcal{T}_{\mathtt{p}}\},\ j\in \mathcal{T}_{\mathtt{u}} $, satisfying constraint \eqref{eq: prob. MSE :: b}.
From Theorem \ref{thm: MSE prob.}, if the $ i $-th pilot is assigned to two arbitrary UEs $j, k \in \mathcal{T}_{\mathtt{u}} $, we have
\begin{align}
\sum_{j'\in\mathcal{T}_{\ul}}  \tilde{\beta}_{j'}\boldsymbol{\upsilon}_{j}^H\boldsymbol{\upsilon}_{j'}=\sum_{j'\in\mathcal{T}_{\ul}}  \tilde{\beta}_{j'}\boldsymbol{\upsilon}_{k}^H\boldsymbol{\upsilon}_{j'}=\sum_{j'\in\mathcal{T}_{\ul}}  \tilde{\beta}_{j'}(\mathbf{u}_{i}^{\mathtt{ID}})^H\boldsymbol{\upsilon}_{j'}. \nonumber
\end{align}
We define a pilot-reused coefficient (PRC) of the $ i $-th pilot by
$ \bar{\beta}_i \triangleq \sum_{j'\in\mathcal{T}_{\ul}}  \tilde{\beta}_{j'}(\mathbf{u}_{i}^{\mathtt{ID}})^H\boldsymbol{\upsilon}_{j'}, $ and rewrite \eqref{eq: prob. MSE quad.} equivalently as
\setcounter{equation}{51}
\begin{subequations}
	\begin{align}\label{eq: prob. MSE quad. 1}
	\underset{\boldsymbol{\Upsilon}}{\min} &\quad  \underset{i\in\mathcal{T}_{\mathtt{p}}}{\max} \; \{ \bar{\beta}_i \}, \\  \st &\quad \boldsymbol{\upsilon}_j \in \{\mathbf{u}_{i}^{\mathtt{ID}}|i\in \mathcal{T}_{\mathtt{p}}\},\ j\in \mathcal{T}_{\mathtt{u}}.
	\end{align}
\end{subequations}
It is realized that if the $ i $-th pilot is assigned to UE $ j $, the PRC of the $ i $-th pilot increases by a factor of $\tilde{\beta}_j$. To minimize the maximum of PRCs, a heuristic assignment is executed such that the pilot with the smallest PRC is assigned to UE $ j $ with the largest $ \tilde{\beta}_j $. The following example is to illustrate the procedure of the proposed heap-based pilot assignment.

\textbf{Example:} we consider a scenario where  $ \tau=4 $ pilots need to be assigned to $ U=10 $ UEs with a given large-scale fading as: $[\tilde{\beta}_1, \cdots, \tilde{\beta}_{10}]=[0.0107, 0.0881, 0.1384, 0.0309, 0.0798, 0.0531, 0.0130, \allowbreak 0.0765, 0.0109, \allowbreak 0.0102]$. The assignment progress is described in Table \ref{tab: heap states}.}

{\majrev \vspace{40pt}
\begin{spacing}{1.1}
\begin{longtable}{p{1cm}  c  c  P{6.8cm}}
	\caption{States of heap structures per iteration} \vspace{-10pt} \\
	\toprule
	\#Iter. &  Min-heap $ (\mathcal{H}^{\mathtt{min}}) $ & Max-heap $ (\mathcal{H}^{\mathtt{max}}) $ & \hspace{1.5cm}Processing \\
	\cmidrule(lr){1-4} \cmidrule(lr){1-4}
	%\cmidrule(lr){2-2}\cmidrule(lr){3-3}\cmidrule(l){3-4}
	\raisebox{-100pt}{0}	 &  \raisebox{-1\totalheight}{\includegraphics[width=0.18\columnwidth]{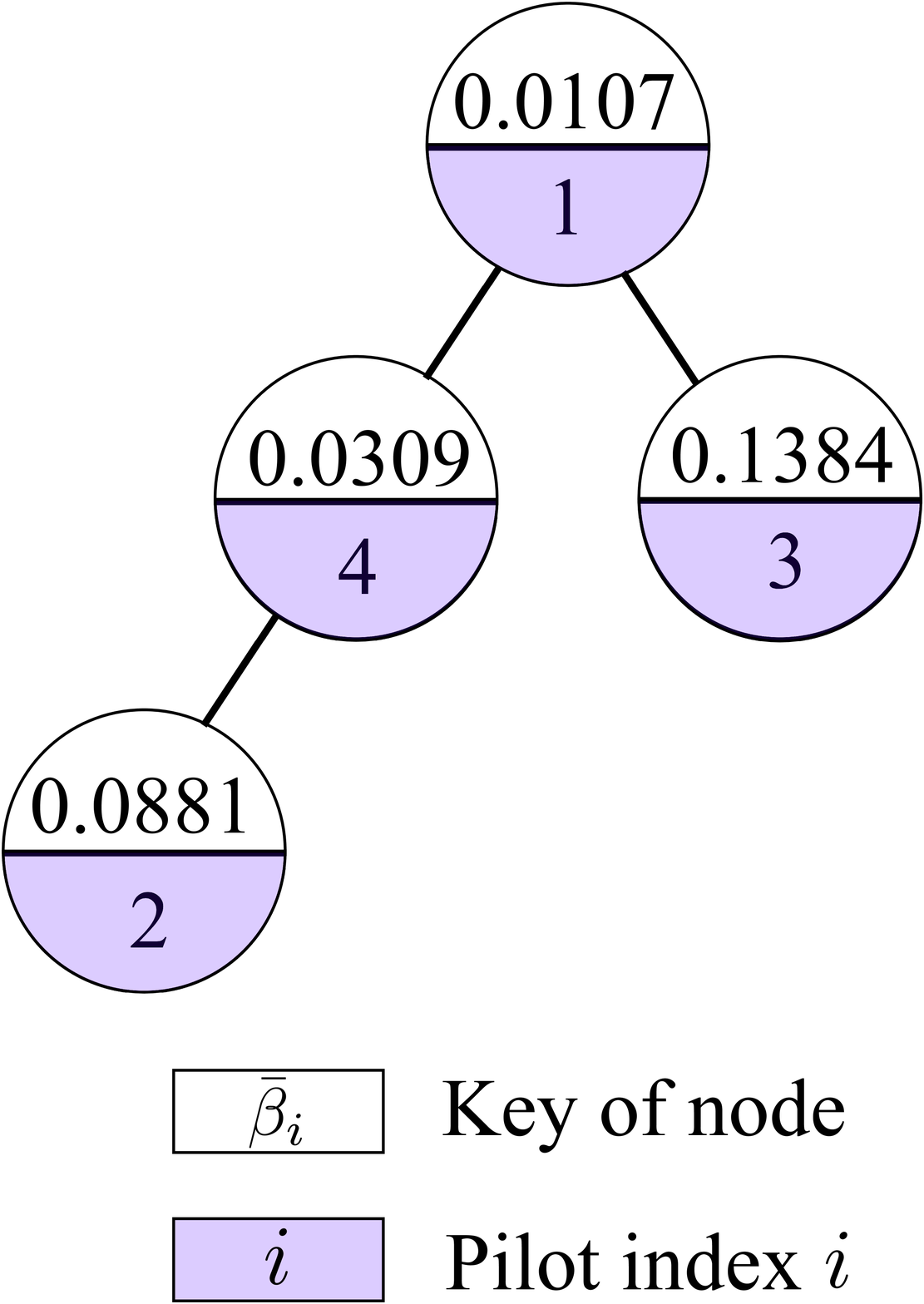}}
	& \raisebox{-1\totalheight}{\includegraphics[width=0.216\columnwidth]{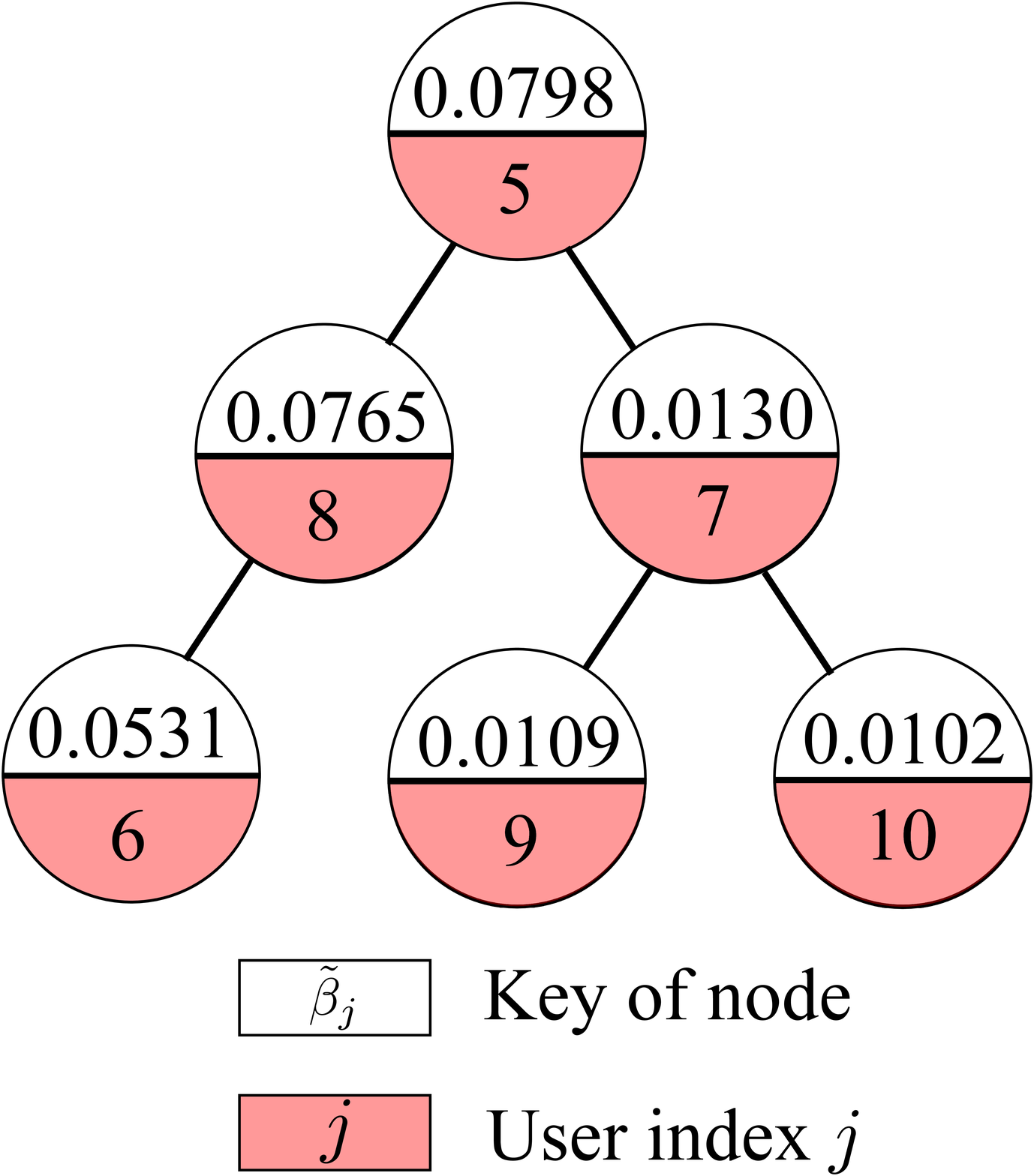}}	& 	Initial state: 
	\begin{itemize}[noitemsep,topsep=0pt,leftmargin=*]
		\item $ \tau $ pilots are assigned to the first $ \tau $ UEs, leading to $ \mathcal{H}^{\mathtt{min}} $ with $ \tau $ nodes
		\item $ U-\tau $ remaining UEs are put into $ \mathcal{H}^{\mathtt{max}} $
	\end{itemize}
	\vspace{5pt} Execution:
	\begin{itemize}[noitemsep,topsep=0pt,leftmargin=*,partopsep=0pt,parsep=0pt]
		\item Pop the root node of $ \mathcal{H}^{\mathtt{max}} $, and assign pilot $ 1 $ to UE $ 5 $
		\item Compute $\bar{\beta}_1=0.0107+0.0798=0.0905$
		\item Update $ \mathcal{H}^{\mathtt{min}} $ and $ \mathcal{H}^{\mathtt{max}} $
	\end{itemize}
	\\ 
%	\cmidrule(lr){1-4}
	\cmidrule(r){1-1}\cmidrule(lr){2-2}\cmidrule(lr){3-3}\cmidrule(l){4-4}
	\raisebox{-50pt}{1} & \raisebox{-0.85\totalheight}{\includegraphics[width=0.18\columnwidth]{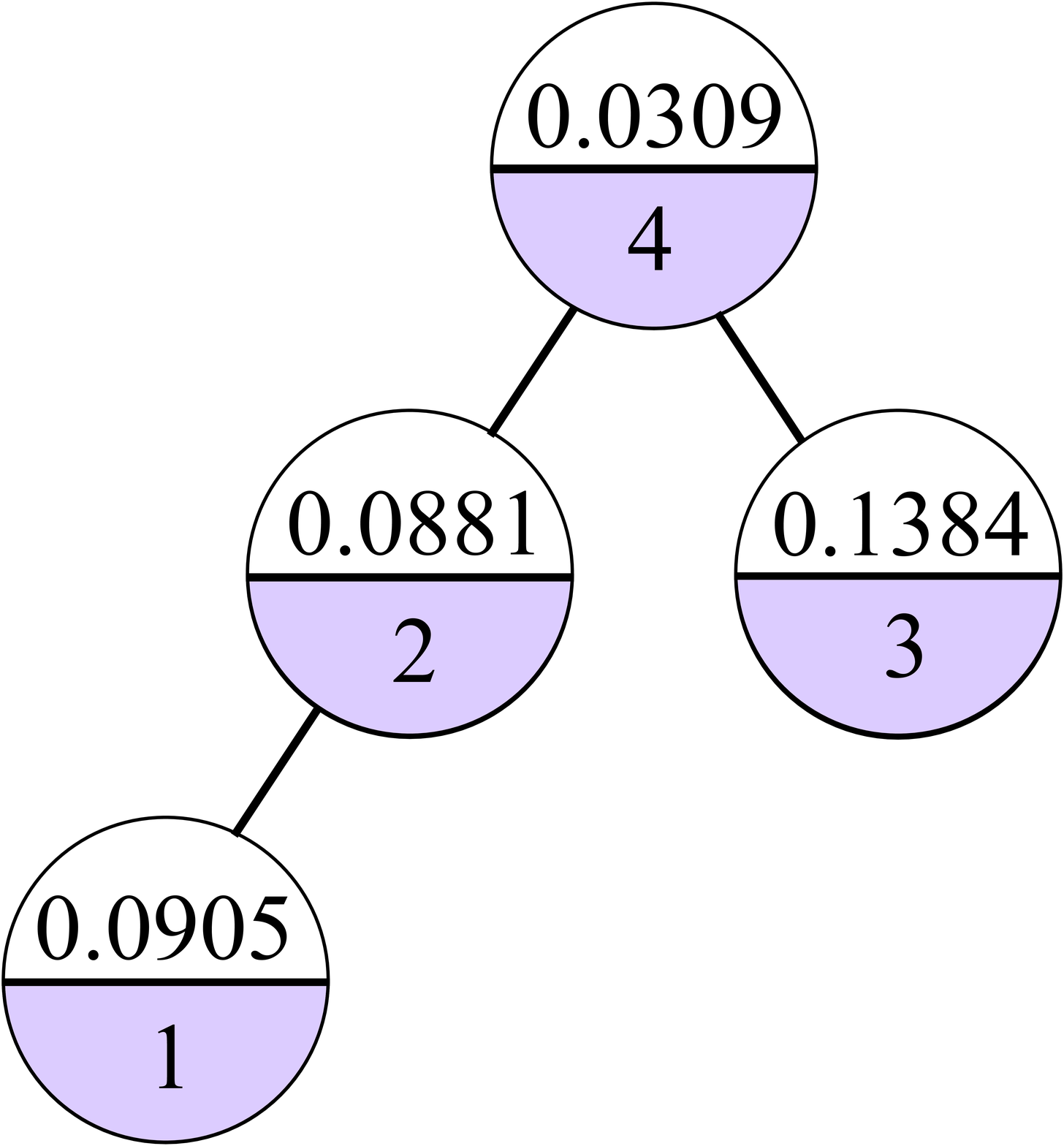}}
	& \raisebox{-0.85\totalheight}{\includegraphics[width=0.18\columnwidth]{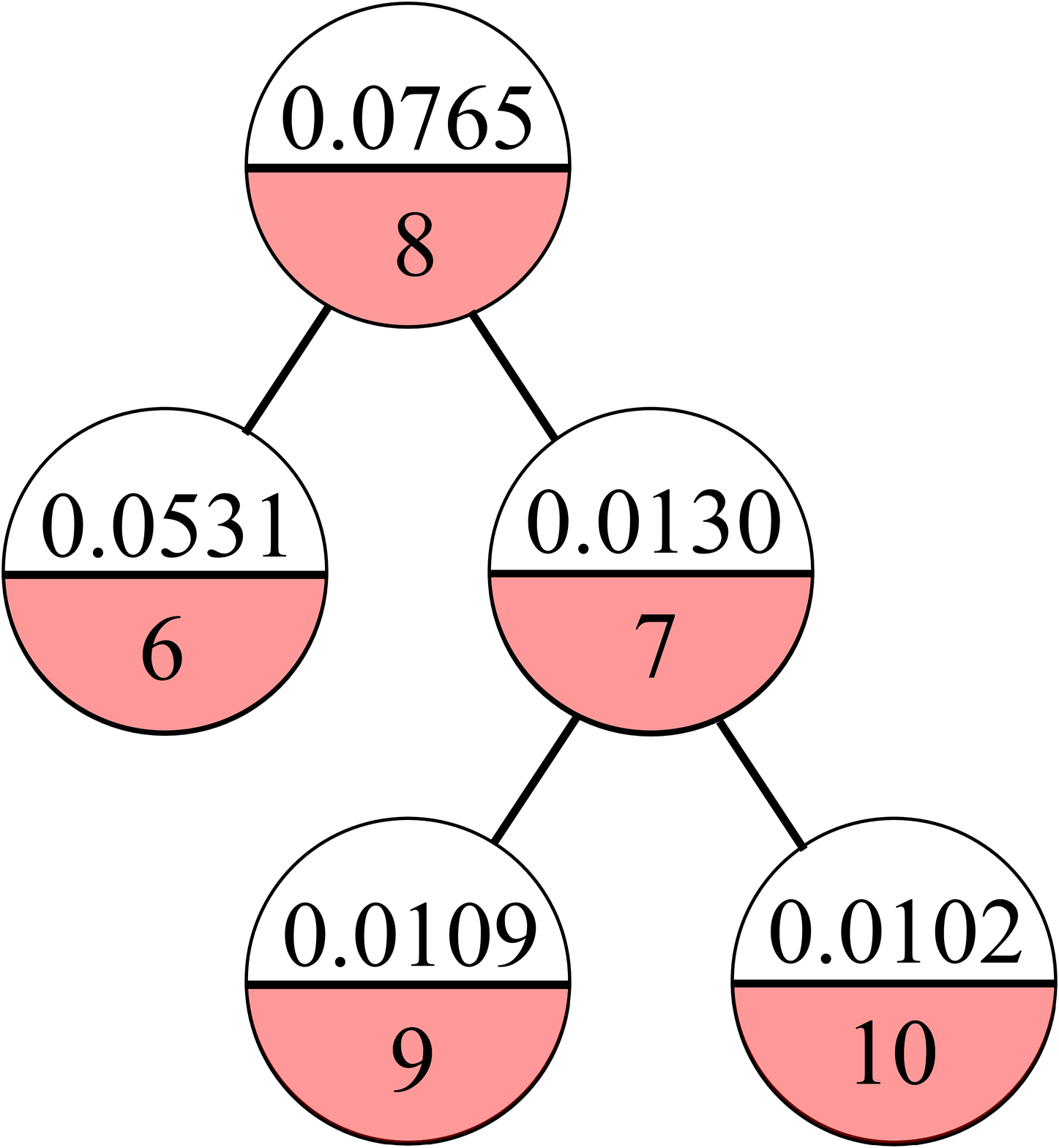}}	& 
	\begin{itemize}[noitemsep,topsep=0pt,leftmargin=*]
		\item Pop the root node of $ \mathcal{H}^{\mathtt{max}} $ and assign pilot $ 4 $ to UE $ 8 $
		\item Compute $\bar{\beta}_4=0.0309+0.0765=0.1074$
		\item Update $ \mathcal{H}^{\mathtt{min}} $ and $ \mathcal{H}^{\mathtt{max}} $
	\end{itemize} \\ %\cmidrule(lr){1-4}
	\cmidrule(r){1-1}\cmidrule(lr){2-2}\cmidrule(lr){3-3}\cmidrule(l){4-4}
	\raisebox{-40pt}{2} & \raisebox{-0.8\totalheight}{\includegraphics[width=0.18\columnwidth]{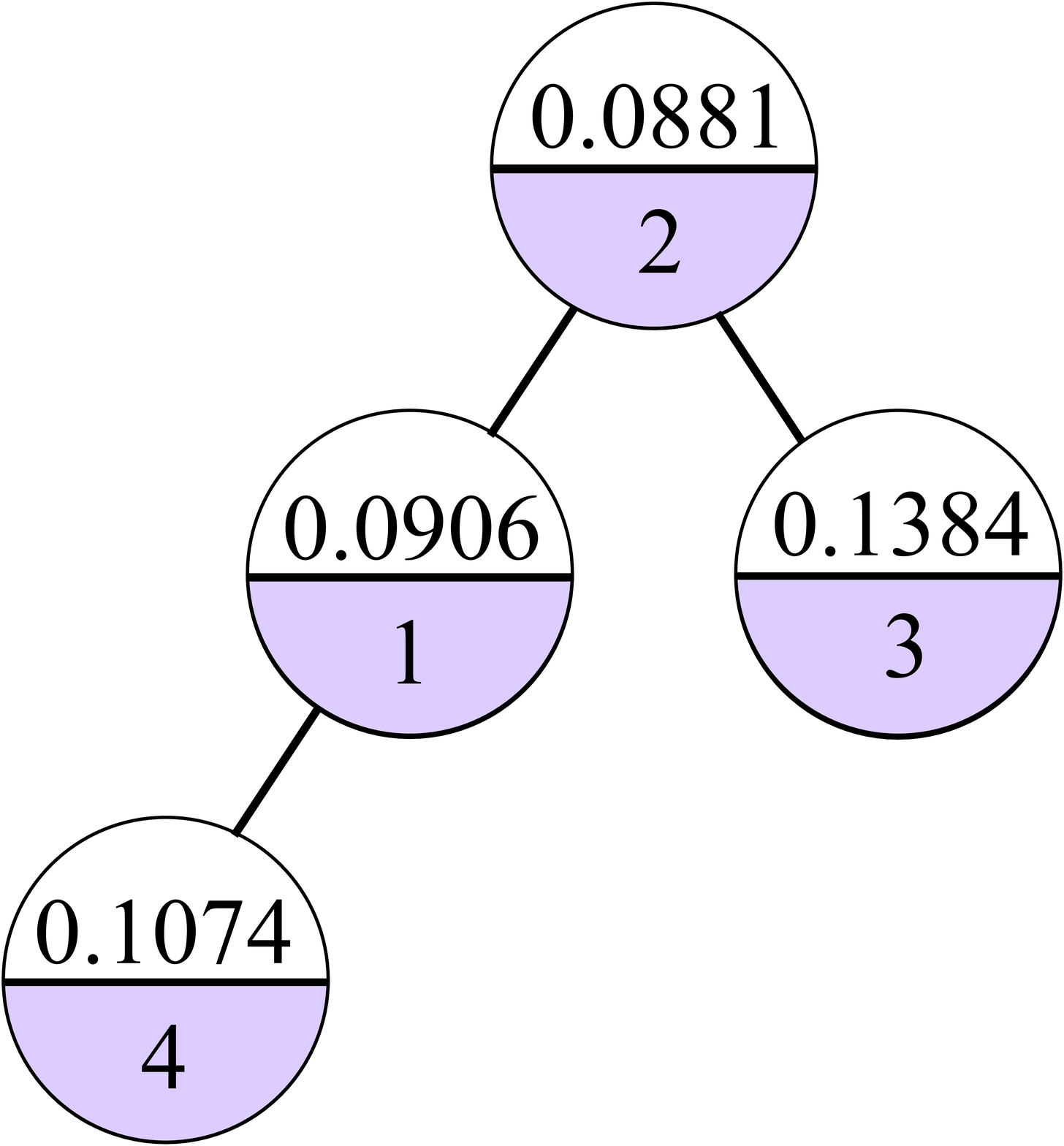}}
	& \raisebox{-0.8\totalheight}{\includegraphics[width=0.144\columnwidth]{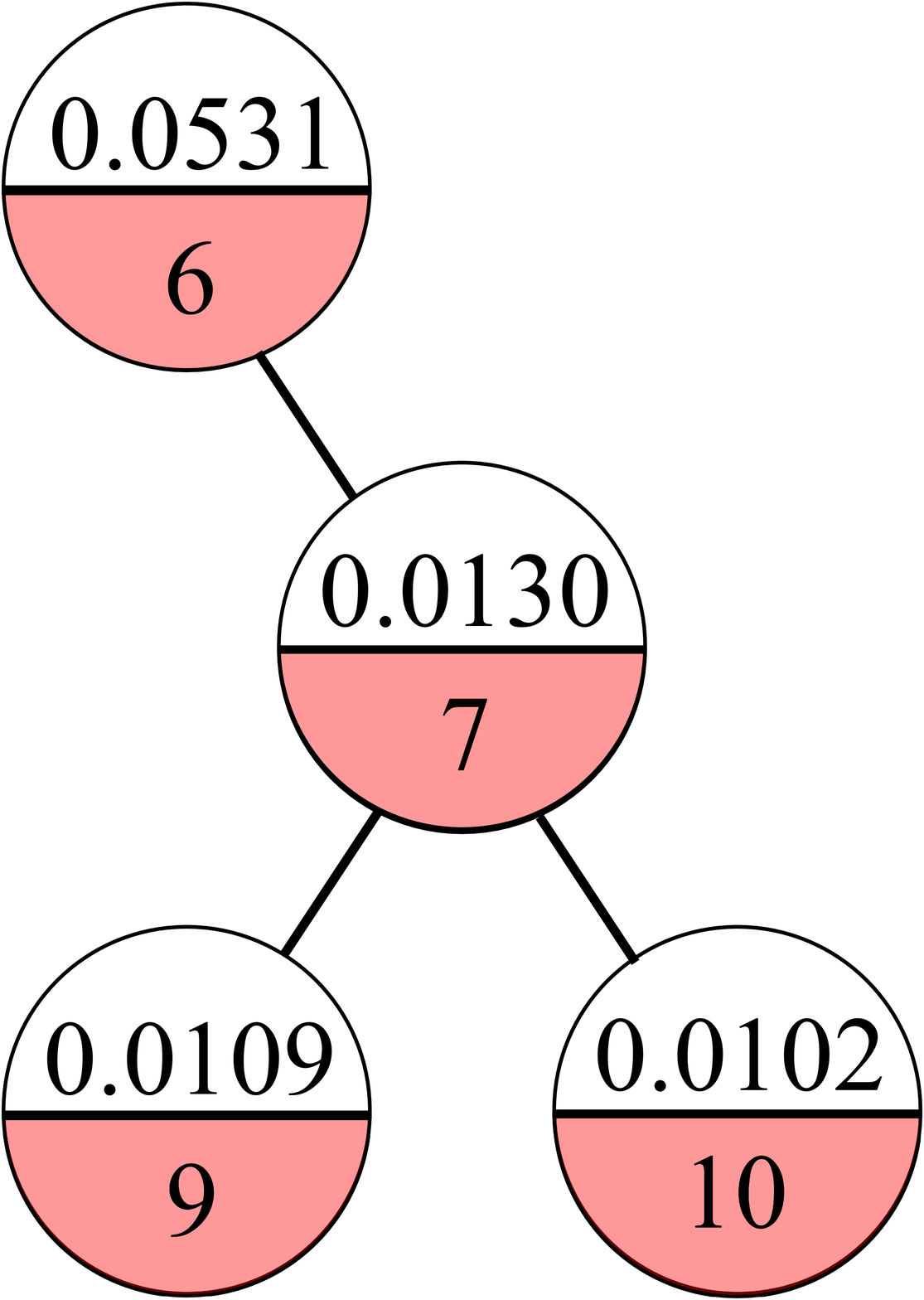}}	& \begin{itemize}[noitemsep,topsep=0pt,leftmargin=*]
		\item Pop the root node of $ \mathcal{H}^{\mathtt{max}} $ and assign pilot $ 2 $ to UE $ 6 $
		\item Compute $\bar{\beta}_2=0.0881+0.0531=0.1412$
		\item Update $ \mathcal{H}^{\mathtt{min}} $ and $ \mathcal{H}^{\mathtt{max}} $
	\end{itemize} \\ %\cmidrule(lr){1-4}
	\cmidrule(r){1-1}\cmidrule(lr){2-2}\cmidrule(lr){3-3}\cmidrule(l){4-4}
	\raisebox{-40pt}{3} & \raisebox{-0.8\totalheight}{\includegraphics[width=0.18\columnwidth]{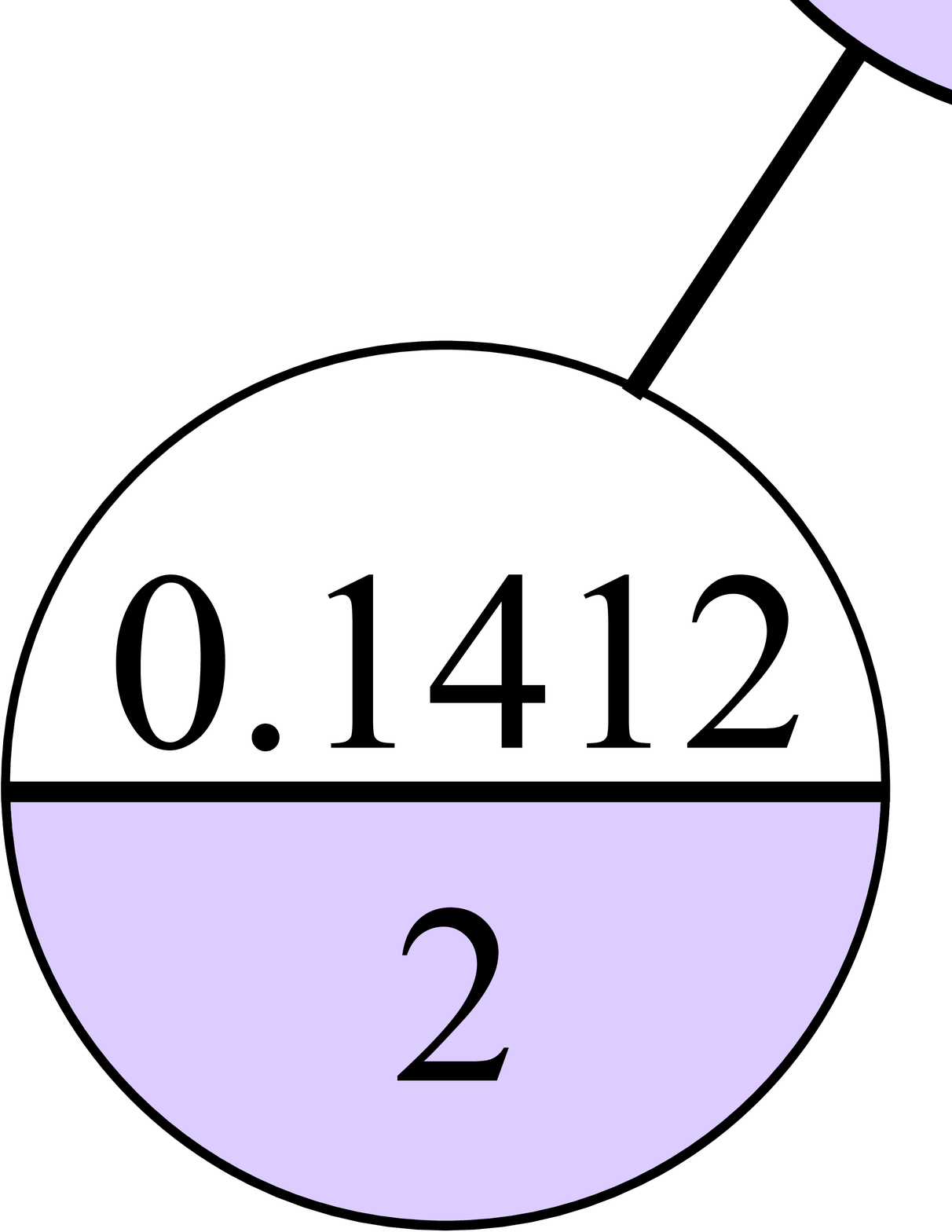}}
	& \raisebox{-0.8\totalheight}{\includegraphics[width=0.144\columnwidth]{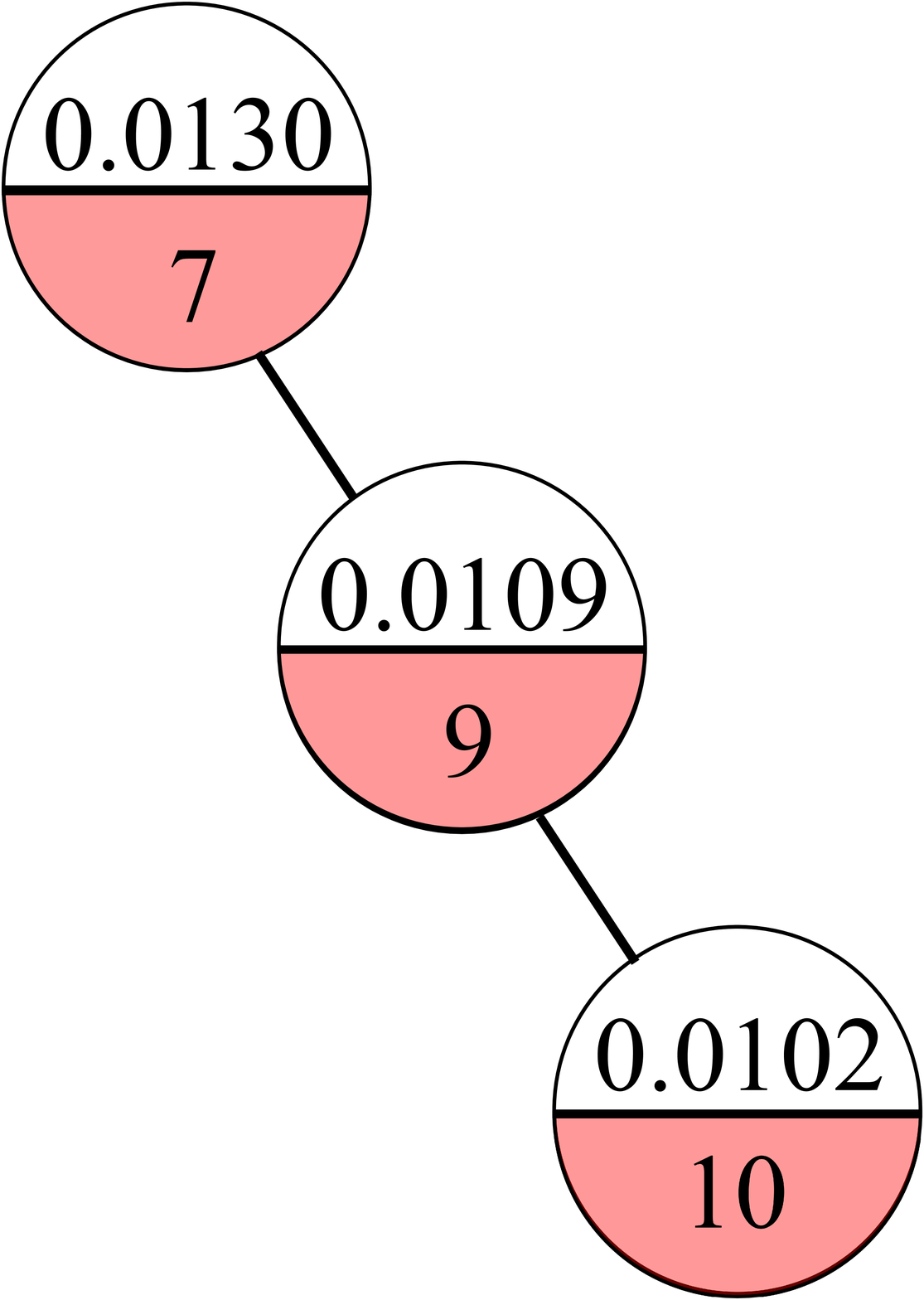}}	& \begin{itemize}[noitemsep,topsep=0pt,leftmargin=*]
		\item Pop the root node of $ \mathcal{H}^{\mathtt{max}} $ and assign pilot $ 1 $ to UE $ 7 $
		\item Compute $\bar{\beta}_1=0.0906+0.0130=0.1036$
		\item Update $ \mathcal{H}^{\mathtt{min}} $ and $ \mathcal{H}^{\mathtt{max}} $
	\end{itemize} \\ %\cmidrule(lr){1-4}
	\cmidrule(r){1-1}\cmidrule(lr){2-2}\cmidrule(lr){3-3}\cmidrule(l){4-4}
	4 & \raisebox{-0.5\totalheight}{\includegraphics[width=0.2\columnwidth]{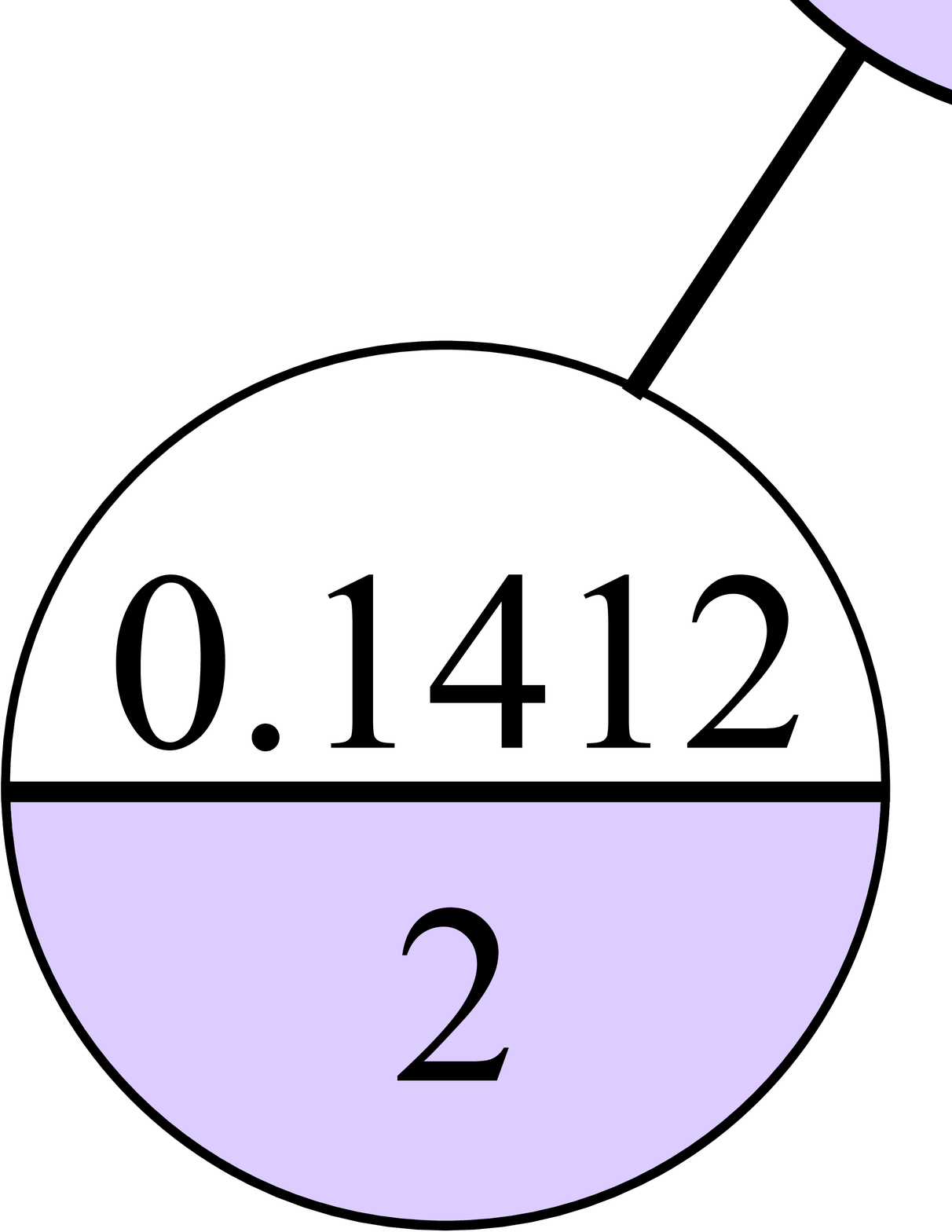}}
	& \raisebox{-0.25\totalheight}{\includegraphics[width=0.112\columnwidth]{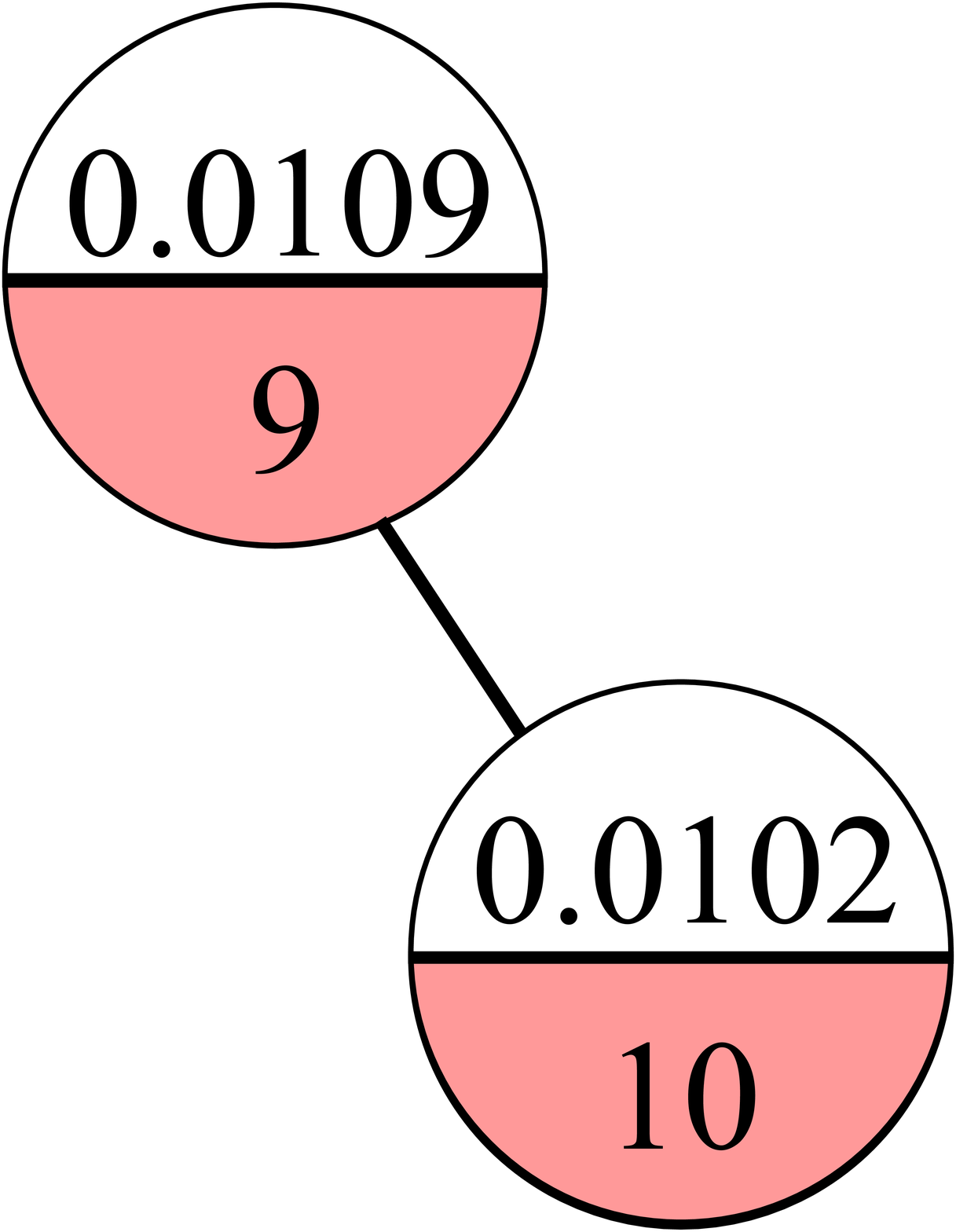}}	& \vspace{-40pt} \begin{itemize}[noitemsep,topsep=0pt,leftmargin=*]
		\item Pop the root node of $ \mathcal{H}^{\mathtt{max}} $ and assign pilot $ 1 $ to UE $ 9 $
		\item Compute $\bar{\beta}_1=0.1036+0.0109=0.1145$
		\item Update $ \mathcal{H}^{\mathtt{min}} $ and $ \mathcal{H}^{\mathtt{max}} $
	\end{itemize} \\ %\cmidrule(lr){1-4}
	\cmidrule(r){1-1}\cmidrule(lr){2-2}\cmidrule(lr){3-3}\cmidrule(l){4-4}
	5 & \raisebox{-0.5\totalheight}{\includegraphics[width=0.2\columnwidth]{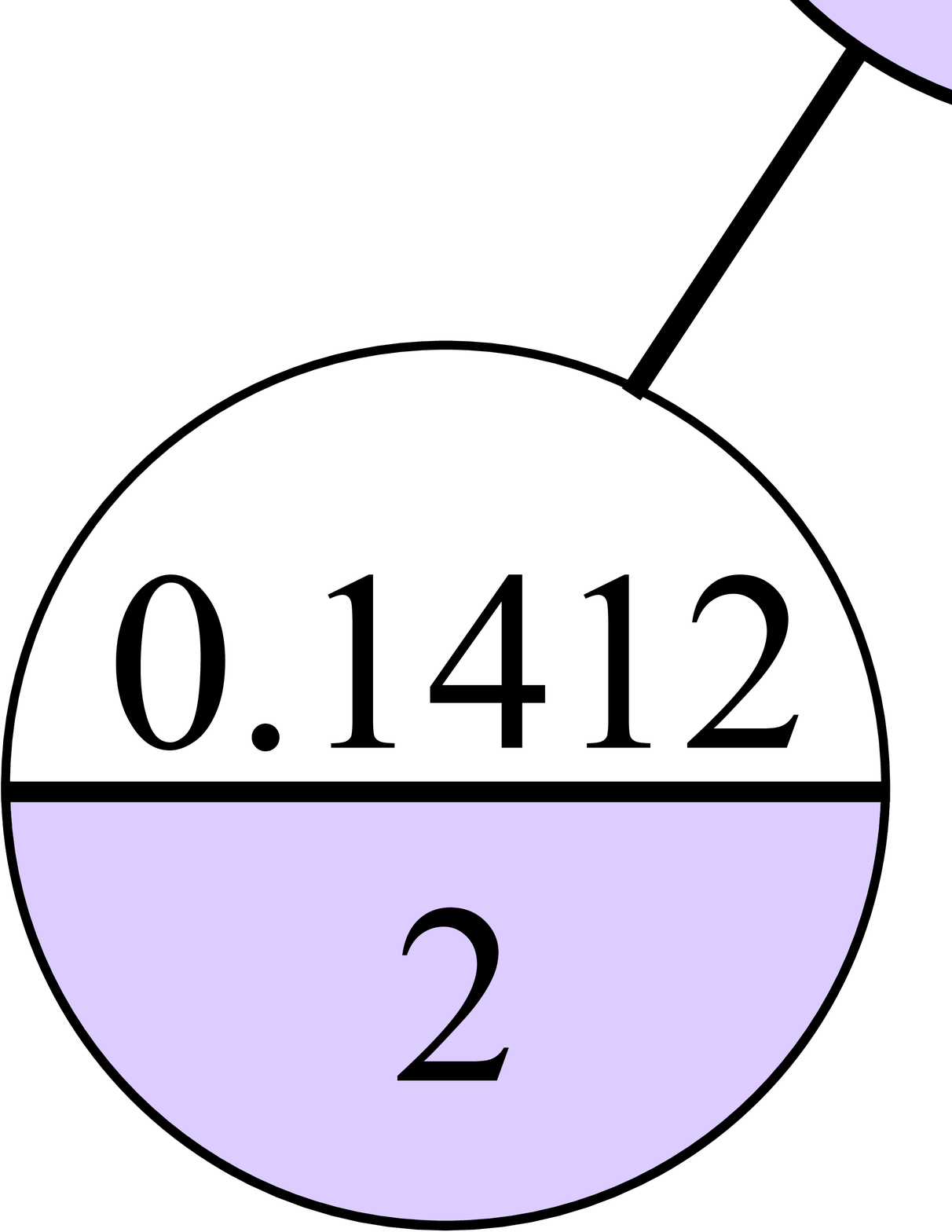}}
	& \raisebox{0.6\totalheight}{\includegraphics[width=0.0656\columnwidth]{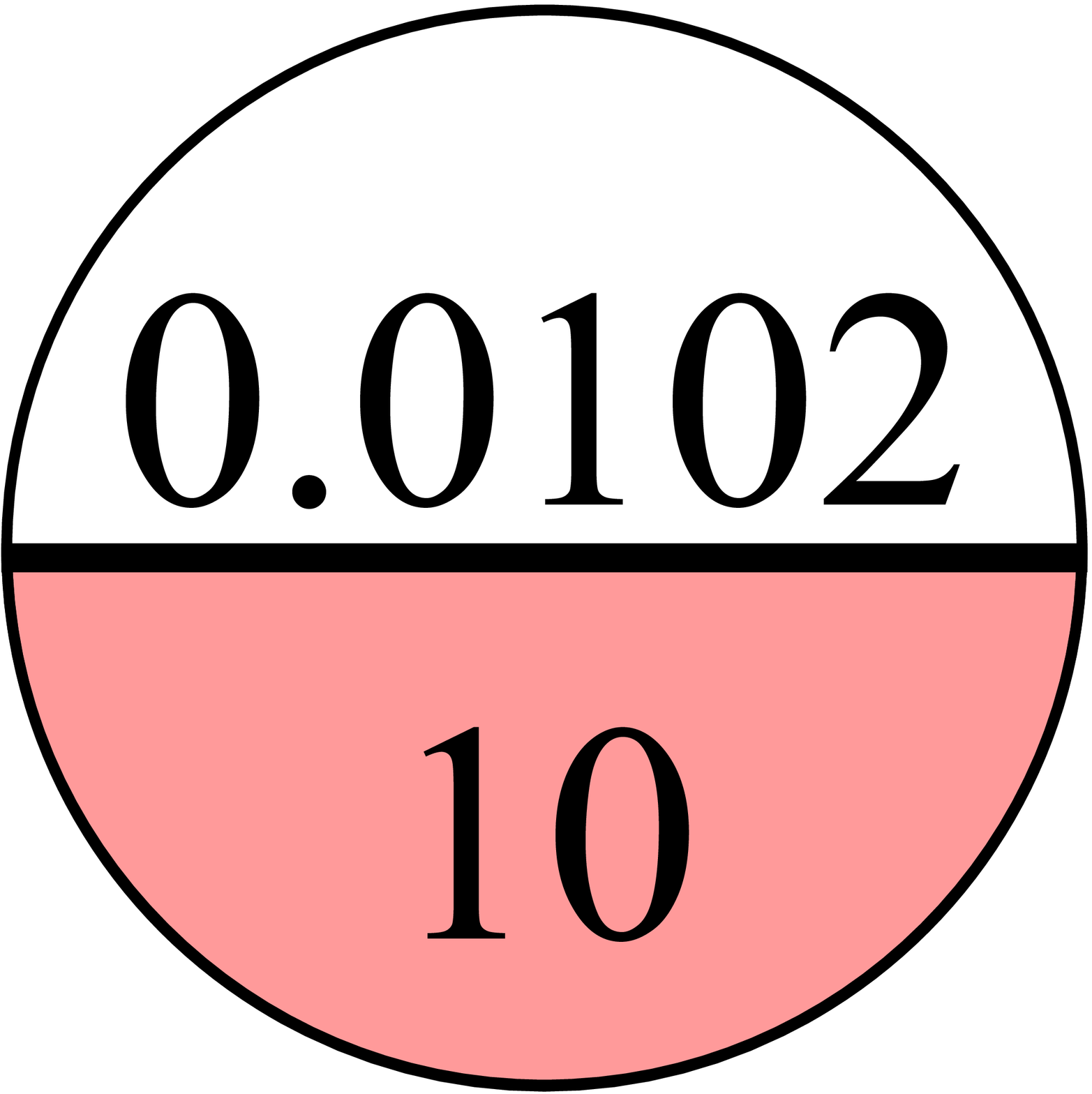}}	& \vspace{-30pt}
	\begin{itemize}[noitemsep,topsep=0pt,leftmargin=*]
		\item Pop the root node of $ \mathcal{H}^{\mathtt{max}} $ and assign pilot $ 4 $ to UE $ 10 $ which complete the pilot assignment
	\end{itemize} \\ 
	\bottomrule
	\vspace{-20pt}
	\label{tab: heap states}
\end{longtable} 
\end{spacing}
}

The proposed algorithm for pilot assignment is summarized in Algorithm~\ref{alg: MSE problem}. It takes the complexity of $ \mathcal{O}\bigr(U\log_{2}(\tau U )\bigl) $ for deriving the assignment solution, which is relatively low complexity.  For simplicity, this training strategy is referred to as Heap-FD, in which Algorithm~\ref{alg: MSE problem} is operated twice, i.e., with $ U=L $ for UL channel estimation in the first phase, and with $ U=K $ for achieving DL  and CCI channel estimates in the second phase. On the other hand, the training strategy for HD systems can be done by setting $ U=K+L $, called Heap-HD. For a given $ \tau $-length of pilot sequences, Heap-FD requires the training time of $ 2\tau $, while Heap-HD needs only $ \tau $. However, it is anticipated  that the channel estimate quality  of Heap-FD would be better than  Heap-HD, due to  smaller number of UEs sharing the same pilot set in Heap-FD.

\begin{algorithm}[t]
	\begin{algorithmic}[1]
		\fontsize{10}{11}\selectfont
		\protect\caption{Proposed Heap-Based Pilot Assignment for MSE Minimization Problem \eqref{eq: prob. MSE quad.}}
		
		\label{alg: MSE problem}
		
		\STATE Compute $ \boldsymbol{\tilde{\beta}}\triangleq[\tilde{\beta}_{j}]_{j\in\mathcal{T}_{\ul}} $ as in \textbf{Theorem}~\ref{thm: MSE prob.}.
		
		\STATE Randomly assign $ \tau $ pilots to the first $ \tau $ UEs in $ \mathcal{T}_{\ul} $, yielding $ \boldsymbol{\upsilon}_{j},\;\forall j=1,\cdots,\tau $.
		
		\STATE Execute $ \mathcal{G}([\boldsymbol{\tilde{\beta}}]_{1:\tau}, \{\boldsymbol{\upsilon}_j \}_{j=1,\cdots,\tau})\rightarrowtail\mathcal{H}^{\mathtt{min}} $.
		
		\STATE Execute $ \mathcal{G}([\boldsymbol{\tilde{\beta}}]_{\tau+1:U},\{\tau+1,\cdots,U\} )\rightarrowtail\mathcal{H}^{\mathtt{max}} $.
		
		\WHILE {$ \mathcal{H}^{\mathtt{max}}\neq \emptyset $}
		\STATE $ \mathcal{H}^{\mathtt{max}} \vdash (\tilde{\beta}_{j'},\{j'\}) $. \COMMENT{\textit{Root node is removed from} $ \mathcal{H}^{\mathtt{max}} $}
		\STATE $ \mathcal{H}^{\mathtt{min}} \rightarrow (\bar{\beta}_{i},\{\mathbf{\boldsymbol{\upsilon}}_{i}\}) $. 
		\STATE $ \mathbf{\boldsymbol{\upsilon}}_{j'}:=\mathbf{\boldsymbol{\upsilon}}_{i} $.
		\STATE $ \mathcal{H}^{\mathtt{min}} \dashv (\bar{\beta}_{i}+\tilde{\beta}_{j'},\{\mathbf{\boldsymbol{\upsilon}}_{i}\}) $.
		\ENDWHILE 
		
		\STATE Concatenate  assignment variable vectors as $ \boldsymbol{\Upsilon}:=[\boldsymbol{\upsilon}_{1},\cdots,\boldsymbol{\upsilon}_{U}] $.

		\STATE {\textbf{Output:}  Pilot assignment matrix $\boldsymbol{\bar{\Xi}} = \boldsymbol{\Xi}\boldsymbol{\Upsilon}$}.
	\end{algorithmic} 
\end{algorithm}

{\majrev
\begin{remark}
The SE-EE problem \eqref{eq: prob. general form bi-obj. trade-off} can be reformulated as a worst-case robust design by treating CSI errors as noise. More specifically, we use channel estimates (rather than the perfect ones) to
perform data transmission. The additional component introduced by CSI errors  in the denominators of SINRs
  is a linear function, and thus, can be easily tackled by our proposed methods given in Sections \ref{sec:ZF} and \ref{sec:IZF}. We refer the interested reader to \cite[Sec. V]{Hieu:IEEETWC:June2019} for further details of the derivations. 
\end{remark}
}

\section{Numerical Results}\label{NumericalResults}
In this section, we provide numerical examples to quantitatively evaluate the performance of the proposed FD CF-mMIMO.
\subsection{Simulation Setup and Parameters}
\begin{table}[t]
	\begin{minipage}{0.45\linewidth}
		\centering
		\fontsize{11}{12}\selectfont
		\captionof{table}{Simulation Parameters}
		\label{tab: parameter}
		\vspace{-5pt}
		\scalebox{0.8}{
			\begin{tabular}{l|l}
				\hline
				Parameter & Value \\
				\hline\hline
				System bandwidth, $B$ & 10 MHz \\
				Reference distances, $ (d_0,d_1) $ & (10, 50) m\\
				Residual SiS, $\rho^{\RSI}=\rho_{mm}^{\RSI},\;\forall m$& -110 dB \cite{Korpi:IEEETAP:Feb2017}\\
				Noise power at receivers & -104 dBm \\
				Number of APs and UEs,\; ($ M $, $K$, $L$) & (64, 10, 10)\\
				Number of antennas per AP, $ N_m,\forall m $ & 2\\
				Rate threshold, $  \bar{R}=\bar{R}_{k}^{\dl}=\bar{R}_{\ell}^{\ul} ,\;\forall k, \ell $	&  0.5 bits/s/Hz\\
				{\color{black}PA efficiency at $\AP$,\ $\ \nu_{\AP},\forall m $} & 0.39 \\
				{\color{black}PA efficiency at $ \ULU $,\ $\ \nu_{\ell}^{\ul},\;\forall \ell $} & 0.3 \\
				{\color{black}Backhaul traffic power, $ P^{\mathtt{bh}} $} & 0.25 W/(Gbits/s) \\
				{\color{black}Baseband power, $ P_{km}^{\dl}=P_{m}^{\ul}, \forall k,m $} & 0.1 W \\
				{\color{black}APs' power in active/sleep modes, $ (P_{\AP}^{\mathtt{a}},P_{\AP}^{\mathtt{s}}),\forall m $} & (10, 2) W \\
				{\color{black}APs' circuit operation power, $ P_{\AP}^{\mathtt{cir}},\forall m $} & 1 W \\
				{\color{black}UEs' circuit operation power, $ P_{k}^{\dl,\mathtt{cir}}=P_{\ell}^{\ul,\mathtt{cir}},\forall k,\ell $} & 0.1 W \\
				{\color{black}Power budget at UL UEs, $P_{\ell}^{\max},\forall\ell $} & 23 dBm \\
				{\color{black}Total power budget for all APs, $MP_{\mathtt{AP}}^{\max}$}  & 43 dBm \\
				\hline		   				
			\end{tabular}
		}
	\end{minipage}
	\hfill
	\begin{minipage}{0.44\linewidth}
		\centering
		\includegraphics[width=0.8\columnwidth,trim={0cm 0cm 0cm 0cm}]{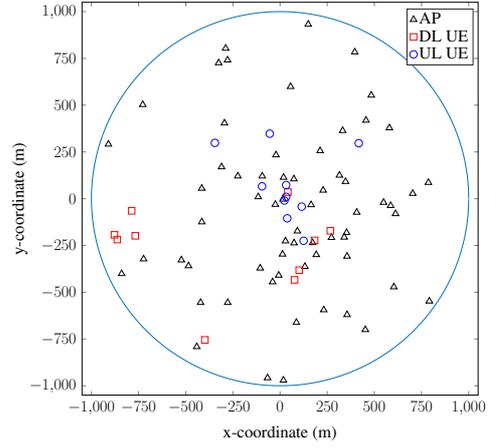}
		\vspace{-10pt}
		\captionof{figure}{A system topology with $M=64$ and $K=L=10$  located within a circle of $ 1 $-km radius is used in numerical examples.}
		\label{fig: Layout}
	\end{minipage}
\end{table}

A system topology illustrated in Fig. \ref{fig: Layout} is considered, 
where all  APs and UEs are located within a circle of 1-km radius. The entries of the fading loop channel $ \mathbf{G}^{\SI}_{mm},\forall m\in\mathcal{M} $ are modeled as independent and identically distributed Rician RVs, with the Rician factor of $ 5 $ dB \cite{Dinh:JSAC:18}. The large-scale fading of other channels is modeled as \cite{Ngo:TWC:Mar2017}
\begin{align}
\beta = 10^{\frac{\mathtt{PL}(d)+\sigma_{\mathtt{sh}}z }{10}},
\end{align}
where $ \beta\in\{ \beta_{mm'}^{\AtoA},\beta_{km}^{\dl}, \beta_{m\ell}^{\ul},\beta_{k\ell}^{\mathtt{cci}} \} $, $\forall  m,m'\in\mathcal{M},k\in\mathcal{K},\ell\in\mathcal{L} $ and $ m\neq m' $; The shadow fading is considered as an RV $ z\in\{ z_{mm'}^{\AtoA},z_{km}^{\dl},z_{m\ell}^{\ul},z_{k\ell}^{\mathtt{cci}} \} \sim \mathcal{N}(0,1) $ with standard deviation $ \sigma_{\mathtt{sh}} =8$  dB. The three-slope model for the
path loss in dB is given by \cite{Ngo:TWC:Mar2017,TangVTC01}
\begin{align} \label{eq: PL model}
\mathtt{PL}(d)=&-140.7-35\log_{10}(d) %\nonumber\\               &
+20c_0\log_{10}(\frac{d}{d_0})+15c_1\log_{10}(\frac{d}{d_1}), 
\end{align}
where $ d\in\{d_{mm'}^{\AtoA},d_{km}^{\dl}, d_{m\ell}^{\ul},d_{k\ell}^{\mathtt{cci}}\} $ is the distance between transceivers as corresponding to $ \beta $; $ d_i$, with $i=\{0,1\}$, denotes the reference distance and $ c_i\triangleq\max\{0,\frac{d_i-d}{|d_i-d|}\} $. Note that the distances $ d $ and $ d_i $ in \eqref{eq: PL model} are measured in km.    Unless specifically stated otherwise,  other parameters are given in Table~\ref{tab: parameter}, where all APs are assumed to have the same power budget $P_{\mathtt{AP}}^{\max}=P_{\mathtt{AP}_m}^{\max},\forall m$.  Herein, the parameters of power consumption and PA efficiencies follow the study in \cite{Bjornson:IEETWC:June2015}.  We use the modeling tool YALMIP in the MATLAB environment. The SEs are divided by $ \ln2 $ to be presented in bits/s/Hz.

For comparison, the following two known schemes are  considered:
\begin{enumerate}
	\item ``Co-mMIMO:'' A BS is deployed at the center of the considered area to serve all UEs. To conduct a fair comparison, the centered-BS is equipped with $N=\sum_{m\in\mathcal{M}}N_m$ number of antennas with the total power budget of $MP_{\mathtt{AP}}^{\max}$  = 43 dBm.
	\item ``SC-MIMO:'' Under the same setup with CF-mMIMO,  each UE is only served by  one AP, but each AP can serve more than one UE. To make ZF feasible, the number of UEs served by one AP must be less than its own number of antennas.
\end{enumerate}
Both FD and HD operations are employed to evaluate the  performance of those schemes. For HD operation, DL and UL transmissions are separately carried out in two independent communication time blocks. As a result, there is no CCI at DL UEs, and no RSI and IAI on  UL reception, but the achieved SE and EE are devided by two. To  show the effectiveness of the proposed ZF- and IZF-based transmissions presented in Sections \ref{sec:ZF} and \ref{sec:IZF}, respectively, we additionally examine the following transmission strategies:
\begin{enumerate}
	\item ``ONB-ZF:'' The precoder matrix for DL transmission  $ \mathbf{W}^{\mathtt{ONB}\text{-}\mathtt{ZF}} $ is computed as in \eqref{eq:ONB-ZF}, while   UL reception adopts the ZF-SIC receiver similar to IZF.
	\item ``MRT/MRC:'' MRT and MRC are applied to DL and UL, respectively. This is easily done by replacing $ \mathbf{H}^{\mathtt{ZF}} $ and $ \mathbf{A}^{\mathtt{ZF}} $ with $ (\mathbf{H}^{\dl})^H $ and $ (\mathbf{H}^{\ul})^H $, respectively.
\end{enumerate}

\subsection{Numerical Results for SE Performance}
\begin{figure}[t]
	\begin{minipage}{0.48\linewidth}
	\centering
	\includegraphics[width=\columnwidth,trim={0cm 0.0cm 0.0cm 0cm}]{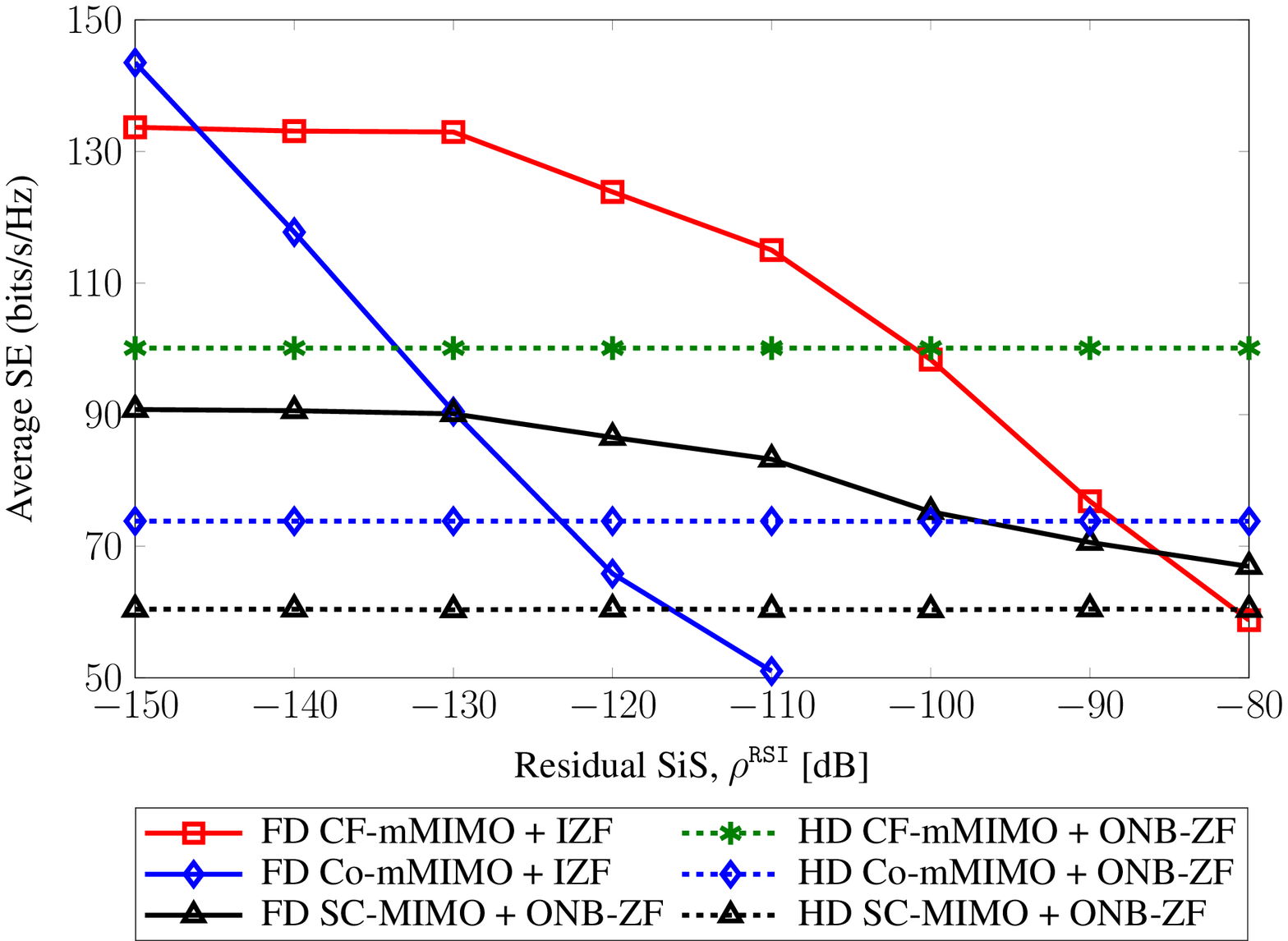}
	\vspace{-5pt}
	\caption{Average SE versus the residual SiS, $ \rho^{\RSI} $, for different schemes with both FD and HD operations.}
	\label{fig: SE vs RSI}
	\end{minipage}
	\hfill
	\begin{minipage}{0.48\linewidth}
		\centering
		\includegraphics[width=1\columnwidth,trim={0cm 0.0cm 0.0cm 0cm}]{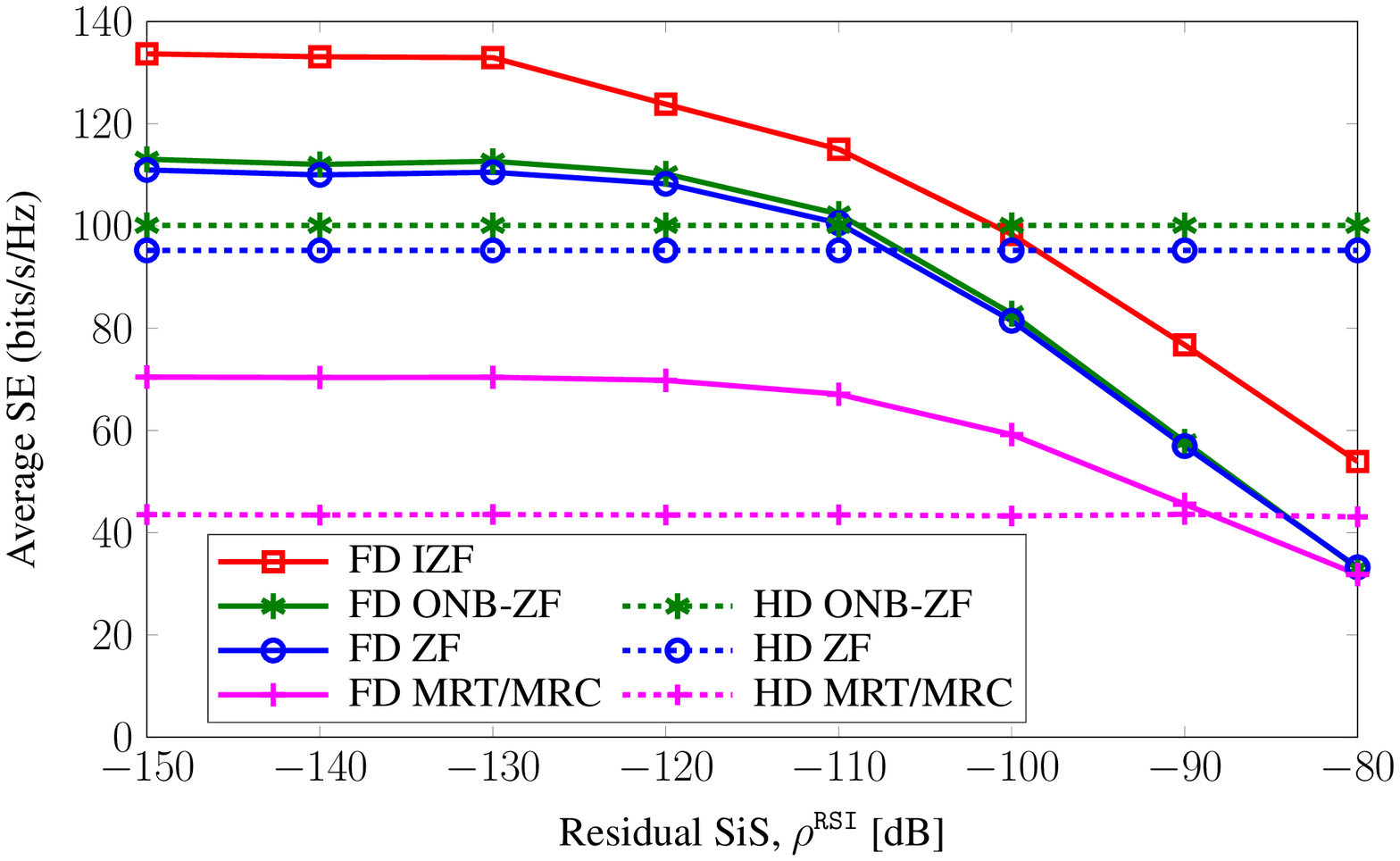}
		\vspace{-20pt}
		\caption{Average SE versus the residual SiS, $ \rho^{\RSI} $, for different transmission strategies in CF-mMIMO.}
		\label{fig: SE vs RSI DLUL Schemes}
	\end{minipage}
%	\hfill
%	\begin{minipage}{0.32\linewidth}
%		\centering
%		\includegraphics[width=0.9\columnwidth,trim={0cm 0cm 0.0cm 0cm}]{figures/CDFvsSE.eps}
%		\vspace{-5pt}
%		\caption{Cumulative distribution function versus the average SE.}
%		\label{fig: CDF vs SE}
%	\end{minipage}
\end{figure}

Fig. \ref{fig: SE vs RSI} depicts the average SE performance as a function of the residual SiS $\rho^{\RSI}$ for  different schemes in both FD and HD operations. We recall that the SI has no effect on the performance of HD-based schemes. The first observation is that FD-based schemes outperform  HD counterparts at a sufficiently small level of $\rho^{\RSI}$. In particular, at $ \rho^{\RSI}=-130 $ dB, FD CF-mMIMO and FD SC-MIMO provide more than 40\% SE performance gain over those of HD ONB-ZF designs. The SE of FD Co-mMIMO scheme and its HD counterpart confirms the promising performance of the FD system. Specifically, when the residual SiS is very small (i.e., $\rho^{\RSI}$ = -150 dB), the performance of FD Co-mMIMO  nearly doubles that of HD  Co-mMIMO. However, its performance is dramatically degraded when $\rho^{\RSI}$ increases. The reason is that the FD Co-mMIMO system with one centered BS serving all UEs in a large area must allocate a high transmit power to far DL UEs, leading to  higher SI power. On the other hand, when $\rho^{\RSI}$ is more severe, the probabilities of infeasibility of the FD schemes are higher, resulting in the degraded performance. Another interesting observation is that the proposed FD CF-mMIMO scheme provides the best performance among FD-based schemes and significantly better performance than HD for a wide range of $\rho^{\RSI}$, including the practical value of $\rho^{\RSI}$ = -110 dB \cite{Korpi:IEEETAP:Feb2017}. These results further confirm  that FD operation is well suited for  CF-mMIMO systems.

{\majrev To evaluate the effectiveness of the proposed ZF and IZF transmission strategies in CF-mMIMO, we compare the average SE of our designs with two simple transmission designs: ONB-ZF and MRT/MRC, as shown in  Fig. \ref{fig: SE vs RSI DLUL Schemes}. With our proposed FD IZF design, the SE improves significantly. The ONB-ZF-PCA precoding in  IZF not only cancels MUI for DL transmission, but also depresses the effect of residual SI on UL reception.
In addition,  ONB-ZF  provides slightly better performance than ZF  in both FD and HD. However, the gap between FD ONB-ZF- and FD ZF-based designs is gradually small when $ \rho^{\RSI} $ increases. It reflects the fact that the effect of residual SI on the system performance of CF-mMIMO is more serious than that of MUI, which brings less benefit of using SIC in the MUI cancellation. 
The MRT/MRC design is always inferior in both FD and HD operations since it is unable to manage the network interference effectively. The difficulty in handling the IAI and residual SI of ONB-ZF and ZF designs causes their FD system to provide worse performance than the proposed FD IZF. }

\begin{figure}
	\begin{minipage}{0.48\linewidth}
		\centering
		\includegraphics[width=\columnwidth,trim={0cm 0cm 0.0cm 0cm}]{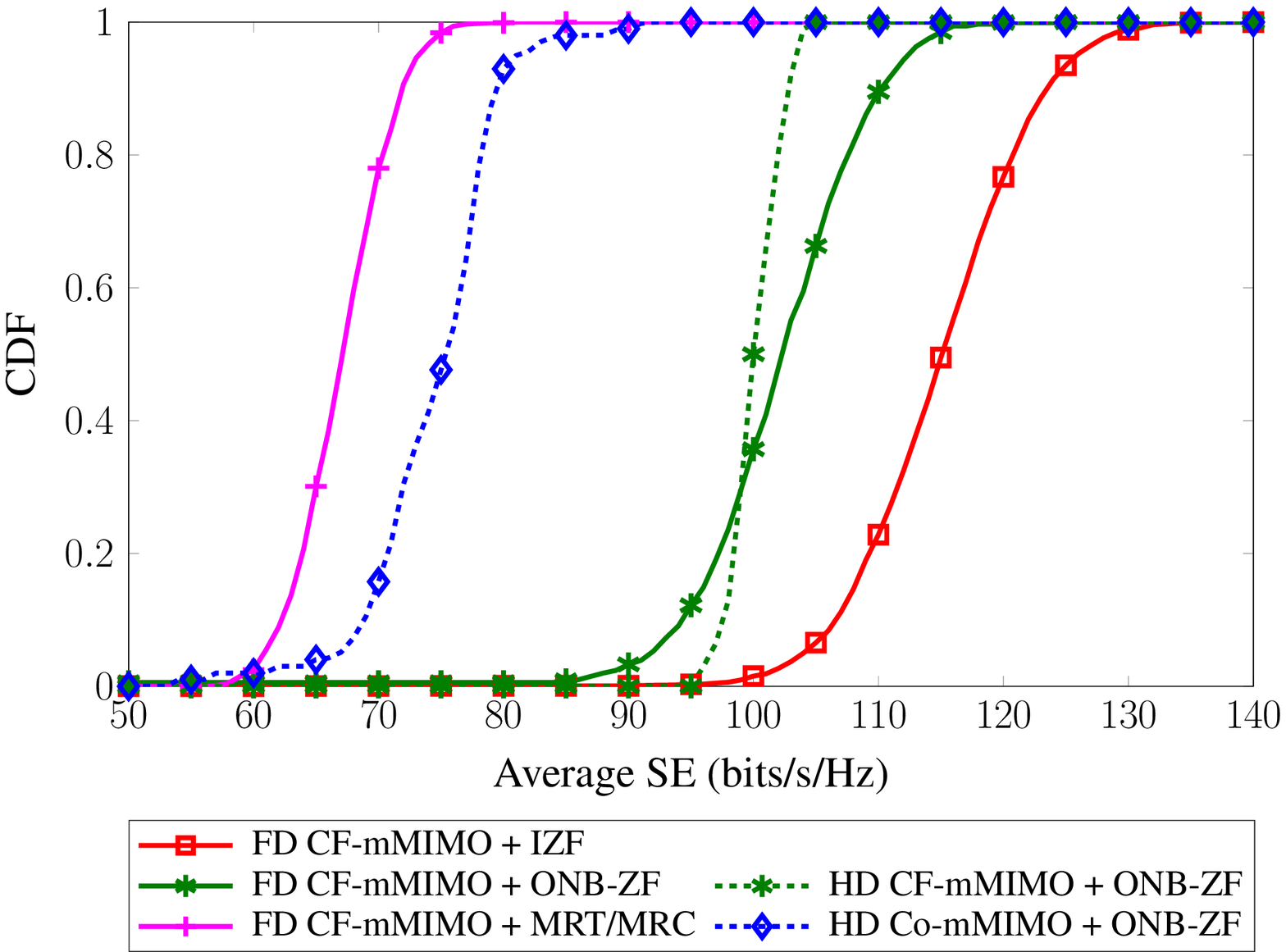}
		\vspace{-5pt}
		\caption{Cumulative distribution function versus the average SE.}
		\label{fig: CDF vs SE}
	\end{minipage}
	\hfill
	\begin{minipage}{0.48\linewidth}
		\centering
		\includegraphics[width=\columnwidth,trim={0cm 0cm 0cm 0cm}]{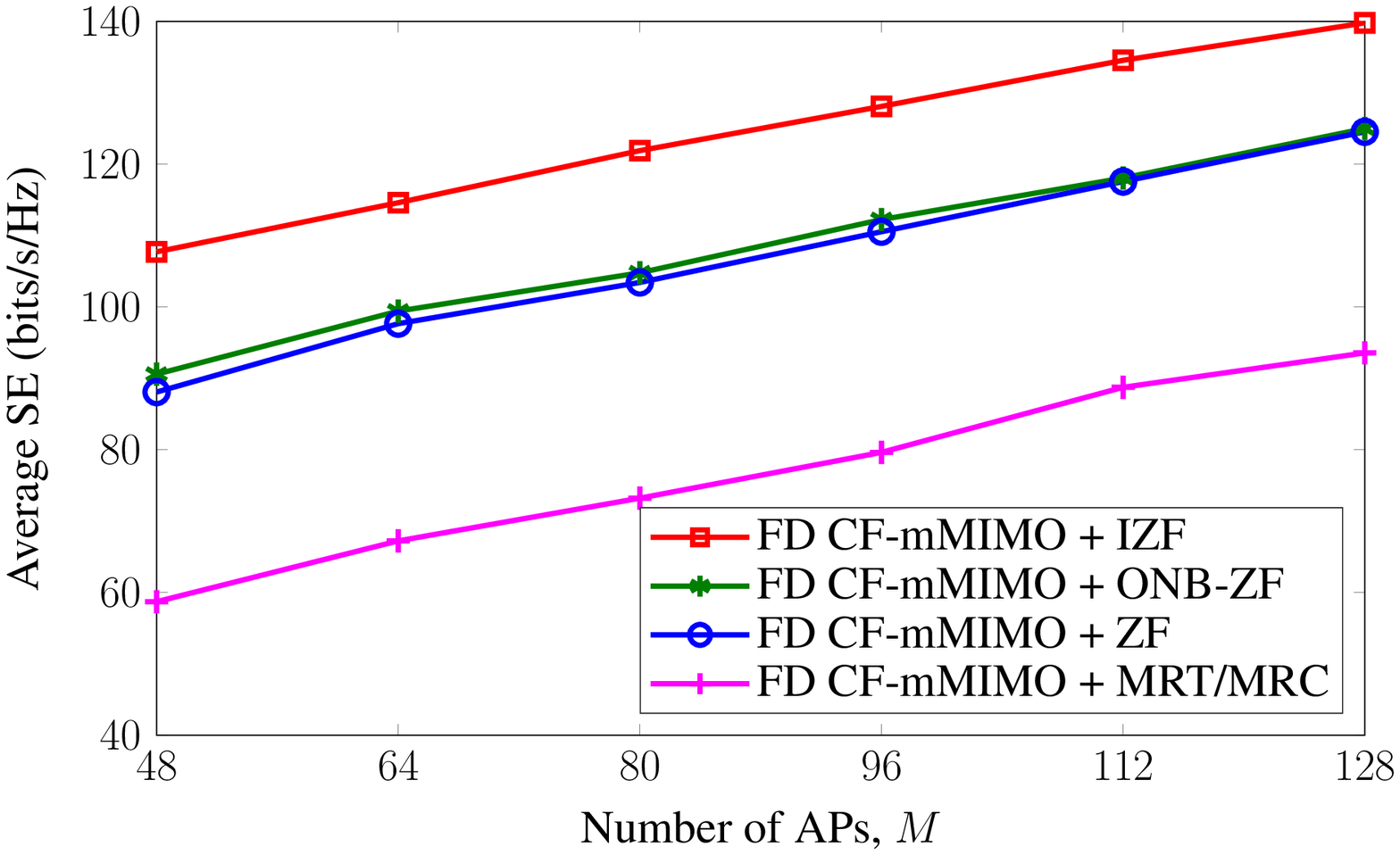}
		\vspace{-5pt}
		\caption{Average SE versus the number of APs with  different transmission strategies in FD CF-mMIMO.}
		\label{fig: SE vs NoAPs}
	\end{minipage}
\end{figure}

 Fig. \ref{fig: CDF vs SE} shows the cumulative distribution function (CDF) of different transmission strategies in FD CF-mMIMO with the references to HD CF-mMIMO and HD Co-mMIMO.  The figure clearly shows  that the feasibility probabilities of all the considered schemes are smaller when the SE is higher. As expected, FD CF-mMIMO with the proposed IZF and ONB-ZF designs outperforms others. In addition, FD CF-mMIMO with IZF offsets SE with about 12 bits/s/Hz more than FD CF-mMIMO with ONB-ZF. Although the ONB-ZF design is able to cancel MUI, it still suffers the large amount of IAI, which is even more severe in FD CF-mMIMO with the dense AP deployment. Its performance is therefore slightly better than that of HD CF-mMIMO at mid-point and 95-percentile point. FD CF-mMIMO with MRT/MRC provides worst performance, which can be explained as follows. MRT/MRC for DL/UL transmission using the channel conjugates (i.e., $ (\mathbf{H}^{\dl})^H $ and $ (\mathbf{H}^{\ul})^H $) is inapplicable to handle the residual SI and also totally passive with the IAI. Such interference is inherent to the MRT/MRC design in FD CF-mMIMO. The CDFs with respect to the SE again validate the advantage of IZF with ONB-ZF for MUI cancellation and PCA procedure for IAI and residual SI depression.

In Fig. \ref{fig: SE vs NoAPs}, we plot the average SE of different transmission strategies in FD CF-mMIMO versus the number of APs, $M\in[48, 128]$. It is straightforward to see that increasing the number of APs causes stronger IAI in CF-mMIMO networks. However, as can be seen from the figure, the SEs of all the considered transmission strategies are monotonically improved when $M$ increases. This results imply that it is enough for each AP to select a suitable number of UEs  around it, without causing much interference to other APs.

\subsection{Numerical Results for EE Performance}
\begin{figure}[t]
%	\begin{minipage}{0.32\linewidth}
%		\centering
%		\includegraphics[width=\columnwidth,trim={0cm 0cm 0cm 0cm}]{figures/SEvsNoAPs.eps}
%		\vspace{-5pt}
%		\caption{Average SE versus the number of APs with  different transmission strategies in FD CF-mMIMO.}
%		\label{fig: SE vs NoAPs}
%	\end{minipage}
%	\hfill
	\begin{minipage}{0.48\linewidth}
	\centering
	\includegraphics[width=\columnwidth,trim={0cm 0cm 0cm 0cm}]{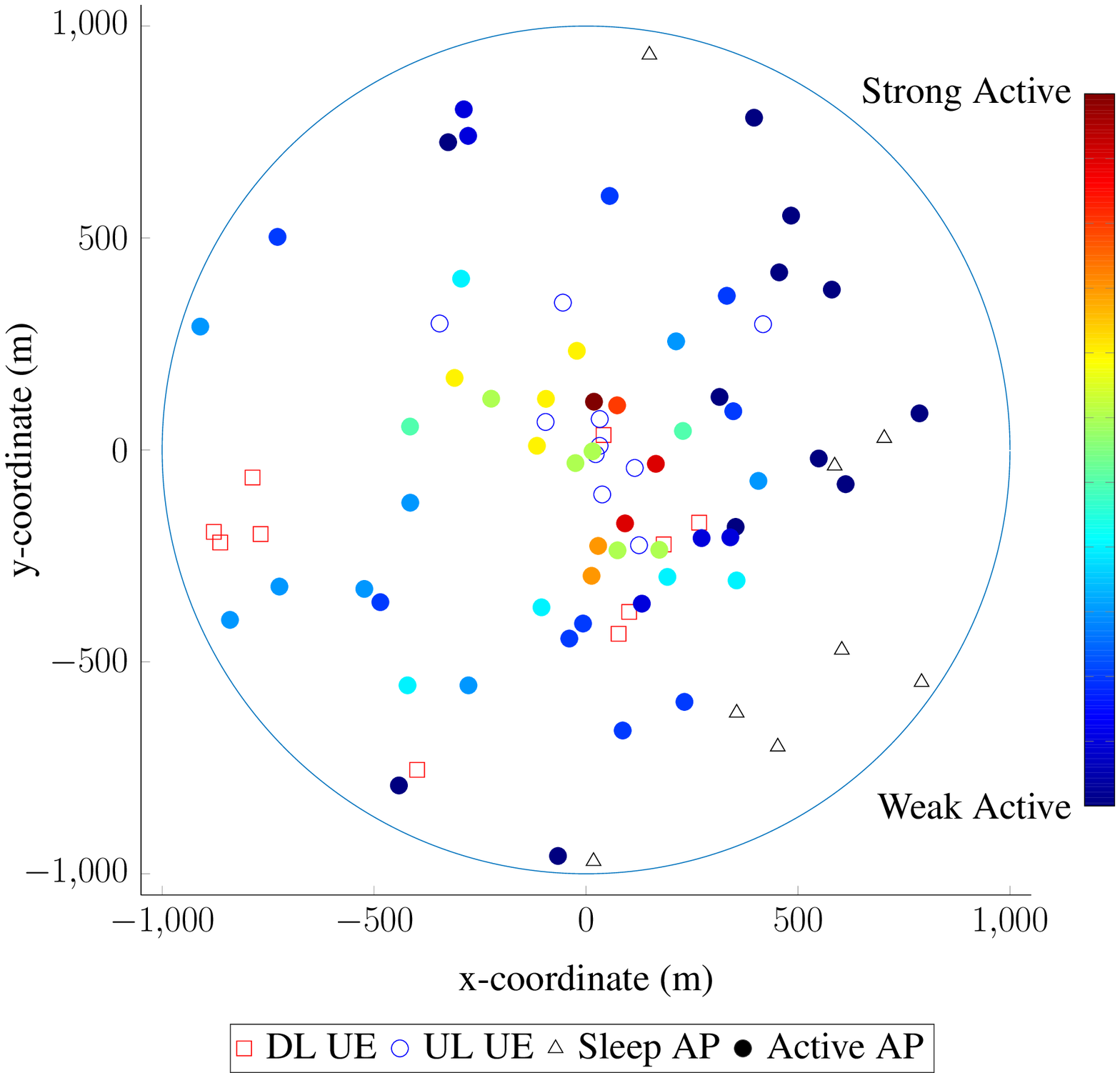}
	\vspace{-5pt}
	\caption{Active and sleep modes of APs using the system topology in Fig. \ref{fig: Layout}.}
	\label{fig: On Off APs}
	\end{minipage}
	\hfill
	\begin{minipage}{0.48\linewidth}
		\centering
		\includegraphics[width=\columnwidth,trim={0cm 0cm 0cm 0cm}]{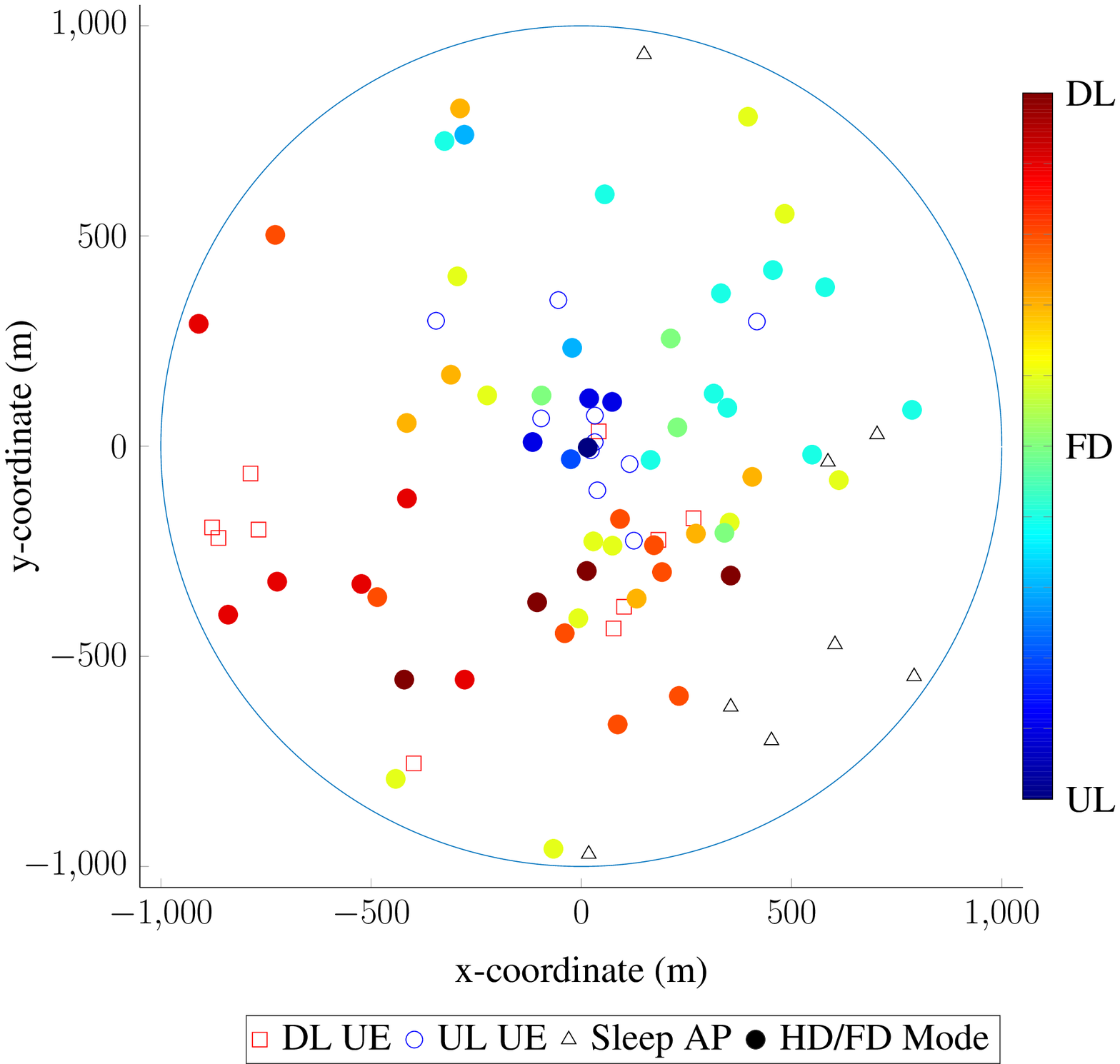}
		\vspace{-5pt}
		\caption{Operation behavior of APs using the system topology in Fig. \ref{fig: Layout}}
		\label{fig: DLUL Serving Map}
	\end{minipage}
\end{figure}

{\majrev Before providing the EE performance, we show the status of APs (i.e., active and sleep modes) in Fig. \ref{fig: On Off APs} and operation behavior of APs (i.e., DL/UL or FD) in Fig.  \ref{fig: DLUL Serving Map}, gaining more insights into the effect of AP selection. We consider  the system topology illustrated in Fig. \ref{fig: Layout}, and the IZF design to maximize the EE.
It can be seen in Fig. \ref{fig: On Off APs} that most APs located far way from UEs switch off to reduce  power consumption, since the large effect of path loss makes  power consumption for these APs inefficient. In contrast,  APs in the area of dense UEs become strongly active in order to enhance the SE  $ \tilde{F}_{\mathtt{SE}}\bigl(\boldsymbol{\Lambda}_{\dl},\boldsymbol{\Lambda}_{\ul}\bigr) $, which dominates the loss caused by $ \phi $ in \eqref{eq: prob. bi-obj. - Dinkelbach}. This result suggests an interesting observation that each UE should be served by a small subset of active APs to manage the network interference more effectively. 
Fig. \ref{fig: DLUL Serving Map} shows how the power budget at an active AP is allocated to UEs. Intuitively, the active APs dynamically connect to  near UEs for the EE improvement, and thus, HD (DL/UL) or FD mode is selected depending on either DL or UL transmission of near UEs, as long as their rate thresholds are satisfied. Specifically, APs are operated in DL and UL modes when they are located close to DL and UL UEs, respectively, and the FD mode, otherwise.  In other words, the proposed method allows  active APs to dynamically switch between DL/UL (in HD) and FD modes based on the channel conditions of all UEs to alleviate the IAI, residual SI and CCI. These observations reflect the importance  of the joint design of AP-UE association and AP selection to obtain the maximum EE performance. It is expected that  FD CF-mMIMO  with AP selection outperforms the case without AP selection, as discussed in the following part.}

We now examine the EE performance versus the number of UEs ($K=L$) in FD CF-mMIMO with different transmission strategies, as illustrated in Fig. \ref{fig: EE vs No UEs}. To ensure a high feasibility of the considered schemes even when the number of UEs is large, we set the number of APs to $ M=256 $. The EE of the  considered transmission strategies first increases, approaches the
optimal point, and then  decreases, as the number of UEs increases. This phenomenon   is attributed to the fact that, for small and medium numbers of UEs, the SE improvement dominates the total power consumption, leading to the significantly enhanced EE performance. However, when the number of UEs is very large (i.e., $K=L>50$), the SE is slightly improved, even reduced, due to stronger network interference, while the total power consumption increases quickly since more APs are required to be active to serve larger number of UEs. Nevertheless, our proposed IZF design outperforms the baseline ones.

{\majrev Fig. \ref{fig: EE vs No APs} plots the average EE as a function of the number of APs $M$ for $\bar{R}\in\{0,0.5\}$ bits/s/Hz, with and without AP selections.  For benchmarking purpose, we consider two baseline schemes: ($i$) Power consumption of sleep APs is completely set to zero, named as ``IZF with perfect AP selection (IZF w/ Perf. AP Sel.);'' ($ii$) AP-UE association and AP selection are used in company with the channel matrix transpose and an equal weight coefficient for all UEs, referred to as ``AP selection without weight beamforming design (AP Sel. w/o WBF).'' It can be seen that the IZF design with AP selection obtains much better EE performance than without AP selection, i.e., up to about 50\% EE gain. However,  the EE with  AP selection degrades when the number of APs becomes large, since the sleep APs still consume a fixed power in sleep modes, as given in \eqref{curcirpower}. This also explains why the perfect AP selection achieves the best performance among all the schemes and gradually increases along with the number of APs. In addition, increasing $M$ brings no benefit to the schemes without AP selection. We can also see that the performance gaps between $ \bar{R}=0 $ bits/s/Hz and $\bar{R} = 0.5 $ bits/s/Hz of all considered schemes are narrower when $M$ increases. The larger the number of APs, the  higher  the  probability for UEs to efficiently select a subset of APs with good channel conditions,  leading to better rate fairness among all UEs.}

\begin{figure}[t]
	\begin{minipage}{0.48\linewidth}
		\centering
		\includegraphics[width=\columnwidth,trim={0cm 0cm 0cm 0cm}]{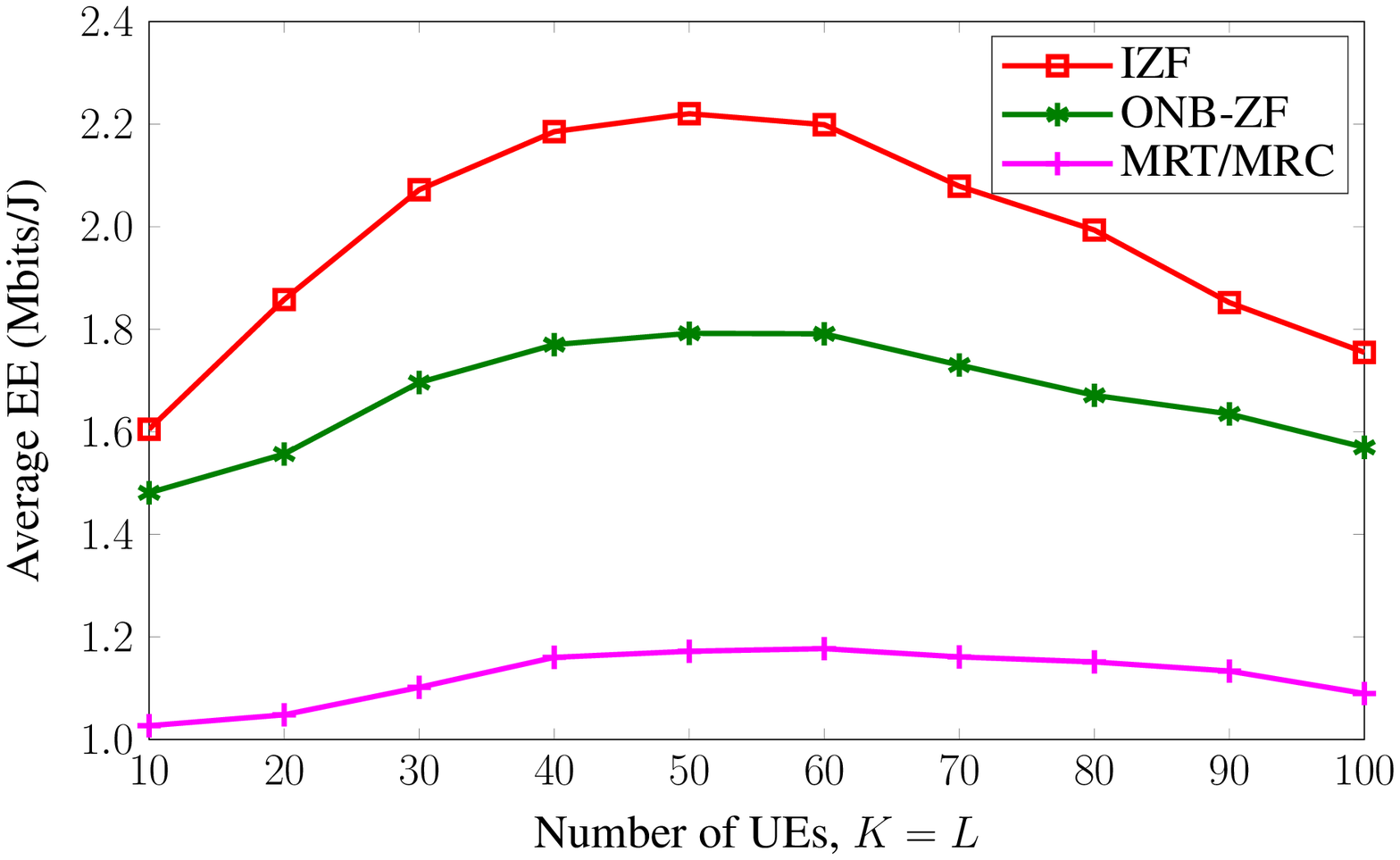}
		\vspace{-10pt}
		\caption{Average EE versus the number of UEs in FD CF-mMIMO with $ M=256 $ and $ \bar{R}=0.5 $ bits/s/Hz.}
		\label{fig: EE vs No UEs}
	\end{minipage}
	\hfill
	\begin{minipage}{0.48\linewidth}
		\centering
		\includegraphics[width=0.88\columnwidth,trim={0cm 0cm 0cm 0cm}]{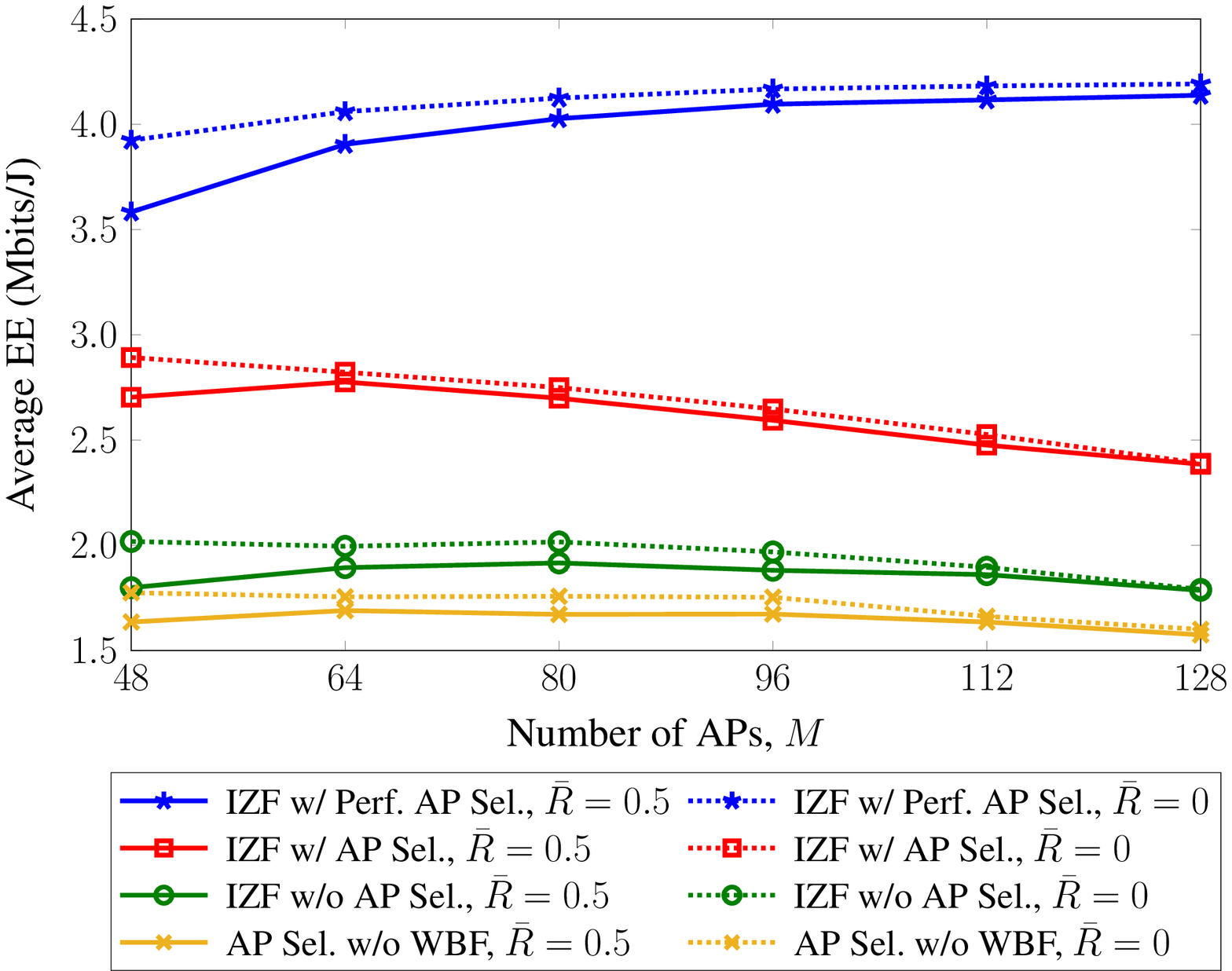}
		\vspace{-7pt}
		\caption{Average EE versus the number of APs in FD CF-mMIMO under different types of AP selections, with $ \bar{R}\in \{ 0, 0.5 \} $ bits/s/Hz.}
		\label{fig: EE vs No APs}
	\end{minipage}
%	\hfill
%	\begin{minipage}{0.33\linewidth}
%		\centering
%		\includegraphics[width=\columnwidth,trim={0cm 0cm 0cm 0cm}]{figures/NMSEvsNoUEs.eps}
%		\vspace{-10pt}
%		\caption{Normalized MSE versus the number of UEs in CF-mMIMO with  $ \tau\in\{2, 8 \}$.}
%		\label{fig: NMSE vs No UEs}
%	\end{minipage}
\end{figure}

\subsection{Numerical Results for Heap-Based Pilot Assignment Algorithm \ref{alg: MSE problem}}

It can be easily foreseen that the quality of channel estimates mainly depends on the relationship between the number of UEs and  dimension of pilot set (or pilot length, $ \tau $). To evaluate the performance of the proposed FD training strategy, we first investigate the normalized MSE (NMSE) as a function of the number of UEs. As depicted in Fig. \ref{fig: NMSE vs No UEs}, we consider four strategies: two heap structures for pilot assignment (Heap-FD and Heap-HD), and two random pilot assignments (Rand-FD and Rand-HD). As expected, the proposed heap training schemes outperform the random ones. It can also be observed that  FD training strategies offer better performance in terms of NMSE compared to HD ones, by exploiting larger dimension of pilot sequences more efficiently. In particular, when $ K=L>\tau $, NMSE of the proposed Heap-FD  is around 5  dB and 7 dB less than Heap-HD, corresponding to $ \tau=2 $ and $ \tau= 8 $, respectively. However, we note that the FD training strategy requires a double training time over its HD counterpart, leading to the difference of the effective time for data transmission. The SE under imperfect CSI can be expressed as
\begin{align}
	\hat{F}_{\mathtt{SE}}\bigl(\mathbf{w}, \mathbf{p},\boldsymbol{\alpha}\bigr) = \frac{\tau_{\mathtt{c}}-\tau_{\mathtt{t}}}{\tau_{\mathtt{c}}}F_{\mathtt{SE}}\bigl(\mathbf{w}, \mathbf{p},\boldsymbol{\alpha}\bigr),
\end{align}
where $\tau_{\mathtt{c}}$ and $\tau_{\mathtt{t}}$ are the coherent time and training time, respectively. We now plot the SE performance for the worst-case robust design by taking into account the channel estimation. In Fig \ref{fig: SE vs No UEs}, we set $ \tau_{\mathtt{c}}=200 $,  $ \tau_{\mathtt{t}}=2\tau $ for  FD and $ \tau_{\mathtt{t}}=\tau $ for HD. Unsurprisingly,  Heap-FD schemes outperform  HD ones, and their performance gaps are even more remarkable when the number of UEs increases. This again demonstrates the effectiveness of the proposed Heap-based pilot assignment algorithm for FD CF-mMIMO by reaping both the advantages of  higher dimension of pilot sequences for training and FD for data transmission.

\begin{figure}
	\begin{minipage}{0.48\linewidth}
		\centering
		\includegraphics[width=\columnwidth,trim={0cm 0cm 0cm 0cm}]{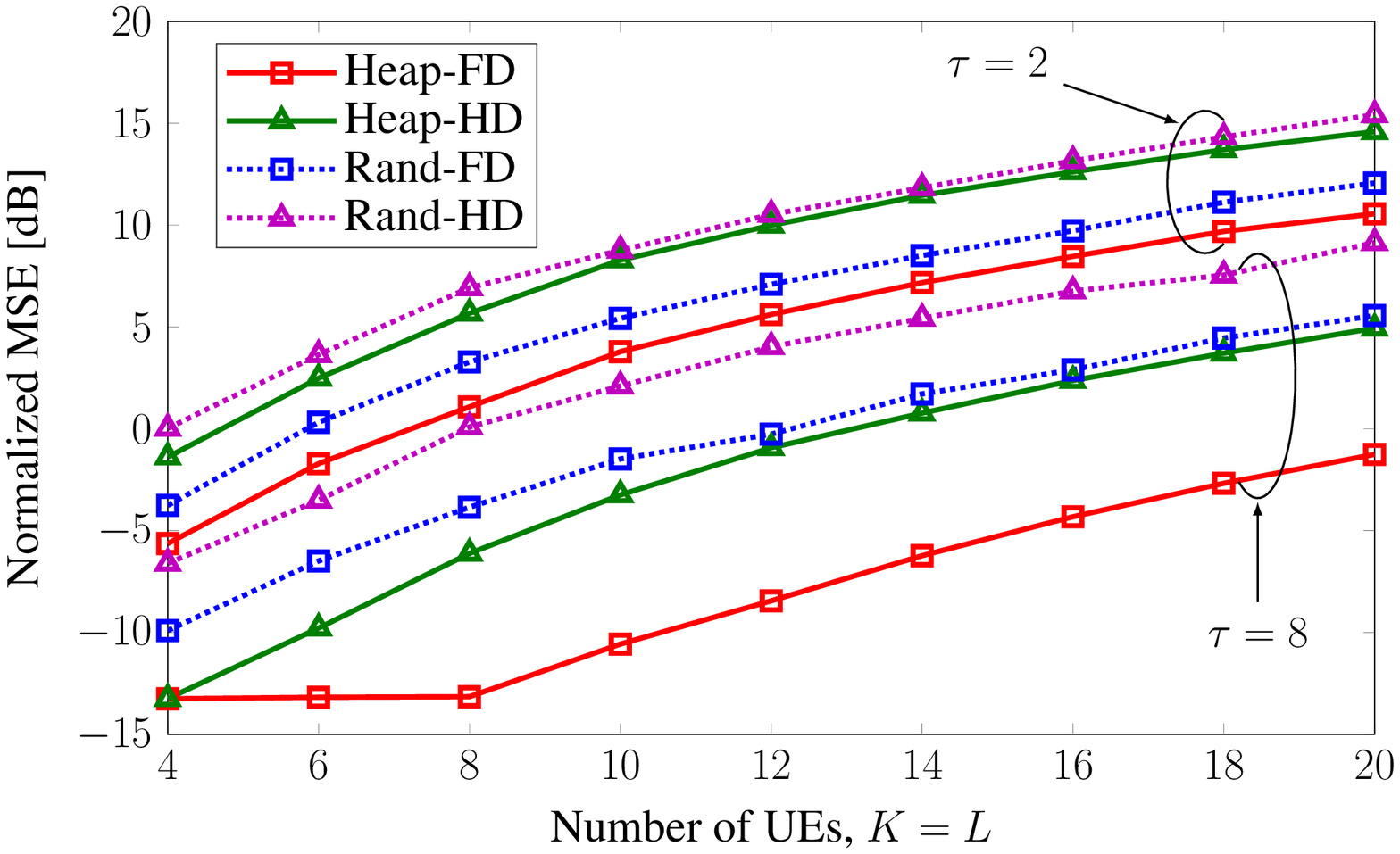}
		\vspace{-10pt}
		\caption{Normalized MSE versus the number of UEs in CF-mMIMO with  $ \tau\in\{2, 8 \}$.}
		\label{fig: NMSE vs No UEs}
	\end{minipage}
	\hfill
	\begin{minipage}{0.48\linewidth}
		\centering
		\includegraphics[width=\columnwidth,trim={0cm 0cm 0cm 0cm}]{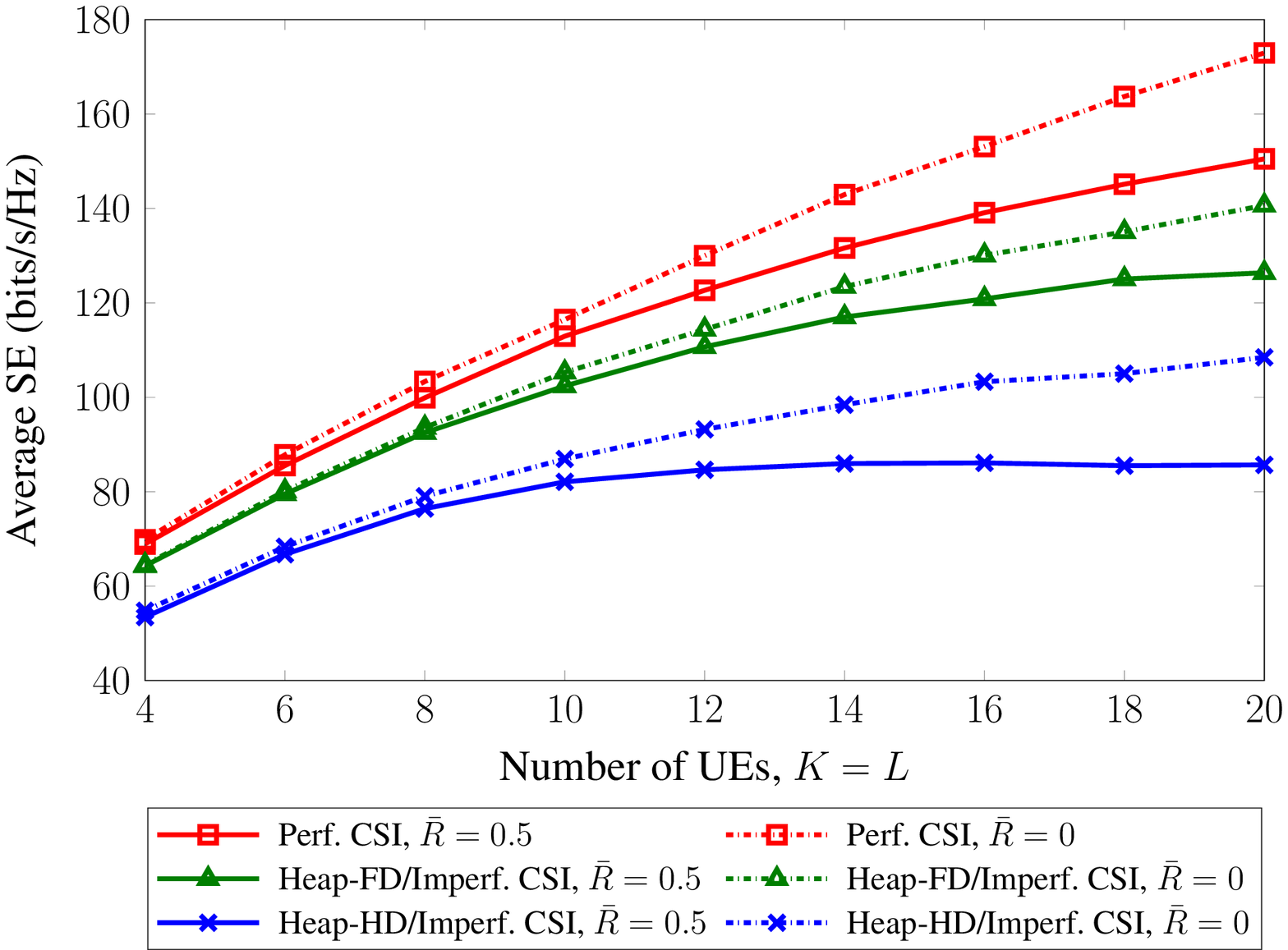}
		\vspace{-10pt}
		\caption{Average SE versus the number of UEs in CF-mMIMO with $ \tau=8 $ and $ \bar{R}\in\{0,0.5\} $ bits/s/Hz.}
		\label{fig: SE vs No UEs}
	\end{minipage}
\end{figure}

\subsection{Convergence Behavior and Computational Capability of Algorithm \ref{alg: ZFD problem}}

\begin{figure}[t]
%	\begin{minipage}{0.3\linewidth}
%		\centering
%		\includegraphics[width=0.9\columnwidth,trim={0cm 0cm 0cm 0cm}]{figures/SEvsNoUEs_ImperfCSI.eps}
%		\vspace{-10pt}
%		\caption{Average SE versus the number of UEs in CF-mMIMO with $ \tau=8 $ and $ \bar{R}\in\{0,0.5\} $ bits/s/Hz.}
%		\label{fig: SE vs No UEs}
%	\end{minipage}
%	\hfill
	\begin{minipage}{\linewidth}
	\centering
	\begin{subfigure}[EE convergence behavior of Algorithm \ref{alg: ZFD problem} in FD CF-mMIMO, with the per-AP power signal ratio $ \varpi=0.1\% $ of $ 1/M $.]
		{
			\includegraphics[width=0.45\columnwidth,trim={0cm 0cm 0cm 0cm}]{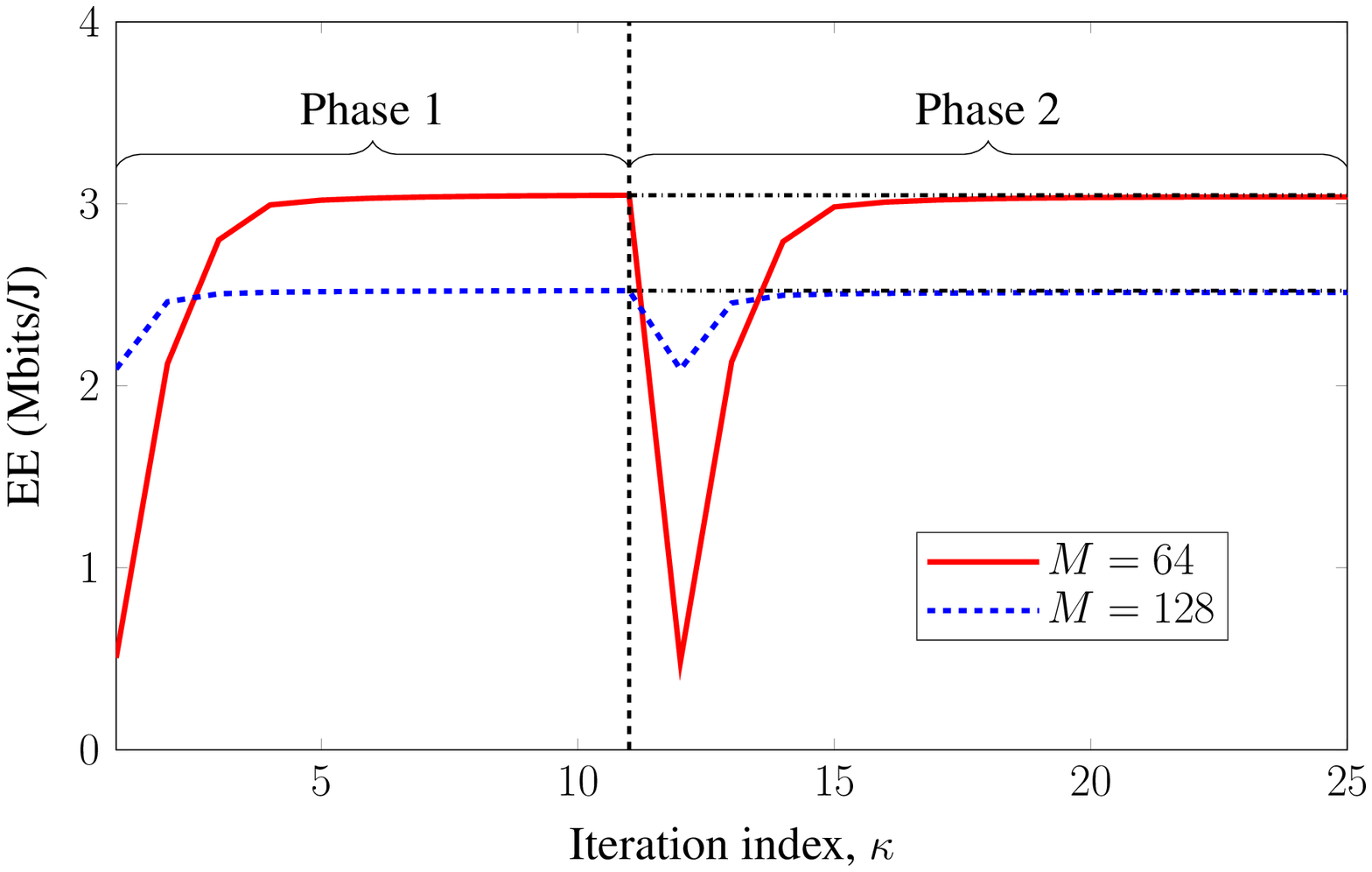}
			\vspace{-5pt}
%			\caption{}
			\label{fig: EE Convergence 0.1}
		}
	\end{subfigure}
	\hfill
	\begin{subfigure}[EE convergence behavior of Algorithm \ref{alg: ZFD problem} in FD CF-mMIMO, with the per-AP power signal ratio $ \varpi=1\% $ of $ 1/M $.]
		{
			\includegraphics[width=0.45\columnwidth,trim={0cm 0cm 0cm 0cm}]{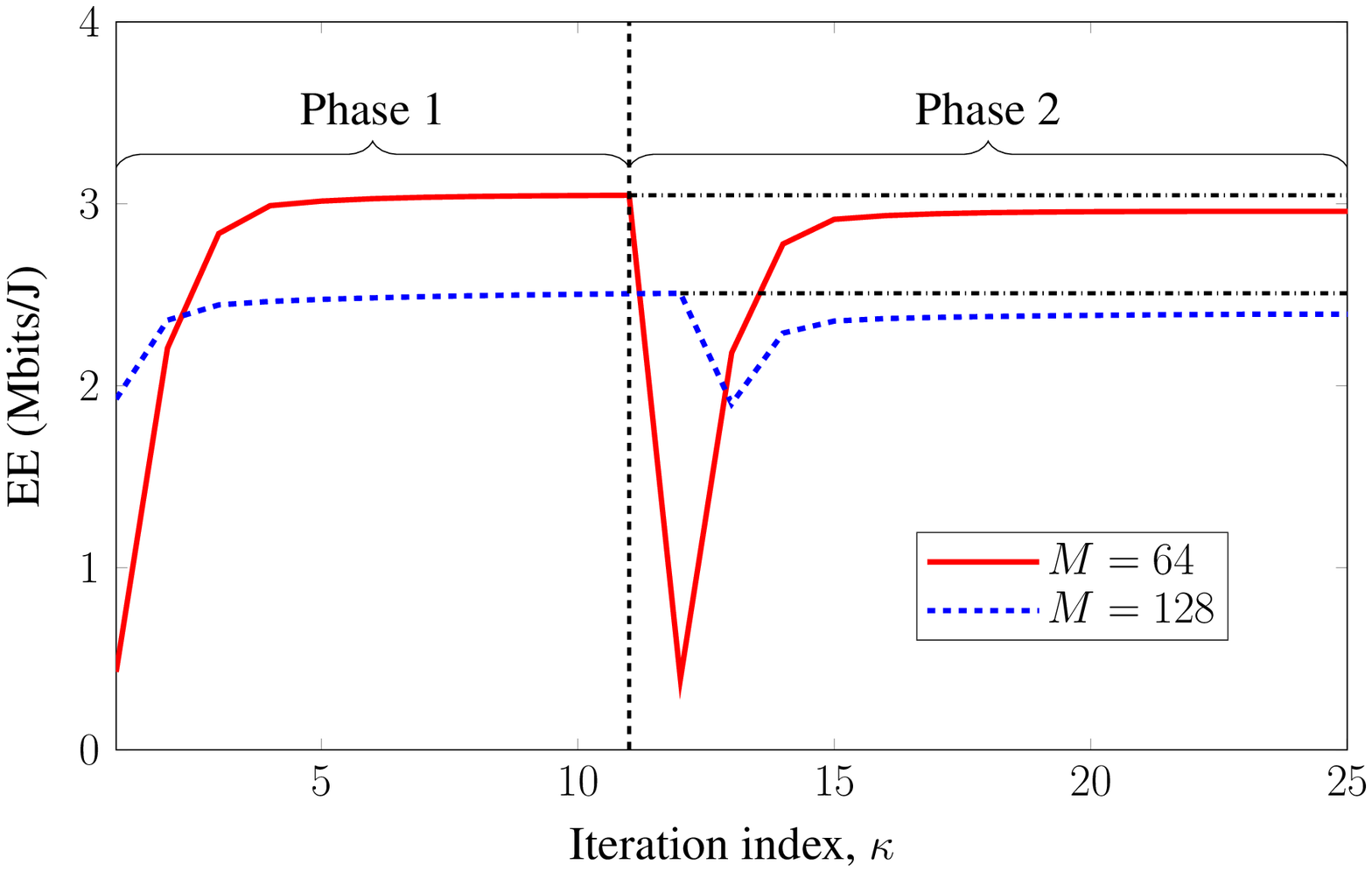}
			\vspace{-5pt}
%			\caption{}
			\label{fig: EE Convergence 1.0}
		}
	\end{subfigure}
	\vspace{-10pt}
	\caption{Typical EE convergence behaviors of Algorithm \ref{alg: ZFD problem} over a random channel realization.}
	\label{fig: EE Convergence}
	\end{minipage}
\end{figure}

We have numerically observed that the convergence behavior of the SE optimization in Algorithm \ref{alg: ZFD problem} is similar to that of the EE. For the sake of brevity, only the EE convergence performance is plotted in Fig. \ref{fig: EE Convergence} by setting $\eta=0$, with $M\in\{64, 128\}$. Since each AP has an average percentage of rate contributed to UEs $ \hat{\varpi}=1/M $, we consider $ \varpi $ given in \eqref{eq:Bfunction} as 0.1\% and 1\% of $1/M$, as illustrated in Fig. \ref{fig: EE Convergence 0.1} and Fig. \ref{fig: EE Convergence 1.0}, respectively. As seen, Algorithm \ref{alg: ZFD problem} converges very fast, and attains 99\% EE performance within about 10 iterations for both the first phase (Steps 1-9 in Algorithm \ref{alg: ZFD problem}) and second phase (Step 11 in Algorithm \ref{alg: ZFD problem}). The figure also clearly demonstrates the effect of the selection of  per-AP power signal ratio $ \varpi $ on the EE performance.  For $ \varpi = 0.1\%$ of $1/M$, the gap between two phases is very small, i.e., lower than 0.05\% of the EE achieved in phase 1. For $ \varpi = 1\%$ of $1/M$,  the EE loss in phase 2 increases up to 3\% and 5\% , corresponding to $ M $ = 64 and  $ M $ =128, respectively. It simply implies that the value of $ \varpi $ should be properly chosen to not only achieve a good performance, but also recover an exact binary value of $\boldsymbol{\alpha}$ and $\boldsymbol{\mu}$.

Finally, we  provide the average execution time of the proposed designs, which mainly aims at showing how their computational complexities scales with the network size. The codes are implemented in MATLAB with the  modeling toolbox YALMIP running on a computer of Intel(R) Core(TM) i7-6700 CPU @ 3.4 GHz, RAM 16 GB and Windows 10. Fig. \ref{fig: Execution time vs No APs} shows that the execution time of all the proposed designs slightly increases even when $M$  increases rapidly, since the problem is less dependent on $M$. However, the execution time scales exponentially with respect to  the number of UEs, as shown in Fig. \ref{fig: Execution time vs No UEs}.  Although the computational complexities of IZF, ONB-ZF, and MRT/MRC design are similar, the distinct feasible regions would lead to  different processing times. The higher execution time of MRT/MRC is due to its smaller feasible region. To be more comprehensive, Fig. \ref{fig: IAI SI pow and time vs delta} plots the normalized effective sum power of IAI and RSI (i.e., $\|\mathbf{\tilde{G}}^{\mathtt{AA}}\mathbf{W}\|^2/\sigma^2$), and the execution time versus the percentage of $ \bar{N} $ eigenvalues in \eqref{eq: N_bar top eigenvalues} over the total eigenvalues, $ \delta \in (0.6, 0.99] $. Clearly,  the normalized sum power of IAI and RSI for ONB-ZF and MRT/MRC are unchanged, regardless of the value of $ \delta $.  The reason is that  ONB-ZF  has to preserve the structure of the ZF matrix, while MRT/MRC is based on the transpose of channel responses to compute the precoder/receiver.  More importantly, the IZF transmission design is capable of providing lower \textit{leakage} IAI/RSI power and execution time at the large value of $ \delta $.

\begin{figure}[t]
	\centering
	\begin{subfigure}[Average execution time versus the  number of APs, $ M $.]
		{
			\includegraphics[width=0.3\columnwidth,trim={0cm 0cm 0cm 0cm}]{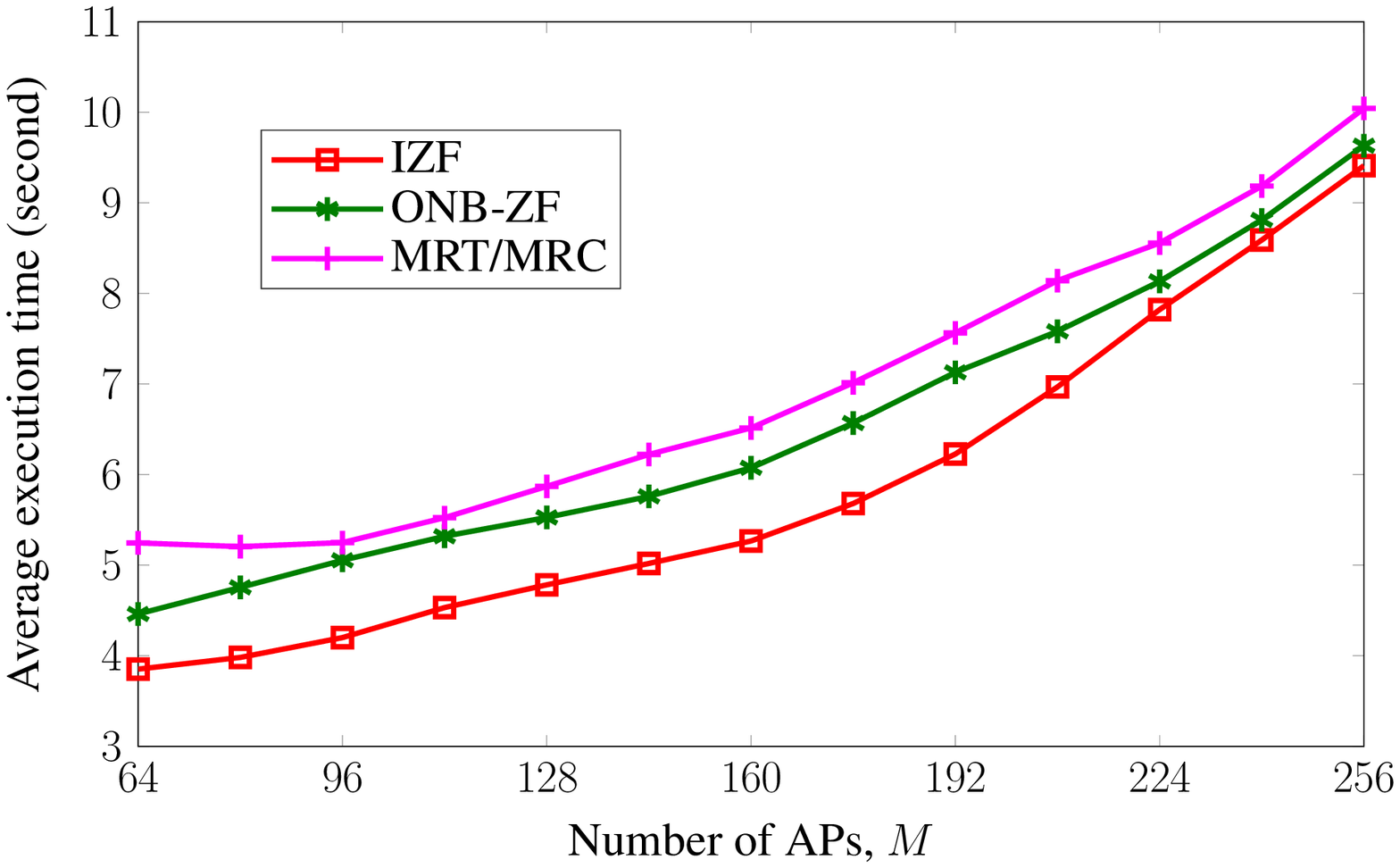}
			\vspace{-5pt}
			\label{fig: Execution time vs No APs}
		}
	\end{subfigure}
	\hfill
	\begin{subfigure}[Average execution time versus the number of UEs, with $ M=256 $.]
		{
			\includegraphics[width=0.3\columnwidth,trim={0cm 0cm 0cm 0cm}]{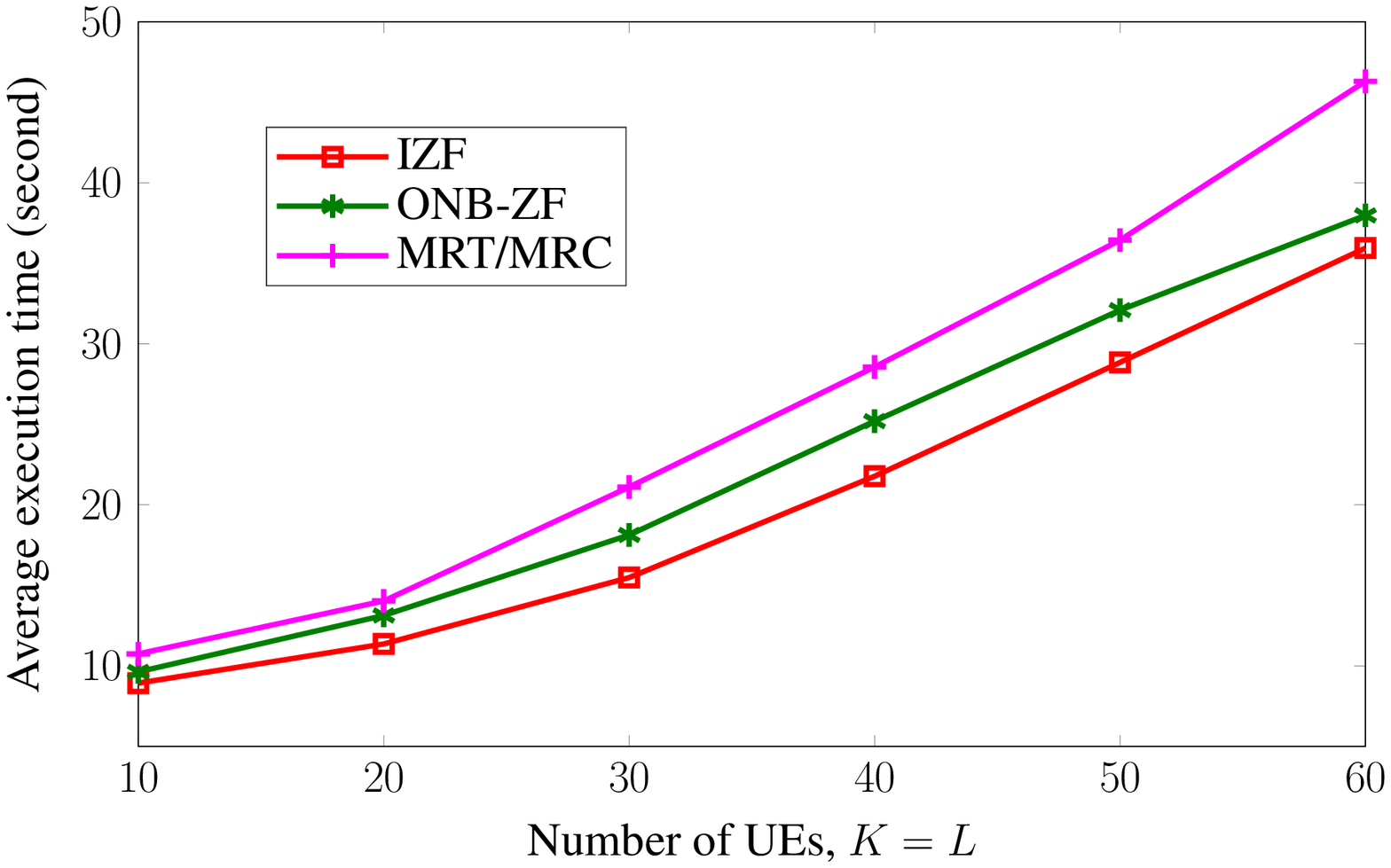}
			\vspace{-25pt}
			\label{fig: Execution time vs No UEs}
		}
	\end{subfigure}
	\hfill
	\begin{subfigure}[Normalized sum power of IAI and RSI and execution time versus $ \delta $, with $ M=256 $, $ K=L=60 $, and $ \mathbf{W}\in\{\mathbf{W}^{\mathtt{IZF}},\mathbf{W}^{\mathtt{ONB-ZF}}, (\mathbf{H}^{\dl})^H\} $.]
		{
			\includegraphics[width=0.3\columnwidth,trim={0cm 0cm 0cm 0cm}]{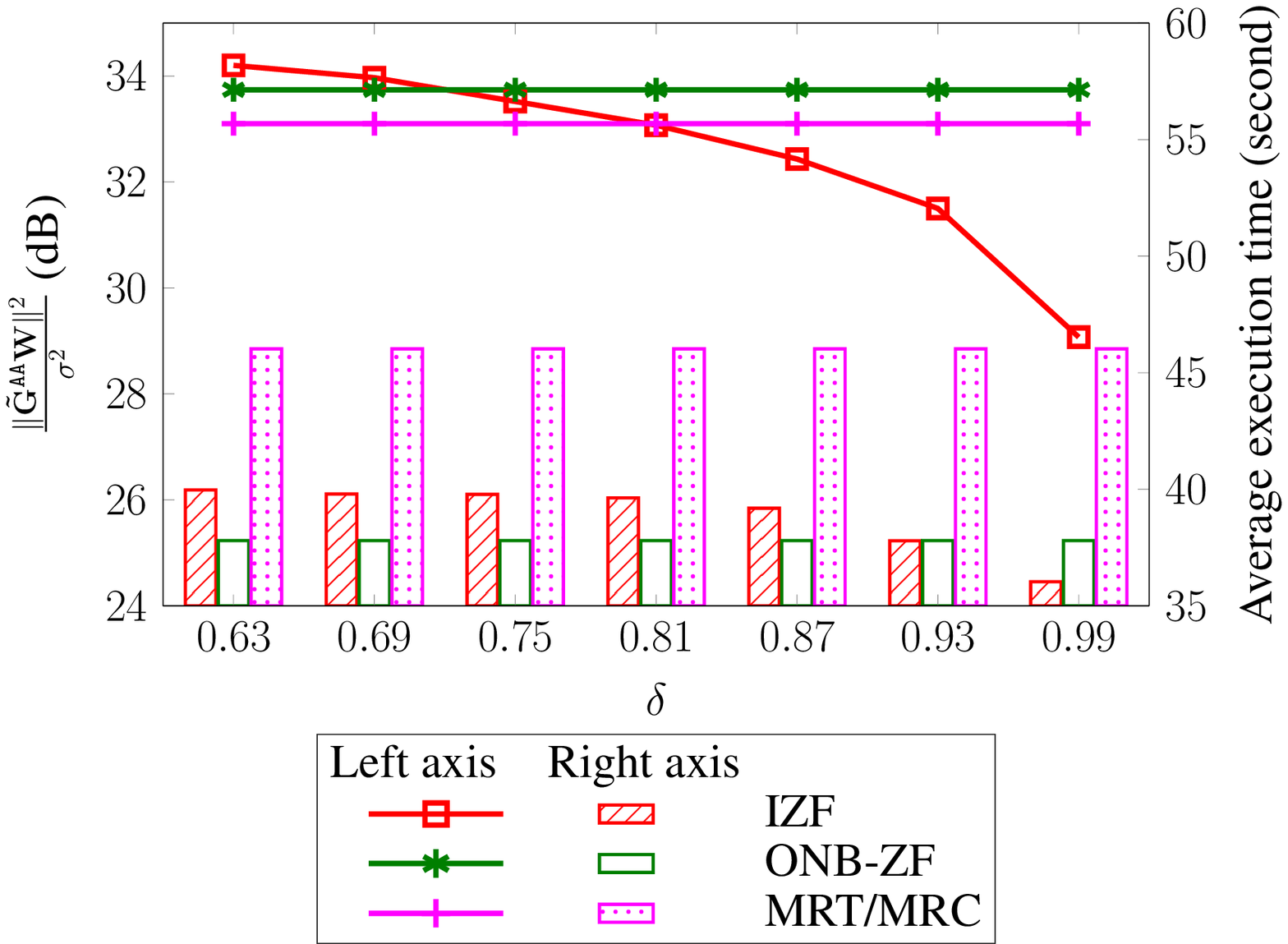}
			\vspace{-25pt}
			\label{fig: IAI SI pow and time vs delta}
		}
	\end{subfigure}
	\caption{Average execution time of Algorithm \ref{alg: ZFD problem} for different transmission strategies in FD CF-mMIMO.}
	\label{fig: Computational time}
\end{figure}

\section{Conclusion}\label{Conclusion}
We have investigated the SE and EE of an FD CF-mMIMO network by jointly optimizing power control, AP-UE association and AP selection. The realistic  power consumption model, which accounts for data transmission, baseband processing and  circuit operation, has been taken into consideration in characterizing the EE performance. Also, the special relationship between binary and continuous variables has been efficiently exploited to reduce  the number of optimization variables. First, we have derived  the iterative procedure based on the ICA  framework and Dinkelbach method to solve the ZF-based problem, where each iteration only solves a simple convex program. Aiming at efficient network interference management, we have then proposed an improved ZF-based transmission by incorporating ONB-and-PCA in the DL, and SIC in the UL.  In addition, a novel and low-complexity pilot assignment algorithm based on the  heap structure has been  developed to improve the quality of channel estimates.  

The proposed algorithm admitted  fast convergence rate, and 
showed to significantly  outperform  SC-MIMO and Co-mMIMO in terms of  SE and EE by jointly optimizing the parameters
involved. Via presented results, it can be concluded that FD CF-mMIMO with IZF transmission design is much more robust against the effects of residual SiS and IAI and requires lower execution time than ZF, ONB-ZF and MRT/MRC. Numerical results also showed that much better EE performance can be yielded by our joint design together with the AP selection. 

\appendices
\renewcommand{\thesectiondis}[2]{\Alph{section}:}

\section{Proof of Lemma \ref{thm: relationship alpha and beamforming vector}} \label{app: relationship alpha and beamforming vector}
\renewcommand{\theequation}{\ref{app: relationship alpha and beamforming vector}.\arabic{equation}}\setcounter{equation}{0}
	Suppose that the optimal solution for \eqref{eq: prob. general form bi-obj. trade-off} is found as a 2-tuple $\mathbf{u}^*\triangleq(\mathbf{s}^{*},\boldsymbol{\mu}^*) $, where $\mathbf{s}^{*}\triangleq(\mathbf{w}^*,\mathbf{p}^*,\boldsymbol{\alpha}^*|$ $\alpha_{km}^*=0\;\&\; \mathbf{w}_{km}^*\neq\mathbf{0})$. Let  $\mathbf{\tilde{s}}\triangleq(\mathbf{\tilde{w}},\mathbf{p}^*,\boldsymbol{\alpha}^*|\alpha_{km}^*=0\;\&\; \mathbf{\tilde{w}}_{km}=\mathbf{0})$, yielding $\mathbf{\tilde{u}}=(\mathbf{\tilde{s}},\boldsymbol{\mu}^*) $. We are now in a position to prove $ F(\mathbf{u}^*)\leq F(\mathbf{\tilde{u}}) $ as $ \mathbf{\tilde{u}} $ is an optimal point.
		Inspired from \cite{Kim:IEEEAccess:2019}, it can be realized that the numerator and denominator of $ \gamma_{k}^{\dl}(\mathbf{s}^*) $ in \eqref{eq: DL SINR - general}  remain unchanged for any $ \mathbf{w}_{km}^* $. From \eqref{eq: F-SE}, $ F_{\mathtt{SE}}(\mathbf{s}^*) $ is the same with respect to $\{\mathbf{w}_{km}^*\}$. Moreover, $\alpha_{km}$ coupled with $\mathbf{w}_{km}$ in \eqref{eq: UL SINR - general} gives $\gamma_{\ell}^{\ul}(\mathbf{s}^*)=\gamma_{\ell}^{\ul}(\mathbf{\tilde{s}}) $, and thus, $ F_{\mathtt{SE}}(\mathbf{s}^*)=F_{\mathtt{SE}}(\mathbf{\tilde{s}}) $. On the other hand, if $ \|\mathbf{w}_{km}^*\|^2>0 $, then $ P_{\mathtt{D}}(\mathbf{u}^*)\geq  P_{\mathtt{D}}(\mathbf{\tilde{u}}) $ given by the first term in \eqref{eq: power consump. for data}. The equality between $ P_{\mathtt{D}}(\mathbf{u}^*) $ and $ P_{\mathtt{D}}(\mathbf{\tilde{u}}) $ holds if $ \mu_m^*=0 $. As a result, $ P_{\mathtt{T}}(\mathbf{u}^*)\geq  P_{\mathtt{T}}(\mathbf{\tilde{u}}) $ leads to $ F_{\mathtt{EE}}(\mathbf{u}^*)\leq F_{\mathtt{EE}}(\mathbf{\tilde{u}}) $ as well as $ F(\mathbf{u}^*)\leq F(\mathbf{\tilde{u}}) $. Moreover, the indexes $ k $ and $ m $ in $\mathbf{s}^*$ are  arbitrary in the sets  $\mathcal{K}$ and $\mathcal{M}$, respectively. It is concluded that if $\alpha_{km}^*=0$, $\mathbf{u}^*$ admits $ \mathbf{\tilde{w}}_{km}=\mathbf{0} $ to generate an optimal solution, which completes the proof.

%\vspace{-5pt}
\section{Proof of Theorem \ref{thr: theorem 1}} \label{app: theorem 1}
\renewcommand{\theequation}{\ref{app: theorem 1}.\arabic{equation}}\setcounter{equation}{0}
Let us define $\mathbf{u}=(\mathbf{s},\boldsymbol{\mu})\in\mathcal{U}\triangleq\bigl\{(\mathbf{w}, \mathbf{p},\boldsymbol{\alpha},\boldsymbol{\mu})|\mathbf{w}_{km}\in\ker(\mathbf{h}_{km}^\dl)\bigr\}\subseteq \mathcal{F}$, $\forall k\in\mathcal{K}$, $m\in\mathcal{M}$, where  $ \mathbf{s} $ represents the triple $ (\mathbf{w}, \mathbf{p},\boldsymbol{\alpha}) $ as part of quadruple $ \mathbf{u} $. To prove Theorem \ref{thr: theorem 1}, we need to show two states: $(i)$ $ F(\mathbf{u}|\mathbf{w}_{km}=\mathbf{0})\geq F(\mathbf{u}|\mathbf{w}_{km}\neq\mathbf{0}) $,  $ \forall \alpha_{km}\in\{0,1\} $; and $(ii)$ $ F(\mathbf{u}|\mathbf{w}_{km}=\mathbf{0}\;\&\;\alpha_{km}=0)\geq F(\mathbf{u}|\mathbf{w}_{km}\neq\mathbf{0}\;\&\;\alpha_{km}=1) $. 	For the first state, we  denote $\mathbf{u}_{0}\in\bigl\{(\mathbf{s}_0, \boldsymbol{\mu})\in\mathcal{U}|\mathbf{w}_{km}= \mathbf{0}\bigr\}$, and $\mathbf{u}_{1}\in\bigl\{(\mathbf{s}_1, \boldsymbol{\mu})\in\mathcal{U}|\mathbf{w}_{km}\neq \mathbf{0}\bigr\}$ and consider  DL SINRs in \eqref{eq: DL SINR - general} with respect to $ \mathbf{s}_0 $ and $ \mathbf{s}_1 $. For $ \mathbf{w}_{km}\in\ker(\mathbf{h}_{km}^\dl) $, it follows that $ |\mathbf{h}_{km}^\dl\mathbf{w}_{km}|^2=0 $, and thus, $\gamma_{k}^{\dl}(\mathbf{s}_0)=\gamma_{k}^{\dl}(\mathbf{s}_1)$. On the other hand, $\gamma_{k'}^{\dl}(\mathbf{s}_0)\geq\gamma_{k'}^{\dl}(\mathbf{s}_1),\;\forall k'\in\mathcal{K}\setminus\{k\}$ implies that the component $ |\mathbf{h}_{k'm}^\dl\mathbf{w}_{km}|^2 $ in the denominator of SINR  for $ \DLUi{k'} $ is equal to or greater than zero, where the equality holds if $\mathbf{w}_{km}=\mathbf{0}$. In addition,  we have $ F_{\mathtt{SE}}(\mathbf{s}_0)\geq F_{\mathtt{SE}}(\mathbf{s}_1) $ due to $\gamma_{\ell}^{\ul}(\mathbf{s}_0)\geq\gamma_{\ell}^{\ul}(\mathbf{s}_1), \forall \ell\in\mathcal{L}$. Meanwhile, it is true that $\|\mathbf{w}_{km}\|^2>0$ for any $\mathbf{w}_{km}\neq\mathbf{0}$, yielding $ P_{\mathtt{D}}(\mathbf{u}_{0})\leq  P_{\mathtt{D}}(\mathbf{u}_{1}) $ and $ P_{\mathtt{T}}(\mathbf{u}_{0})\leq  P_{\mathtt{T}}(\mathbf{u}_{1}) $. That is to say $ F_{\mathtt{EE}}(\mathbf{u}_0)\geq F_{\mathtt{EE}}(\mathbf{u}_1) $ as well as $ F(\mathbf{u}_0)\geq F(\mathbf{u}_1) $, concluding the first state. The second state is easily proved by following the same steps in Appendix \ref{app: relationship alpha and beamforming vector}.

\section{Proof of Theorem \ref{thr: theorem 2}} \label{app: theorem 2}
\renewcommand{\theequation}{\ref{app: theorem 2}.\arabic{equation}}\setcounter{equation}{0}
The power consumption for data transmission and baseband processing $ P_{\mathtt{D}}(\mathcal{V},f_{\mathtt{spr}}(\mathbf{w}),\boldsymbol{\mu}) $ is rewritten as
\begingroup\allowdisplaybreaks\begin{IEEEeqnarray}{rCl} \label{eq: power consump. for data without alpha}
P_{\mathtt{D}}(\mathcal{V},f_{\mathtt{spr}}(\mathbf{w}),\boldsymbol{\mu})   = \sum\nolimits_{m\in\mathcal{M}}\mu_{m}\Biggl(\sum\nolimits_{k\in\mathcal{K}}  \Bigl(\frac{\|\mathbf{w}_{km}\|^2}{\nu_{m}^{\mathtt{AP}} }+ LP_{m}^{\ul} %\nonumber\\
+ r_{\mathtt{sp}}\bigl(\mathbf{w}_{km},\mathbf{h}_{km}^{\dl}|\mathbf{w}_{k}^{(\kappa)},\mathbf{h}_{k}^{\dl}\bigr)P_{km}^{\dl}\Bigr) \Biggr) \nonumber \\
+B\cdot F_{\mathtt{SE}}(\mathbf{w}, \mathbf{p},\boldsymbol{\alpha})\cdot P^{\mathtt{bh}}+ \sum\nolimits_{\ell\in\mathcal{L}}\frac{p_{\ell}}{\nu_{\ell}^{\mathtt{u}}}.\qquad
\end{IEEEeqnarray}\endgroup
It can be foreseen that if $ \mu_m=0$, the signal power of all UEs served by $ \AP $ becomes zero. In other words, $ \mu_{m} $ is also coupled with $ \mathbf{w}_{km},\forall k\in\mathcal{K} $, and $ \mathbf{a}_{m\ell},\forall \ell\in\mathcal{L} $. The first term in \eqref{eq: power consump. for data without alpha} is associated with the DL and UL power allocation, and thus, $ \mu_{m} $ must be strictly updated  with respect to the DL and UL signal power, showing \eqref{eq: compute mu}.

\section{Proof of Theorem \ref{thr: theorem 3}} \label{app: theorem 3}
\renewcommand{\theequation}{\ref{app: theorem 3}.\arabic{equation}}\setcounter{equation}{0}
The proof is done by showing the fact that problems \eqref{eq: prob. bi-obj. - ZF} and \eqref{eq: prob. bi-obj. - equiZF} share the same  optimal objective and solution set.  From the introduction of soft SINRs $\boldsymbol{\lambda}$, it is straightforward to prove that constraints \eqref{eq: prob. bi-obj. - equiZF :: c} and \eqref{eq: prob. bi-obj. - equiZF :: d} must be active (i.e., holding with equalities) at optimum. As a result, constraints \eqref{eq: prob. bi-obj. - ZF :: c} and \eqref{eq: prob. bi-obj. - ZF :: d} can be converted to linear constraints \eqref{eq: prob. bi-obj. - equiZF :: e} and \eqref{eq: prob. bi-obj. - equiZF :: f}, respectively. In addition,  we can decompose $\bar{F}_{\mathtt{SE}}\bigl(\boldsymbol{\Gamma}_{\dl}^{\mathtt{ZF}},\boldsymbol{\Gamma}_{\ul}^{\mathtt{ZF}}\bigr)$ as $\bar{F}_{\mathtt{SE}}\bigl(\boldsymbol{\Gamma}_{\dl}^{\mathtt{ZF}},\boldsymbol{\Gamma}_{\ul}^{\mathtt{ZF}}\bigr) =  R_{\Sigma}(\boldsymbol{\Gamma}_{\dl}^{\mathtt{ZF}}) + R_{\Sigma}(\boldsymbol{\Gamma}_{\ul}^{\mathtt{ZF}})$, where
\begingroup\allowdisplaybreaks\begin{subequations} \label{eq: log det DL UL}
	\begin{align}
	R_{\Sigma}(\boldsymbol{\Gamma}_{\dl}^{\mathtt{ZF}}) & \geq R_{\Sigma}(\boldsymbol{\lambda}_{\dl})  = \sum\nolimits_{k\in\mathcal{K}} \ln (1+\lambda_{k}^{\dl}) %\nonumber\\ &
	=\ln \Bigl( \prod\nolimits_{k\in\mathcal{K}} (1+\lambda_{k}^{\dl}) \Bigr) = \ln  |\mathbf{I} + \boldsymbol{\Lambda}_{\dl} |, \\
	R_{\Sigma}(\boldsymbol{\Gamma}_{\ul}^{\mathtt{ZF}}) & \geq R_{\Sigma}(\boldsymbol{\lambda}_{\ul})  = \sum\nolimits_{\ell\in\mathcal{L}} \ln (1+\lambda_{\ell}^{\ul}) %\nonumber\\&
	=\ln \Bigl( \prod\nolimits_{\ell\in\mathcal{L}} (1+\lambda_{\ell}^{\ul}) \Bigr) = \ln  |\mathbf{I} + \boldsymbol{\Lambda}_{\ul} |.
	\end{align}
\end{subequations}\endgroup
 Then, we have
\begin{align} \label{eq: F-SE log det}
\bar{F}_{\mathtt{SE}}\bigl(\boldsymbol{\Gamma}_{\dl}^{\mathtt{ZF}},\boldsymbol{\Gamma}_{\ul}^{\mathtt{ZF}}\bigr) &\geq \ln  |\mathbf{I} + \boldsymbol{\Lambda}_{\dl} |+\ln  |\mathbf{I} + \boldsymbol{\Lambda}_{\ul} | %\nonumber\\&
:= \tilde{F}_{\mathtt{SE}}\bigl(\boldsymbol{\Lambda}_{\dl},\boldsymbol{\Lambda}_{\ul}\bigr). 
\end{align}
Equalities in \eqref{eq: prob. bi-obj. - equiZF :: c} and \eqref{eq: prob. bi-obj. - equiZF :: d} also lead to an equality of \eqref{eq: F-SE log det}. In the same spirit, we can further show that constraint \eqref{eq: prob. bi-obj. - equiZF :: b} is also active at optimum, which completes the proof.

\section{Proof of Procedure \ref{procedure1}} \label{app: procedure1}
\renewcommand{\theequation}{\ref{app: procedure1}.\arabic{equation}}\setcounter{equation}{0}
Suppose that a projection matrix $ \mathbf{P}=\mathbf{I} - (\mathbf{\tilde{G}}^{\AtoA})^H\bigl(\mathbf{\tilde{G}}^{\AtoA}$ $(\mathbf{\tilde{G}}^{\AtoA})^H\bigr)^{-1}\mathbf{\tilde{G}}^{\AtoA} $ is applied to cancel the IAI and RSI. Accordingly, the effects of MUI, IAI and RSI can be ignored by considering $ \mathbf{\bar{W}} = \mathbf{P}\mathbf{Q}^H\mathbf{\tilde{T}}(\mathbf{D}^{\dl})^{\frac{1}{2}} $ as  a precoder matrix.
	\begin{itemize}
		\item \textbf{MUI cancellation}: it follows that
		\begin{align}
		\mathbf{H}^{\dl}\mathbf{\bar{W}} 
		 \overset{[a]}{=} \mathbf{T}\mathbf{Q}\mathbf{Q}^H\mathbf{\tilde{T}}(\mathbf{D}^{\dl})^{\frac{1}{2}}
		 \overset{[b]}{=} \mathbf{\bar{T}}(\mathbf{D}^{\dl})^{\frac{1}{2}}, \nonumber
		\end{align}
		where $ [a] $ is obtained via Step 4, while $ [b] $ comes from the structure of $ \mathbf{\tilde{T}} $ in Step 5. Clearly, MUI is completely removed, since both $ \mathbf{\bar{T}} $ and $ \mathbf{D}^{\dl} $ are  diagonal matrices. 
		\item \textbf{IAI and RSI cancellation}: the effective IAI and RSI are generally expressed by $\mathbf{\tilde{G}}^{\AtoA}$. We have
		\begin{align}
		\mathbf{\tilde{G}}^{\AtoA}\mathbf{\bar{W}} = \mathbf{\tilde{G}}^{\AtoA}\mathbf{P} \mathbf{Q}^H\mathbf{\tilde{T}}(\mathbf{D}^{\dl})^{\frac{1}{2}} = \mathbf{0}. \nonumber
		\end{align}
		due to $\mathbf{\tilde{G}}^{\AtoA}\mathbf{P} = \mathbf{0}$. However, $ \mathbf{\tilde{G}}^{\AtoA} $ is a concatenation of IAI and RSI matrices, leading to a full-rank matrix in most cases. In other words, $\mathbf{P}$ must be forced to $\mathbf{0} $, and thus,  should not be joined into the precoder matrix $\mathbf{\bar{W}}$. To overcome this issue, we exploit the PCA-based method to depress the IAI and RSI in the rest of this proof.
	\end{itemize}
To derive matrix $ \mathbf{P} $, we consider a low-rank approximation via the PCA method as in Steps 1-3. From \eqref{eq: SVD}  and \eqref{eq: N_bar top eigenvalues}, the low-rank approximation of $ \mathbf{\tilde{G}}^{\AtoA} $ can be derived from the $ \bar{N} $-top eigenvalues as
	\begin{align}
	(\mathbf{\tilde{G}}_{\bar{N}}^{\AtoA})^H\mathbf{\tilde{G}}_{\bar{N}}^{\AtoA} = \mathbf{\bar{U}}\mathbf{\bar{E}}^{\frac{1}{2}}\mathbf{\bar{E}}^{\frac{1}{2}}\mathbf{\bar{U}}^H,
	\end{align}
	where  $ \mathbf{\bar{U}}\in\mathbb{C}^{N\times \bar{N}} $ involves the first $ \bar{N} $ columns of $ \mathbf{U} $, and the diagonal matrix $ \mathbf{\bar{E}}\in\mathbb{C}^{\bar{N}\times\bar{N}} $ has the main diagonal with the $ \bar{N} $-top eigenvalues of $ \mathbf{E} $ in \eqref{eq: SVD}. We note that $ \mathbf{\bar{U}}^H\mathbf{\bar{U}}=\mathbf{I} $, but $ \mathbf{\bar{U}}\mathbf{\bar{U}}^H\neq\mathbf{I} $. By treating $ \mathbf{\tilde{G}}^{\AtoA} $ as $ \mathbf{\tilde{G}}_{\bar{N}}^{\AtoA} = \mathbf{\bar{E}}^{\frac{1}{2}}\mathbf{\bar{U}}^H $, the projection matrix $ \mathbf{P} $ can be calculated as
	\begin{align}
	\mathbf{P} & =\mathbf{I} - (\mathbf{\tilde{G}}_{\bar{N}}^{\AtoA})^H\bigl(\mathbf{\tilde{G}}_{\bar{N}}^{\AtoA}(\mathbf{\tilde{G}}_{\bar{N}}^{\AtoA})^H\bigr)^{-1}\mathbf{\tilde{G}}_{\bar{N}}^{\AtoA} %\nonumber \\&
	 = \mathbf{I} - \mathbf{\bar{U}}\mathbf{\bar{E}}^{\frac{1}{2}}(\mathbf{\bar{E}}^{\frac{1}{2}}\mathbf{\bar{U}}^H\mathbf{\bar{U}}\mathbf{\bar{E}}^{\frac{1}{2}})^{-1}\mathbf{\bar{E}}^{\frac{1}{2}}\mathbf{\bar{U}}^H 
	 = \mathbf{I} - \mathbf{\bar{U}}\mathbf{\bar{U}}^H,
	\end{align}
	showing  Step 3.

\section{Proof of Theorem \ref{thm: MSE prob.}} \label{app: MSE prob.}
\renewcommand{\theequation}{\ref{app: MSE prob.}.\arabic{equation}}\setcounter{equation}{0}
 From $\boldsymbol{\bar{\Xi}} = \boldsymbol{\Xi}\boldsymbol{\Upsilon}$, it is clear that $ \boldsymbol{\bar{\varphi}}_{j}=\boldsymbol{\Xi}\boldsymbol{\upsilon}_j $ with $ \boldsymbol{\upsilon}_j $ being the $ j $-th column of $ \boldsymbol{\Upsilon} $. Therefore, $ \varepsilon_{mj} $ can be expressed as a function of assignment variables, i.e.,
	\begingroup\allowdisplaybreaks{\small\begin{align} \label{eq: MSE with assign. var.}
	\frac{N_{m}\varepsilon_{mj}}{\beta_{mj}} &= N_{m} \Bigl(1 - \frac{\tau p_{j}^{\mathtt{tr}} \beta_{mj}}{\sum_{j'\in\mathcal{T}_{\ul}} \tau p_{j'}^{\mathtt{tr}} \beta_{mj'} |\boldsymbol{\upsilon}_{j}^H\boldsymbol{\Xi}^H\boldsymbol{\Xi}\boldsymbol{\upsilon}_{j'}|^2 + \sigma^2_{\mathtt{AP}}} \Bigr) %\nonumber \\& 
	= N_{m}\Bigl(1 - \frac{\tau p_{j}^{\mathtt{tr}} \beta_{mj}}{\sum_{j'\in\mathcal{T}_{\ul}} \tau p_{j'}^{\mathtt{tr}} \beta_{mj'} |\boldsymbol{\upsilon}_{j}^H\boldsymbol{\upsilon}_{j'}|^2 + \sigma^2_{\mathtt{AP}}} \Bigr) \nonumber \\
	& \overset{[a]}{=} N_{m}\frac{\sum_{j'\in\mathcal{T}_{\ul}\backslash \{j\} } \tau p_{j'}^{\mathtt{tr}} \beta_{mj'} \boldsymbol{\upsilon}_{j}^H\boldsymbol{\upsilon}_{j'} + \sigma^2_{\mathtt{AP}}}{\tau p_{j}^{\mathtt{tr}} \beta_{mj}+\sum_{j'\in\mathcal{T}_{\ul}\backslash \{j\}} \tau p_{j'}^{\mathtt{tr}} \beta_{mj'} \boldsymbol{\upsilon}_{j}^H\boldsymbol{\upsilon}_{j'} + \sigma^2_{\mathtt{AP}}},\quad
	\end{align}}\endgroup
	where $ [a] $ comes from the fact that $ \boldsymbol{\upsilon}_{j}^H\boldsymbol{\upsilon}_{j'}\in\{0,1\},\;\forall j, j' \in\mathcal{T}_{\ul} $, and  $ \boldsymbol{\upsilon}_{j}^H\boldsymbol{\upsilon}_{j'}=1 $ when $ j=j' $. In addition, for arbitrary values of $ x, y \text{ and } z $, such that $ 0\leq z<x<y $, it is true that
$\frac{x-z}{y-z}\leq\frac{x}{y}\leq\frac{x+z}{y+z}.$
	Upon setting
 $ x=\sum_{j'\in\mathcal{T}_{\ul}\backslash \{j\} } \tau p_{j'}^{\mathtt{tr}} \beta_{mj'} \boldsymbol{\upsilon}_{j}^H\boldsymbol{\upsilon}_{j'} + \sigma^2_{\mathtt{AP}} $ and $ y=x + \tau p_{j}^{\mathtt{tr}} \beta_{mj} $,  $ z $ is the amount of disparity in $ x $ when $ \boldsymbol{\upsilon}_{j} $ changes.
	It implies that when $ \boldsymbol{\upsilon}_{j} $ changes, $ x $ and $ y $ vary by the same amount of $ z $. Consequently,  we replace $x/y$ with $ y $ for the ease of solution derivation, and thus the objective \eqref{eq: prob. MSE :: a} can be rewritten as
	\begingroup\allowdisplaybreaks\begin{align}\label{eq:F.2}
	&\underset{j\in\mathcal{T}_{\ul}}{\max} \sum\nolimits_{m\in\mathcal{M}} N_m  \sum\nolimits_{j'\in\mathcal{T}_{\ul} } \tau p_{j'}^{\mathtt{tr}} \beta_{mj'} \boldsymbol{\upsilon}_{j}^H\boldsymbol{\upsilon}_{j'} + \sigma^2_{\mathtt{AP}} %\nonumber\\ &
	= \underset{j\in\mathcal{T}_{\ul}}{\max} \sum\nolimits_{j'\in\mathcal{T}_{\ul} }\boldsymbol{\upsilon}_{j}^H\boldsymbol{\upsilon}_{j'}\sum\nolimits_{m\in\mathcal{M}} N_m  \tau p_{j'}^{\mathtt{tr}} \beta_{mj'}  + \sigma^2_{\mathtt{AP}}. 
	\end{align}\endgroup
{\hili Since $ \sigma^2_{\mathtt{AP}} $ in \eqref{eq:F.2} is the constant, we can arrive at a tractable optimization problem \eqref{eq: prob. MSE quad.}.}

\begingroup
\setstretch{1.0}
\bibliographystyle{IEEEtran}
\bibliography{IEEEfull}
\endgroup

\end{document}